\documentclass{article}

\usepackage[main, preprint]{neurips_2026}

\usepackage[utf8]{inputenc} 
\usepackage[T1]{fontenc}    
\usepackage{hyperref}       
\usepackage{url}            
\usepackage{booktabs}       
\usepackage{amsfonts}       
\usepackage{nicefrac}       
\usepackage{microtype}      
\usepackage{xcolor}         
\usepackage{colortbl}       
\usepackage{multirow}

\usepackage{graphicx}
\usepackage{algorithm}
\usepackage[noend]{algpseudocode}
\usepackage{amsmath}
\usepackage{amssymb}
\usepackage{amsthm}
\usepackage{subcaption}
\usepackage{wrapfig}
\usepackage{enumitem}
\usepackage{xspace}
\usepackage[normalem]{ulem}

\newtheorem{lemma}{Lemma}

    \newcommand{\bfx}{\boldsymbol{x}}

    \newcommand{\eps}{\varepsilon}
    \newcommand{\R}{\mathbb{R}}
    
    \newcommand{\Iab}{I_{a,b}}
    \newcommand{\F}{\mathcal{F}}
    \newcommand{\I}{\mathcal{I}}
    \newcommand{\Iw}{\mathcal{I}_w}
    
    \newcommand{\all}{\textsf{all}}
    \newcommand{\inn}{\textsf{in}}
    \newcommand{\out}{\textsf{out}}
    \newcommand{\KRR}{\textsc{krr}}
    \newcommand{\KR}{\textsc{nwkr}}

    \newcommand{\SRA}{\mathsf{SRA}}
    \newcommand{\SRI}{\mathsf{SRI}}
    \newcommand{\SRO}{\mathsf{SRO}}
    \newcommand{\neigh}{\mathcal{N}}

\title{Efficient Regression Models for Scan Statistics}

\author{%
  Gazi Abdur Rakib \\
  University of Utah \& CosmicAI \\
  \texttt{gaziabdur.rakib@utah.edu} \\
  \And
  Tristan Ashton \\
  National Radio Astronomy Observatory \& CosmicAI \\
  \texttt{tashton@nrao.edu} \\
  \And
  Ryan A. Loomis \\
  National Radio Astronomy Observatory \& CosmicAI \\
  \texttt{rloomis@nrao.edu} \\
  \And
  Brian S. Mason \\
  National Radio Astronomy Observatory \& CosmicAI \\
  \texttt{bmason@nrao.edu} \\
  \And
  Eric J. Murphy \\
  National Radio Astronomy Observatory \& CosmicAI \\
  \texttt{emurphy@nrao.edu} \\
  \And
  Ci Xue \\
  National Radio Astronomy Observatory \& CosmicAI \\
  \texttt{cxue@nrao.edu} \\
  \And
  Jeff M. Phillips \\
  University of Utah \& CosmicAI \\
  \texttt{jeffp@cs.utah.edu} \\
}

\begin{document}

\maketitle

\begin{abstract}
We introduce a new class of regression models for scan statistics on real-valued signals.  These allow for improved fitting of non-stationary signals to contrast with the interval anomalies identified by the scan statistics.  Our models can represent generalized likelihood ratio statistics.  While these methods naively require $O(n^4)$ for a length $n$ signal, we provide algorithmic improvements which lead to linear time algorithms (with assumptions on max interval width).  Our methods, especially ones based on Nadaraya-Watson kernel regression, are demonstrated as especially effective in detecting both synthetically planted anomalies, and for identifying a real ``platforming'' issue in interferometric astronomy.  
\end{abstract}

\clearpage

\section{Introduction}
\label{sec:intro}

For a signal $\bfx = \langle x[1], x[2], \ldots, x[n] \rangle$, measured at regular intervals $[n] = \{1,2, \ldots, n\}$ we tackle the challenge of finding an anomalous interval.  This task occurs in widespread applications in modeling such as for time series (e.g., stocks~\citep{braun2018impact} or weather data~\citep{horel2002mesowest}), interferometric sensing (e.g., from radio telescopes~\cite{escoffier2007alma} or satelites~\citep{suto2013characterization}), and genomics (e.g., for copy-number~\citep{zack2013pan}).  Such anomalous intervals often manifest from instrumental error; that is, these anomalous interval subsets are interruptions in the proper collection of the raw measurements. While some cases can be easily spotted by eye, these data are often directly aggregated and passed to complex downstream analyses, so it is critical to detect these anomalies early and robustly so that they can be dealt with before they corrupt or degrade the complex processes that follow.

Variants of this problem have been studied under the framework of \emph{scan statistics}~\cite{abolhassani2021up,glaz2024handbook}   which "scan" every possible interval $\Iab = \{a,a+1,\ldots, b\}$ of the signal, scoring each one, and returning the interval and score of the most extreme one.  This most extreme score is the scan statistic.  These models traditionally searched for dense clusters of data, often in irregularly-spaced sequences.  Regardless, these approaches fit a probabilistic model to the full signal and observe how much better the fit can be if separate models are fit inside and outside the interval.  The most relevant formalization of this setting has that each observation $x[i]$ is in $\R$ and is drawn independently from a normal distribution $\mathcal{N}(\mu,\sigma)$ with unknown mean $\mu$.  From here a closed form log-likelihood ratio score (c.f., \cite{huang2007spatial,agarwal2006spatial,kulldorff2009normal}) can be devised, which we review (and extend) in Section \ref{sec:scan}.  

The interval anomaly -- an extreme $\Iab$ in the set of all such intervals $\I$ -- is an essential modeling component for the situations we care about, since it corresponds to a rare, and temporary instrumental issue in an otherwise useful signal.  However, a more widely studied setting called \emph{change point detection} has otherwise similar modeling.  It also considers regular 1-dimensional signals, and derives statistical scores to determine when at certain indexes $i \in [n]$ the fit of an interval $\langle a, \ldots, i\rangle$ of the signal likely changes to a new fit in $\langle i+1, \ldots, b\rangle$.  This broader area~\cite{aminikhanghahi2017survey,niu2016multiple} has more developed methods and richer models for the underlying patterns of the signals, and these will resemble the ones we derive and evaluate~\cite{siegmund1995using,yang2020change,yu2022localising,harchaoui2008kernel}; see Appendix \ref{sec:related_work} for more in depth review.  One could imagine using it to detect each boundary of the anomalous interval.  However, this change point setting is structurally different from the scan statistic: it does not enforce that the regions before and after the anomalies adhere to the same model.  

In this framing, we make the \textbf{following contributions}:  
\begin{enumerate}

    \vspace{-3mm}
    \item We formulate new closed-form models of scan statistics that still assume normal noise, but rather than restricting to a constant value, the base model can now fit a regression model -- either polynomial or kernel-based.  Under the assumed generative and noise models, the proposed polynomial regression anomaly score reduces to a log-likelihood ratio, so ranking and filtering by this score inherits the classical Neyman–Pearson optimality property for the associated detection problem. 
    \item We show that the Nadaraya-Watson kernel regression (NWKR) model has favorable modeling properties, requiring only controlled local dependence between consecutive signal values $x[i]$ and $x[i+1]$, as governed by the kernel bandwidth.  The resulting anomaly score is a natural analog to a  generalized likelihood ratio statistic for detecting local departures in a fully nonparametric setting.
    \item We devise and analyze efficient algorithms for most of these models when we restrict intervals considered to be of width at most $w$.  For $d$ degree polynomial regression models the scan statistic can be computed in $O(nwd^3)$ time.  For NWKR we can compute it in $O(nwr)$ time where $r$ is the width of a truncated kernel.  This provides orders of magnitude improvement over direct implementations of these statistics which would take $O(n^4)$ time.  
    \item  Finally, we evaluate these methods on synthetic and real world data sets demonstrating the efficiency of algorithms and effectiveness of the new models, especially the NWKR variant.  A key motivating example is finding ``platforming" anomalies~\citep{alma_pipeline_team_2025_users_guide} that occur in bandpass data from radio telescopes~\citep{escoffier2007alma}.  
\end{enumerate}





\section{Basics of Scan Statistic Model via Regression Families}
\label{sec:scan}
In this section we review various function families from which we fit models.  
Our starting point is the Gaussian Scan Statistic (as modeled by \cite{huang2007spatial}; \cite{agarwal2006spatial}), in which each measurement is the sum of a latent signal $\mu$ and additive Gaussian noise: for every index $i$,
\[
x[i] = \mu + \eps_i,\qquad \eps_i \stackrel{\text{iid}}{\sim} \mathcal{N}(0,\sigma^2),
\]
with an unknown but fixed variance $\sigma^2$. Let $\F$ be the family of functions mapping indices to predicted values, 
$f:\{1,\ldots,n\}\rightarrow\mathbb{R}$. Classical Gaussian scan statistics take $\F = \F_0$ to be constant functions, $f(i)\equiv \hat \mu$; we will later allow richer families.  

The premise of Kulldorff's perspective on scan statistics \citep{kulldorff1999spatial} is that of a hypothesis testing framework. It defines a baseline null distribution $H_0$ that posits no anomalous region; a single $f_\all \in \F$ is fit on all indices $[n]=\{1,\ldots,n\}$. For each candidate \(\Iab \in \I\), the alternative \(H_1(\Iab)\) fits \(f_\inn\in\F\) on \(\Iab\) and $f_\out\in\F$ on $[n]\setminus \Iab$. If there is an interval $\Iab$ which deviates significantly from what one would expect from the outside data, then the joint fit from the pair $f_\inn, f_\out$ should be much better than the global fit from $f_\all$.
Given the Normal noise model, the likelihood (up to constant normalizing factors) of the data $\bfx$ under a function $f \in \F$ is
$
 L(\bfx \mid f) \;=\; \prod_{i=1}^n \exp\!\big( - \frac{(x[i] - f(i) )^2}{\sigma^2}\big).
$
The profile likelihood under $H_0$ is
\[
 L_0(\bfx)
 \;=\;
 \max_{f \in \F} L(\bfx \mid f),
 \quad\quad
\ln L_0(\bfx) = - \frac{1}{\sigma^2} \min_{f \in \F} \sum_{i=1}^n (x[i] - f(i))^2.
\]
Similarly, for a chosen interval \(\Iab\) the alternative hypothesis \(H_1(\Iab)\) uses separate functions on the inside and outside:
\[
  \ln L_1(\bfx; \Iab)
  = - \frac{1}{\sigma^2}
  \left[
    \min_{f_\inn \in\F} \sum_{i=a}^b (x[i] - f_\inn(i))^2
    +
    \min_{f_\out \in\F} \sum_{i\notin \Iab} (x[i] - f_\out(i))^2
  \right].
\]
The log-likelihood ratio is
$      
\mathrm{LLR}(\bfx;\Iab) \;=\; \ln \frac{  L_1(\bfx; \Iab)   }{  L_0(\bfx)  }
$.
The Gaussian normalizing constants and the factor $-1/\sigma^2$ cancel in the log-likelihood ratio, so only the \emph{sums of squared residuals} (SSE) matter. It is convenient to name these SSEs explicitly. We define
\[
  \SRA\!=\!\min_{f\in\F} \! \sum_{i=1}^n (x[i]-f(i))^2; \; 
  \SRI(\Iab)\!=\!\min_{f_\inn\in\F} \! \sum_{i=a}^b (x[i]-f_\inn(i))^2;  \;
  \SRO(\Iab)\!=\!\min_{f_\out\in\F} \! \sum_{i\notin\Iab} (x[i]-f_\out(i))^2.
\]
Under the model above, then the $\mathrm{LLR}(\bfx;\Iab)$ is monotone (up to a positive constant factor) in
$
  \SRA - (\SRI(\Iab) + \SRO(\Iab)).
$
Any monotone transform of this difference yields the same ranking over windows. In practice we use a \emph{normalized} score
\begin{equation}
  S(\Iab)
  \;=\;
  1 - \textstyle{\frac{\SRI(\Iab) + \SRO(\Iab)}{\SRA}},
  \label{eq:normalized-score}
\end{equation}
which is dimensionless and lies in \([0, 1]\). When the inside/outside split does not improve the fit, $\SRI+\SRO \approx \SRA$ and $S(\Iab)\approx 0$; when the split yields a much better joint fit, $\SRI+\SRO \ll \SRA$ and $S(\Iab)$ tends towards 1. Finally, the scan statistic discrepancy used to decide whether there is an anomaly is
\begin{equation}\label{eq:scan-score}
 \Phi(\bfx) \;=\; \max_{\Iab \in \I} S(\Iab),
\end{equation}
which is equivalent (for fixed $\F$ and $\sigma^2$) to maximizing $\mathrm{LLR}(\bfx;\Iab)$.  We will often restrict to a family of intervals $\Iw$ which are not too wide; formally they ensure for all $\Iab \in \Iw$ that $b-a+1 \leq w$.

\subsection{Function Families}
We instantiate $\F$ with three canonical choices.   We observe in this paper that this choice does not affect the above derivation of the score or scan statistics discrepancy function, other than what class $\F$ is optimized within the null ($H_0$) or alternative ($H_1$) hypothesis.  
\begin{enumerate}\vspace{-2mm}
\item \emph{Constant (mean) model} $\F_0$:  In this classic \citep{huang2007spatial,agarwal2006spatial} setting $f(i)\equiv\mu$, where $\mu$ is a constant mean parameter capturing a stationary background.   And classically, the best fit solution is the mean of the data.  

\item \emph{Fixed-degree polynomials} $\F_d$: This can be viewed as a direct extension of $\F_0$ to allow $f(i)=\sum_{k=0}^{d}\beta_k i^k$.  Note for $d=0$ then $\F_0 = \F_d$.  
For small $d$ we can obtain coefficients $\beta$ by unregularized least squares after a linearization expansion; it fits a global model and allows for gentle drifts across the data. 

\item \emph{Kernel ridge regression (KRR)} $\F_\KRR$:   This typically uses a positive-definite kernel $K(i,j)$, so that
$
  f(i) \;=\; \sum_{j=1}^n \alpha_j K(i,j),
$
with
$
   (K + \lambda I) \boldsymbol{\alpha} = \bfx,
$
where $K_{i,j}=K(i,j)$, $\lambda>0$ is a ridge parameter, and $\boldsymbol{\alpha} \in \R^n$ are the kernel weights.

\item \emph{Nadaraya–Watson kernel regression} $\F_\KR$ with symmetric kernel $K(i,j)$:
\begin{equation}
  f(i)\;=\; \textstyle{\frac{\sum_{j=1}^{n} K(i,j)\,x[j]}{\sum_{j=1}^{n} K(i,j)}}.
  \label{eq:nwkr}
\end{equation}
Unlike other families $\F$, this has no parameters to optimize; it is fully non-parametric.  
\end{enumerate}

For the kernel-based families 
we use two kernels in analysis: 
\emph{Gaussian} $K(i,j)=\exp(-\tfrac{(i-j)^2}{2r^2})$ and 
\emph{Laplace} $K(i,j)=\exp(-\tfrac{|i-j|}{r})$, 
each with a bandwidth parameter $r$.

\paragraph{Statistical justification.}

Neyman-Pearson optimality~\citep{neyman1933efficient,lehmann2022testing} implies if one fixes an allowable false-positive rate, then among all tests at that level, the likelihood-ratio test achieves the largest detection probability against the specified alternative.  Any scoring function that is increasing monotonic with the likelihood ratio inherits this optimality; this includes $S(\Iab)$ under the $\F_d$ model.  Thus it is the statistically most powerful way to score evidence for an anomalous interval $\Iab$ under the assumed polynomial model $\F_d$.  

At the scan level, $\Phi(\bfx) = \max_{\Iab \in \I} S(\Iab)$ is an exact generalized likelihood ratio scan statistic~\citep{wilks1938largesample,lehmann2022testing}.  This statistic gives the strongest evidence, over all candidate anomalous intervals, for the best-fitting split model relative to the best-fitting null model, thereby yielding a principled omnibus test for whether any anomalous region is present.  



For the Nadaraya-Watson family $\F_\KR$, because the NWKR model is non-parametric, we cannot claim it is a maximum likelihood estimate, and the Neyman-Pearson optimality statement does not apply for the corresponding $S(\Iab)$.  However, it is the natural analog to a generalized likelihood ratio test statistic for the Nadaraya-Watson smoother.  Because the kernels we consider have exponential decay (or finite range if we truncate) they represent a weighted average over localized regions, and allow the models $f \in \F_\KR$ (in equation \eqref{eq:nwkr}) to locally adapt signal variation.  Moreover, the procedure is tailored to compact, spatially localized departures---precisely the anomaly regime considered in scan statistics---while being relatively insensitive to distant variation that should be explained by the smooth background model.

\section{Scanning Algorithms}
\label{sec:families}

Our computation proceeds in three passes that echo the logic of scan statistics, as illustrated in Algorithm \ref{alg:scan-alg}. First, we fit a single model $f_\all \in \F$ on the whole index set \([n]\) and cache its sum of squared errors,
$
\SRA
\;=\;
\min_{f\in\F}\ \sum_{i\in[n]} \big(x[i]-f(i)\big)^2,
$
which is the SSE under the null model $H_0$. This global fit is performed once and reused for all candidates. 

Next, we \emph{scan} through all possible intervals $\Iab \in \I_w$; this scanning step is an essential part of the scan statistics framework, but an algorithmic nightmare in that it iterates through all options.  

For a specific $\Iab$, we restrict the modeling domain to the inside and the outside and fit $f_\inn$ on $\Iab$ and $f_\out$ on the remainder $[n] \setminus \Iab$.  These two fits yield $\mathrm{SRI}(\Iab)$ and $\mathrm{SRO}(\Iab)$, and thus a score $S(\Iab)$. Finally, we obtain the maximizer
$
\operatorname*{arg\,max}_{\Iab \in \Iw}
S(\bfx;\Iab),
$
which is our scan statistic's estimate of the most discrepant contiguous region.

\begin{algorithm}
\caption{Regression-Scan($x, \F$)}
\label{alg:scan-alg}
\begin{algorithmic}
    \State  Best-fit $\textsf{Cost}_1 \gets\SRA = \min_{f \in \F} \sum_{i \in [n]} (f(i) - x[i])^2$
    \For{ $\Iab \in \I_w$ }
        \State Fit inside $f_\inn = \arg\min_{f \in \F} \sum_{i \in \Iab} (f(i) - x[i])^2$
        \State Fit outside $f_\out = \arg\min_{f \in \F} \sum_{i \in [n] \setminus \Iab} (f(i) - x[i])^2$
        \State Score: $S_{a,b} = \sum_{i \in \Iab} (f_\inn(i) - x[i])^2 + \sum_{i \in [n] \setminus \Iab} (f_\out(i) - x[i])^2$
        \State \textbf{if} ($S_{a,b} < \textsf{Cost}_1$) \textbf{then} $\textsf{Cost}_1 \gets S_{a,b}$
    \EndFor
    \State \textbf{return} $1 - \frac{\textsf{Cost}_1}{\SRA}$
\end{algorithmic}
\end{algorithm}

This algorithm is inherently at least cubic $\Omega(n^3)$ in runtime if followed directly.  There may be $\Omega(n^2)$ intervals $\Iab \in \Iw$ to consider, and within that double loop, just computing the score $S_{a,b}$ for each $\Iab$ sums over $n$ terms.  
In the coming section we will see that this can be even worse for kernel methods where computing $f$ may require $\Omega(n)$ time.  However, there is also a lot of repeated calculations that can be reduced with carefully precomputing and caching partial results.




\subsection{Basic Algorithmic Analysis}
\label{sec:algo-anal}


As a warm up, we consider $f_\mu \in \F_0$ where $f_\mu(j) = \mu$ for all $j$.  The maximum likelihood estimator (MLE) $f^*$ is determined by the choice of $\mu = \frac{1}{n}\sum_{i=1}^n x[i]$.  It can be computed in $O(n)$ time and evaluated in $O(1)$ time.  By updating the domain by increasing it or decreasing it by size $1$ takes $O(1)$ time as well by separately maintaining $\sum_{i=1}^n x[i]$ and $n$, each of which are easy to update and then recombine in $O(1)$ time.  
These algorithmic insights are implicit in \cite{agarwal2006spatial}.  

\textbf{Polynomial regression. }
We next consider the function family $\F_d$ where each $f \in \F_d$ has the form
$
f(i) = \sum_{k=0}^d \alpha_k i^k.
$
Thus each function in this family is parameterized by $\alpha \in \R^{d+1}$.  
The MLE (least squares) model can classically be solved by a standard linear expansion, and representing each $i$ as a $(d+1)$-dimensional vector $v_i = (1, i, i^2, \ldots, i^d) \in \R^{d+1}$.  Then we can solve multi-linear regression for the optimal $\alpha^* \in \R^{d+1}$ by stacking these vectors into the Vandermonde matrix $V \in \R^{n \times {d+1}}$ where the $i$th row is $v_i$, and computing $\alpha^* = (V^T V)^{-1} V^T x$.  The inverse is well-defined when $d+1 \leq n$, and at least $d+1$ rows $v_i$ are linearly independent (which should be true if they are observed with independent noise).  

For the runtime, observe that $V^T V$ is $d \times d$, so the inverse operation takes $O(d^3)$ time, which is not a bottleneck for the common case where $d$ is a small constant like $2,3,4$.   While it takes $O(nd^2)$ time to compute $V^T V$, it can be written as $V^T V = \sum_{i=1}^n v_i^T v_i$ and so can be updated (like $V$ itself) in $O(d^2)$ time.  Thus, the update step also takes $O(d^3)$ time.  Evaluating $f \in \F_d$ takes $O(d)$ time.

\paragraph{Kernel ridge regression (KRR).}
To provide a stronger nonparametric baseline than polynomials while remaining a standard comparator, we implement kernel ridge regression. KRR fits a function in the reproducing kernel Hilbert space (RKHS) associated with a positive semidefinite kernel $K(\cdot,\cdot)$ by trading off squared error with an $\ell_2$ penalty on the RKHS norm. On a subset $S=\{i_1,\ldots,i_m\}$, KRR solves
$
\widehat{f}_S \;\in\;\arg\min_{f\in\mathcal{H}_K}\ \sum_{r=1}^{m}\bigl(x[i_r]-f(i_r)\bigr)^2 \;+\; \lambda\,\|f\|_{\mathcal{H}_K}^2,
$
where $\lambda>0$ is the regularization parameter. By the representer theorem, $\widehat{f}_S$ has the form
$
\widehat{f}_S(i)\;=\;\sum_{r=1}^{m}\alpha_r\,K(i,i_r),
$
and the coefficients $\alpha\in\mathbb{R}^m$ are obtained by solving the linear system
$
(G+\lambda I)\,\alpha \;=\; x_S$, where the gram matrix $G$ is defined $G_{r\ell}=K(i_r,i_\ell)$.  
Predictions on the subset are $\widehat{x}_S[i_r]=(\,G\alpha\,)_r$ (equivalently, evaluate $\widehat{f}_S$ at the training points), and the subset SSE is
$
\mathrm{SSE}(S)\;=\;\sum_{r=1}^{m}\bigl(x[i_r]-\widehat{x}_S[i_r]\bigr)^2
\;=\;\|x_S-G\alpha\|_2^2.
$
For fairness in scanning, KRR is fit separately on the inside and outside sets for each window (as required by the scan-statistic alternative), using the same $(G+\lambda I)^{-1}$ solve on the corresponding subset.  
This directly takes $O(n^3)$ for the matrix inverse, and update time is also slow at $O(n^3)$ because while we can update $K$ in $O(n)$ time, the inverse is still takes $O(n^3)$ time.  And while faster approximate algorithms exist~\citep{MuscoMusco2017Recursive,AvronEtAl2017RFF} the do not give much benefit at the scales we consider.  
Evaluation also is a slow $O(n)$ time.  


%


\textbf{Nadaraya-Watson Kernel Regression. }
We next provide a non-parametric method
Nadaraya-Watson Kernel Regression~\citep{nadaraya1964estimating,watson1964smooth}.  This leverages a kernel $K : \R \times \R \to \R$, but does not solve for an optimal solution.  Instead it can be viewed as a smoothed moving average.  The model is
$
  f_K(i) = \textstyle{\frac{\sum_{j =1}^n K(i,j) x[j]}{\sum_{j =1}^n K(i,j)}}.
$
Other than the choice of kernel (including its bandwidth parameter $r$), there are not model parameters and no optimization to be solved, it simply enforced a sort of weighted neighborhood over which to take a moving average.  

\textbf{Runtime Summary. }
If we restrict to $\Iab \in \Iw$ then in Algorithm \ref{alg:scan-alg} there are $O(nw)$ model windows to consider.  The initial windows are small, and so the costs per interval is asymptotically dominated by (1) updating the model to increase or decrease one additional point to the $f_\inn$ and one fewer point in the $f_\out$ models, and (2) evaluating the model on all $n$ points to compute the $S_{a,b}$ value.  Thus the total runtime of the Algorithm \ref{alg:scan-alg} is $O(nw) \times (\textsf{update time}  +  n \times \textsf{evaluation time})$.  We summarize this for the four models we consider in Table \ref{tab:models}.  
The next subsection will show how improve some runtimes, generating the last two columns.  

\begin{table}[h]
\centering
\begin{tabular}{lcccccc}
\hline
Model & Solve & Update & Evaluate & Total & Improved Total & Truncated Total\\
\hline
$\F_0$      & $O(n)$      & $O(1)$      & $O(1)$      & $O(n^2w)$     & $O(nw)$  & $O(nw)$    \\
$\F_d$      & $O(nd^2)$   & $O(d^3)$    & $O(d)$      & $O(n^2w d^3)$ & $O(nw d^3)$ & $O(nw d^3)$ \\
$\F_{\KRR}$ & $O(n^3)$    & $O(n^3)$    & $O(n)$      & $O(n^4w)$     & $O(n^4w)$ & $O(n^4 w)$ \\
$\F_{\KR}$  & $O(1)$      & $O(1)$      & $O(n)$      & $O(n^3w)$     & $O(n^2w)$ & $O(nwr)$ \\
\hline
\end{tabular}
\caption{\label{tab:models}Runtime Comparison of models in Algorithm \ref{alg:scan-alg} for a sequence of length $n$, with max interval size $w$, and kernels truncated at radius $O(r)$.  Improved results derived in Section \ref{sec:faster}.}
\end{table}


\subsection{Algorthmic Improvements}
\label{sec:faster}

We can improve the runtime in two significant ways.  First, we can precompute aspects of the evaluation of $S_{a,b}$ so that the update time is improved; this applies to $\F_0, \F_d, \F_\KR$.  Second we can consider truncated kernels $K_r$ so that when the bandwidth of the kernels is $r$, and we only evaluate a points so $|i-j| \leq O(r)$ (typically within $3r$) which only has $O(r)$ points total.  This impacts the runtime for models $\F_\KRR, \F_\KR$.

\textbf{Improved Precomputation for Polynomials.  }
First we sketch an improved runtime for the $\F_d$ class; the idea is we can partially pre-compute the evaluation of $S_{a,b+1}$ using $S_{a,b}$. 

Again as a warm-up, we consider $\F_0$ and examine just the term $\sum_{i \in \Iab} (f_\inn(i) - x[i])^2$, and show how we can update this to the same term over $I_{a,b+1}$.  Recall that $f_\inn(i) = \mu_{a,b}$ is constant; it is the mean of all values in the interval $\Iab$.  So we can rewrite:
\[
\sum_{i \in \Iab} (f_\inn(i) - x[i])^2 = \sum_{i \in \Iab} (\mu_{a,b} - x[i])^2 
  = \sum_{i \in \Iab} \mu_{a,b}^2 + \sum_{i \in \Iab} x[i]^2 - 2 \mu_{a,b} \sum_{i \in \Iab} x[i].
\]
Then we observe that as we move from $\Iab$ to $I_{a,b+1}$ we can update $\mu_{a,b}$ in $O(1)$ time (by maintaining the sum and dividing by the size $(b-a+1)$).   Also we can maintain the quantities $\sum_{i \in \Iab} x[i]^2$ and $\sum_{i \in \Iab} x[i]$ in $O(1)$ time.  Finally, re-assembling these terms allows us to efficiently compute that first term in $S_{a,b}$.  The outer term can be decomposed the same way, and hence also updated in $O(1)$ time. This improves the total time for $\F_0$ to $O(nw)$.

For the $\F_d$ case, the analysis is similar, but is more involved since the functions $f \in \F_d$ are not constant.  Yet we can still apply a similar precomputation with a careful analysis of their structure.  

\begin{lemma}
    For $\F_d$ we can update score $S_{a,b}$ to $S_{a,b+1}$ in $O(d^3+1)$ time.  
\end{lemma}
\begin{proof}
The key analysis surrounds evaluating the sum of squared errors cost from $\Iab$ to $I_{a,b+1}$.  We now use that for $f \in \F_d$ that
$
f(i) = \sum_{k=0}^d \alpha_k i^k.
$
Then we expand 
\[
\sum_{i \in \Iab} (f_\inn(i) - x[i])^2 = \sum_{i \in \Iab} (\sum_{k=0}^d \alpha_k i^k - x[i])^2
 = \sum_{k = 0}^d \sum_{j = 0}^d \alpha_k \alpha_j (\sum_{i \in \Iab} i^{k+j}) + \sum_{i \in \Iab} x[i]^2  - \sum_{i \in \Iab} x[i] (\sum_{k=0}^d \alpha_k i^k)
\]
Now as we update from $\Iab$ to $I_{a,b+1}$, we can recompute the optimal $f_\inn$ parametrized by $\alpha \in \R^{d+1}$ in $O(d^3)$ time.
Then the first term of the expansion can be computed from $\alpha$ in $O(d^2)$ time from the maintained $2d$ terms in $\sum_{i \in \Iab} i^{h}$ time for $k+j=h \in [0, \ldots, 2d]$.  The second term can be maintained in $O(1)$ time per update since it does not depend on $\alpha$ and $f$.  
The third term can be re-written as 
\[
\textstyle{\sum_{i \in \Iab} x[i] (\sum_{k=0}^d \alpha_k i^k) = \sum_{k=0}^d \alpha_k (\sum_{i \in \Iab} i^k  x[i] )}
\]
so the value $\sum_{i \in \Iab} i^k x[i]$ can be maintained in $O(d)$ time across all values $k \in [0 \ldots d]$.  After this re-organization, this term can be computed from $\alpha$ in $O(d)$ time.  

Using the same function decomposition, $f_\out$ can be solved for and its sum of squared errors can be recomputed in $O(d^3)$ time when $\Iab$ changes $I_{a,b+1}$.  Thus the bottleneck per step is recomputing $f_\inn, f_\out \in \F_d$ in $O(d^3)$ time, and the rest of the maintenance and recomputation can be done in $O(d^2)$ time to obtain $S_{a,b+1}$ from $S_{a,b}$.  
\end{proof}

This computing $\mathsf{Cost}_1$ takes $O(nw (d^3+1))$ time.  
Since we can also solve for $\SRA$ directly in $O(n d^2 + d^3)$ time, this also implies that we can solve for $\SRA$ within that time bound, and Algorithm \ref{alg:scan-alg} takes $O(nw (d^3+1))$ total time for $\F_d$, including when $d=0$.


\textbf{Improved Precomputation for NW Kernel Regression. }
We cannot directly apply this approach for the kernel based methods, since expanding the function $f \in \F_\KRR$ or $\F_\KR$ in a similar ways leads to $\Omega(n)$ terms in the expansion.  However, we can speed up the $n \times \mathsf{evaluation}$ step as a whole, since each evaluation will be similar to the previous one.  

\begin{lemma}\label{lem:KR-update}
    For $\F_\KR$ we can update score $S_{a,b}$ to $S_{a,b+1}$ in $O(n)$ time.  
\end{lemma}
\begin{proof}
Recall a Nadaraya Watson Kernel Regression function $f_K \in \F_\KR$ can be written as:
\[
f_K(i \mid \Iab)  = \textstyle{\frac{\sum_{j \in \Iab} K(i,j) x[j]}{\sum_{j \in \Iab} K(i,j)} = \frac{N(i \mid \Iab)}{D(i \mid \Iab)}}
\]
where the numerator $N(i \mid \Iab) = \sum_{j \in \Iab} K(i,j) x[j]$ and denominator $D(i \mid \Iab) = \sum_{j \in \Iab} K(i,j)$ both have a linear number of terms to sum up.  

For each $i \in [n]$ we can store these numerators $N(i \mid \Iab)$ and denominators $D(i \mid \Iab)$ if $i \in \Iab$; and the same for the complementary ones $N(i \mid [n] \setminus \Iab)$ and $D(i \mid [n] \setminus \Iab)$ if $i \notin \Iab$ which contributes to $f_\out$.  Then when $\Iab$ shifts to $I_{a,b+1}$ it takes $O(1)$ time to update each numerator and denominator, for each $i \in [n]$, in total $O(n)$ time.  Summing these in $O(n)$ time yields $S_{a,b+1}$.  
\end{proof}

To run Algorithm \ref{alg:scan-alg} on $\F_\KR$ now takes $O(n^2)$ time to compute $\SRA$, and then each of the $O(nw)$ iterations of the for loop takes $O(n)$ time, so the total improved runtime is now $O(n^2w)$.  
Computing $\SRA$ for the Laplace kernel can be reduced to $O(n)$ using the special structure of that kernel, but this does not improve the overall runtime.  


\textbf{Improved Runtime for Truncated Kernels. }
Finally, we note that because of the very structured nature of the data we consider, there are rarely many pairs $i,j \in [n]$ that have a significant effect on functions $f$ and the resulting scan statistic.  Both Gaussian and Laplace kernels have (squared) exponential decay in effect as $|i-j|$ increases beyond the bandwidth parameter $r$.  As a result, it is common to use truncated kernels where we set $K(i,j) = 0$ if $|i-j| \geq C \cdot r$ for some constant $C$ (e.g., $C=3$) multiple of the bandwidth. 
We can improve efficiency for $\F_\KR$ 
with key insight that changing $\Iab \to I_{a,b+1}$ only updates the function values $f_\inn$ and $f_\out$ in $O(r)$ locations.  
\begin{lemma}\label{lem:KR-update-r}
    For $\F_\KR$ using a kernel $K_r$ truncated at radius $O(r)$, we can update score $S_{a,b}$ to $S_{a,b+1}$ in $O(r)$ time.  
\end{lemma}
\begin{proof}
As in the proof of Lemma \ref{lem:KR-update}, we can maintain numerators $N(i \mid \Iab)$ and denominators $D(i \mid \Iab)$ for each $i \in \Iab$ and similar for $i \notin \Iab$.  We now also save their ratio $f_K(i \mid \Iab)$, and their sum of squared errors from $x[i]$.  

Now in transition from interval $\Iab$ to $I_{a,b+1}$, the observation at location $b+1$ moves from part of the $f_\out$ to the $f_\inn$ model.  This means we need to subtract its effect on $f_\out$ and add its effect to $f_\inn$.  But due to kernel truncation, this will only effect $O(r)$ entries in each.  
We first subtract the effect of those $O(r)$ entries from the sum of square errors using the stored ratio $f_K(i \mid \Iab)$, in $O(r)$ time.  
Next for each of those entries, we update their numerators and denominators, again in total $O(r)$ time.  
Finally, we recompute their ratio and updates the sum of square errors in $O(r)$ time.  
\end{proof}

Computing $\SRA$ for the $\F_\KR$ is also faster.  Determining the numerators and denominators for each $i$ only sums over $O(r)$ terms, so this takes $O(nr)$ time in total.  Then computing their estimates and sum of squared errors for $\SRA$ only takes $O(n)$ additional time.  As result for a kernel truncated at $O(r)$, the spatial scan statistics $1-\frac{\textsf{Cost}_1}{\SRA}$ can be completed in $O(nwr)$ time.  We outline this full process in Algorithm \ref{alg:nwkr-scan}, with $\neigh_r(j) = \{i \in [n] \mid |i-j| \leq Cr\}$ as the truncated neighborhood.  

\begin{algorithm}
\caption{Truncated-NWKR-Scan($\bfx, w, r$)}
\label{alg:nwkr-scan}
\begin{algorithmic}
    \For{$i \in [n]$}
        \State Compute $N_r(i) = \displaystyle{\sum_{j \in [i-Cr,i+Cr]}} K_r(i,j) x[j]$, 
        $D_r(i) = \displaystyle{\sum_{j \in [i-Cr,i+Cr]}} K_r(i,j)$, 
        $\hat{x}_{\text{all}}(i) = \frac{N_r(i)}{D_r(i)}$
    \EndFor
    \State $\textsf{Cost}_1 \gets \SRA = \sum_{i=1}^n \bigl(x[i]-\hat{x}_{\text{all}}(i)\bigr)^2$

    \For{each starting index $a \in [n-1]$}
        \State Initialize $\I_{a,a-1}=\{\}$; \; $\Omega_{a,a-1} = [n]$ and $N_\text{out} = N_r$; \; $D_\text{out} = D_r$; \; $\hat x_{\text{out}} = \hat x_{\text{all}}$

        \For{$b \in [a,\dots,a+w]$}
            \State Add the new index:
            $
            \I_{a,b}=\I_{a,b-1}\cup\{b\},
            \qquad
            \Omega_{a,b}=\Omega_{a,b-1}\setminus\{b\}.
            $


            \State  $S^r_{a,b} = S_{a,b-1} - \sum_{i \in \Iab \cap \neigh_r(b)} (x[i]-\hat{x}_{\text{in}}(i))^2  - \sum_{i \in \Omega_{a,b} \cap \neigh_r(b)} (x[i]-\hat{x}_{\text{out}}(i))^2$
            
            \For{$i \in \neigh_r(b)$}
                \State  \textbf{if}($i \in \Iab$):  
                $N_{\text{in}}(i) \mathrel{+}= K_r(i,b) x[b]$; \;\;\;
                $D_{\text{in}}(i) \mathrel{+}= K_r(i,b)$; \;\;\;
                $\hat{x}_{\text{in}}(i)= \frac{N_{\text{in}}(i)}{D_{\text{in}}(i)}$

                \State \textbf{if}($i \notin \Iab$): $N_{\text{out}}(i) \mathrel{-}= K_r(i,b) x[b]$; \;\;
                $D_{\text{out}}(i) \mathrel{-}= K_r(i,b)$; \;\;
                $\hat{x}_{\text{out}}(i)= \frac{N_{\text{out}}(i)}{D_{\text{out}}(i)}$
            \EndFor

            \State $S_{a,b} = S^r_{a,b} + \sum_{i \in \Iab \cap \neigh_r(b)} (x[i]-\hat{x}_{\text{in}}(i))^2  + \sum_{i \in \Omega_{a,b} \cap \neigh_r(b)} (x[i]-\hat{x}_{\text{out}}(i))^2$

            \State \textbf{if} $S_{a,b}<\textsf{Cost}_1$ \textbf{then} $\textsf{Cost}_1 \gets S_{a,b}$
        \EndFor
    \EndFor

    \State \textbf{return} $1-\frac{\textsf{Cost}_1}{\SRA}$
\end{algorithmic}
\end{algorithm}


\section{Evaluation}
\label{sec:eval}
We next evaluate the efficiency and effectiveness of our methods at identifying interval anomalies in noisy smoothly varying signals; we evaluate both how well it localizes known anomalies and how well if flags anomalous signals from non-anomalous ones.  We use synthetic (for controlled experiments) and real world examples (frequency-domain radio telescope signals and solar monitoring time series, in Section \ref{app:solar}).  
We generate synthetic data from two baseline models: one with a polynomial model (of degree 2) and another with an AR(2) model with parameters $\phi_1 = 1.985; \phi_2 = -0.985056$ (derived via Vieta's formula with roots $r_1 = 0.993$, $r_2 = 0.992$; 
so the process is stationary, and slowly drifting).  For both we add iid normal noise $\mathcal{N}(0,\sigma^2)$ to simulate sensing variability about 5\% of the natural signal variation; see Figure \ref{fig:poly_ar2}.  The default signal length is $n=500$.  
Then for some of the generated signals we ``plant" an anomalous interval that deviates from the original signal in a continuous block, with signal to noise ratio (SNR) varying from $1\!\times\!\sigma$ to $5\!\times\!\sigma$, and interval width in 1\% to 20\% of the signal.  Additional details and plots are shown in Appendix \ref{sec:synthetic-data}.
\begin{figure}[htbp]
    \centering
    \begin{subfigure}[b]{0.47\textwidth}
        \centering
        \includegraphics[width=\linewidth]{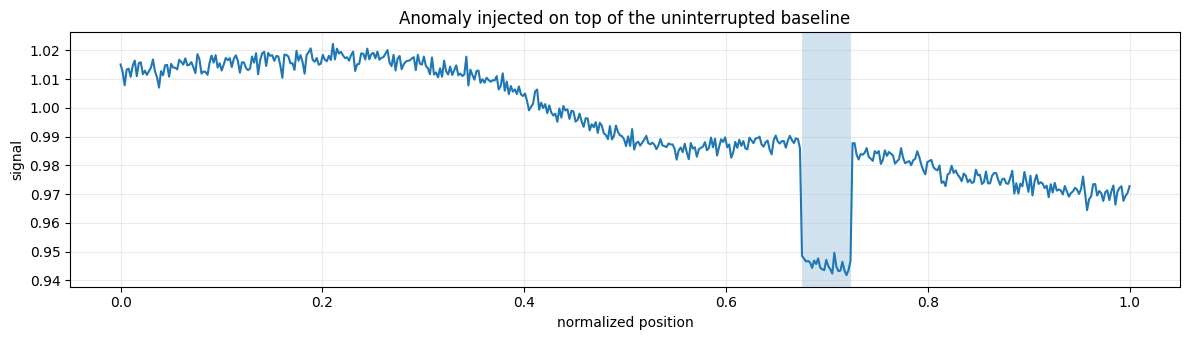}
        \caption{Polynomial signal}
        \label{fig:signal4}
    \end{subfigure}
    \hfill
    \begin{subfigure}[b]{0.47\textwidth}
        \centering
        \includegraphics[width=\linewidth]{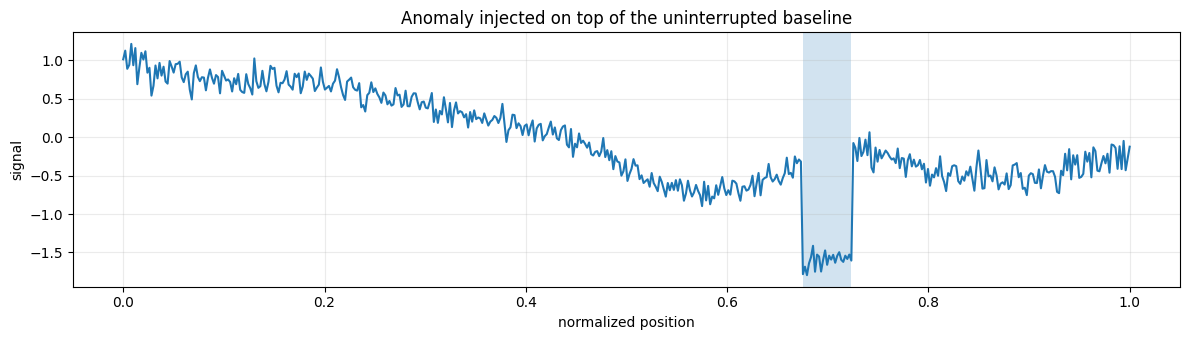}
        \caption{Ar(2) signal}
        \label{fig:signal10}
    \end{subfigure}

    \caption{Example synthetic signal. Anomaly width at $5\%$ and SNR=$2.5\!\times\!\sigma$.}
    \label{fig:poly_ar2}
\end{figure}

In addition to comparing to scan statistics for the $\F_0$ (representing models of \cite{huang2007spatial,agarwal2006hunting}), we consider four common methods for change point detection 
 (Ruptures KernelCPD~\cite{truong2020selective};  
 Bayesian Online Changepoint Detection (BOCPD)~\citep{adams2007bayesian};
 Gaussian LRT~\citep{siegmund1995using}; and
 Collective and Point Anomaly detection (CAPA)~\citep{fisch2022capa}) 
 adapted to find an interval between the two most promising changepoints.  
 Second, we consider STUMPY \citep{law2019stumpy} a method for discord~\citep{yeh2016matrixprofile} interval mining (STUMPY \citep{law2019stumpy}).  
 Third we consider three deep unsupervised detectors (USAD~\citep{audibert2020usad}, TranAD~\citep{tuli2022tranad}, M2N2~\cite{kim2024model}) that build a model of good data, and then identify points which deviate from it.  We adapt all three to find interval anomalies by finding large contiguous intervals of excess score.  
 Finally, we consider two variants which use NWKR based on our own code:  
   NWKR-CPD fits a single NWKR model, subtracts it, and feeds the residual to ruptures, and 
   NWKR-FL, uses our code to find a single change point by fixing the first endpoint to be the left boundary and only searching over the right boundary.  
 See Appendix \ref{sec:related_work} for more details.

\subsection{Recovering and Localizing Planted Anomalies}
\label{sec:recover}
We first demonstrate the effectiveness of our methods in recovering planted anomalies on synthetic data, the polynomial data results are shown in Figure \ref{fig:poly-localization} and the AR(2) data results are shown in Figure \ref{fig:ar2-localization}.
The first experiment in Figure \ref{fig:poly-localization}(left) shows the effect of changing the signal-to-noise ratio of the planted anomalies from $1\times$ to $5\times$  the normal noise in the original signals.  Using $N=1000$ generated signals with randomly planted anomalies of length $50$ channels, we report the average (and show std.dev. bars) \emph{localization score} (it is the geometric mean of recall and precision). 
Then Figure \ref{fig:poly-localization}(right) shows as we fix the anomaly depth as $2.5\times$ the background noise, and the vary width of the planted anomaly from $5$ to $50$ channels. 

\begin{figure}[ht]
  \includegraphics[width=0.49\linewidth]{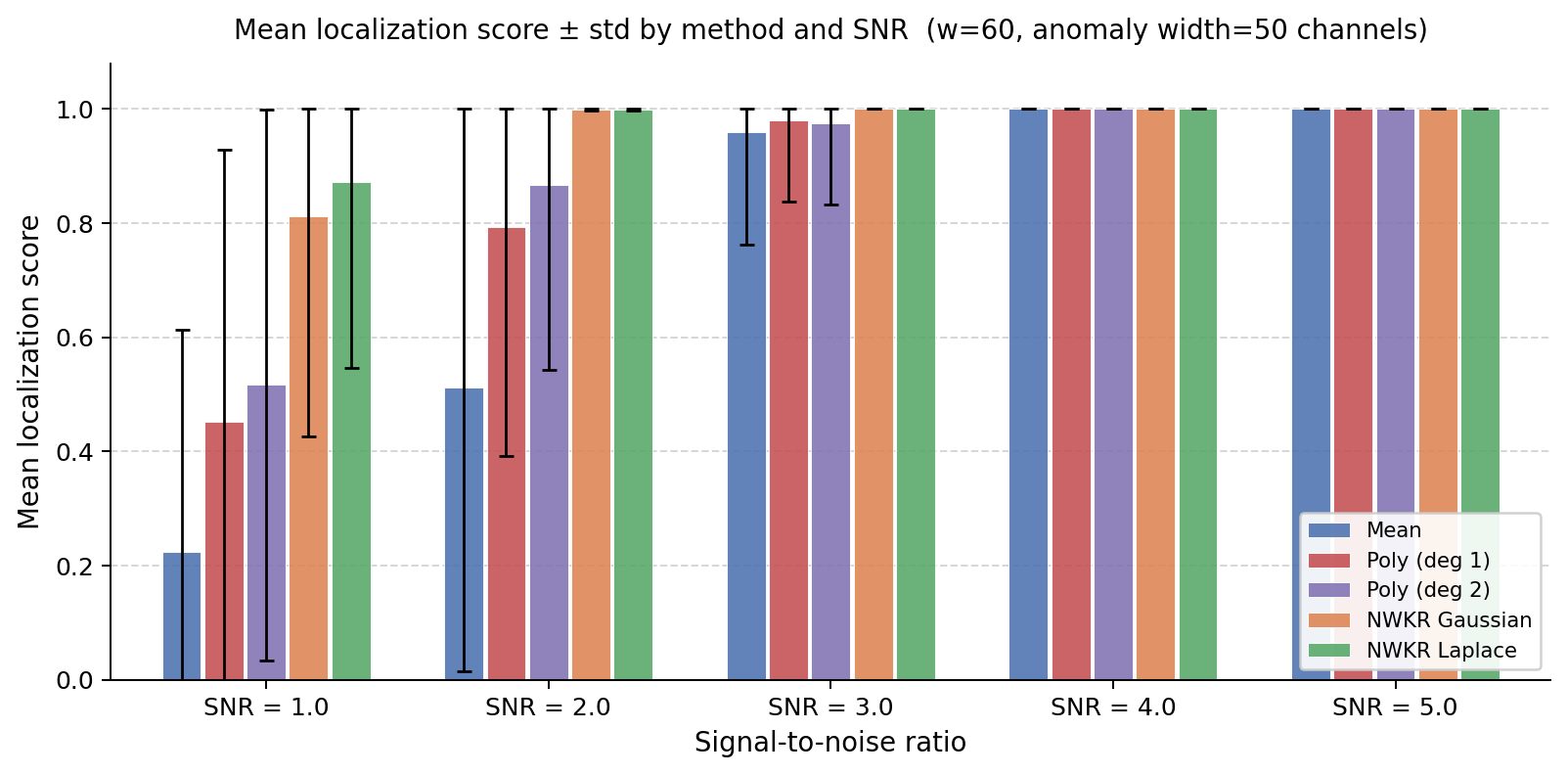} \includegraphics[width=0.49\linewidth]{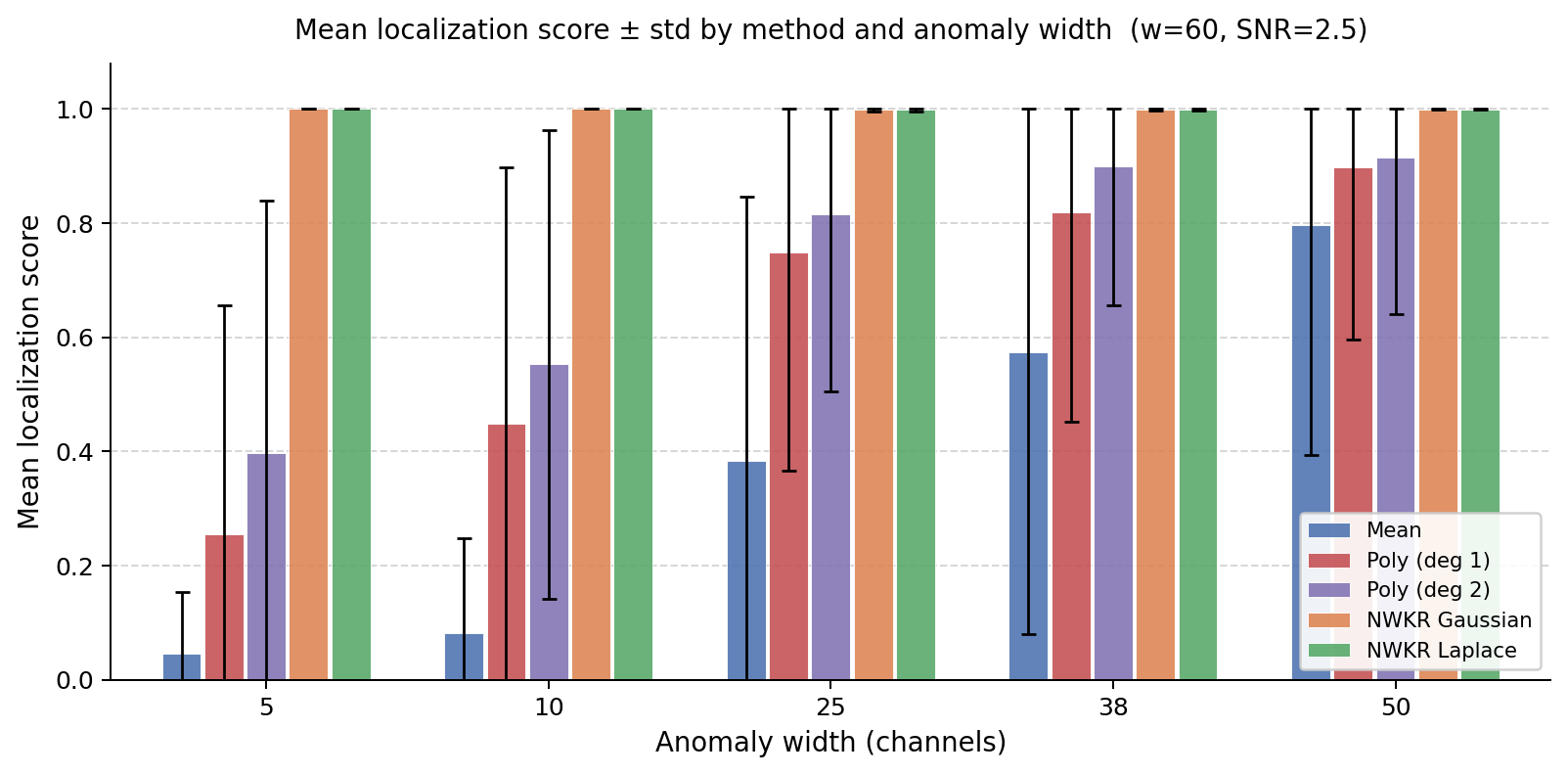}
  \vspace{-2mm}
  \caption{Mean localization score comparison on quadratic trend data by varying SNR or anomaly width.}
  \label{fig:poly-localization}
\end{figure}

\begin{figure}[ht]
  \includegraphics[width=0.49\linewidth]{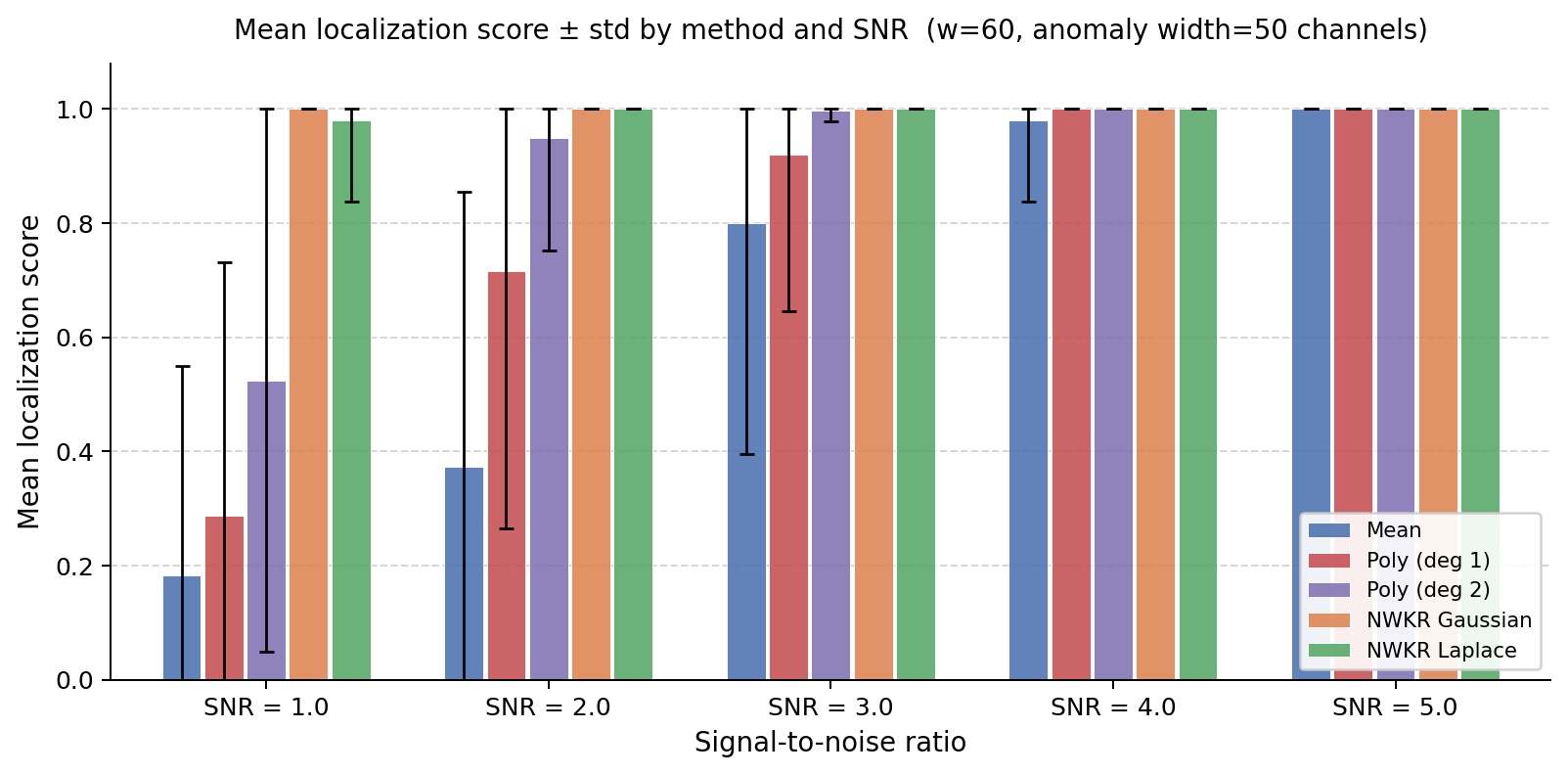} \includegraphics[width=0.49\linewidth]{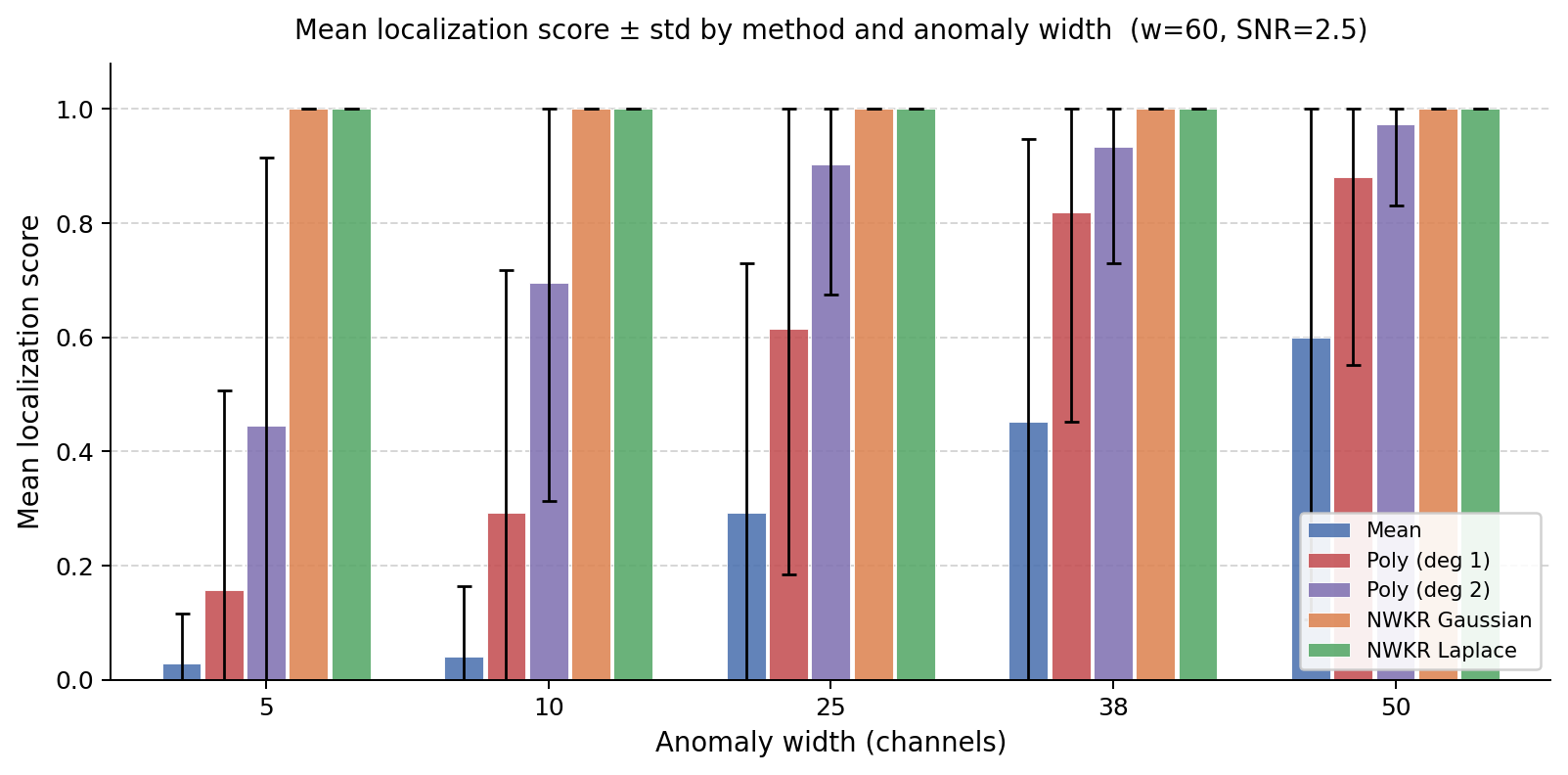}
  \vspace{-2mm}
  \caption{Mean localization score comparison on AR(2) data by varying SNR or anomaly width.}
  \label{fig:ar2-localization}
\end{figure}

\begin{figure}[ht]
  \includegraphics[width=\linewidth]{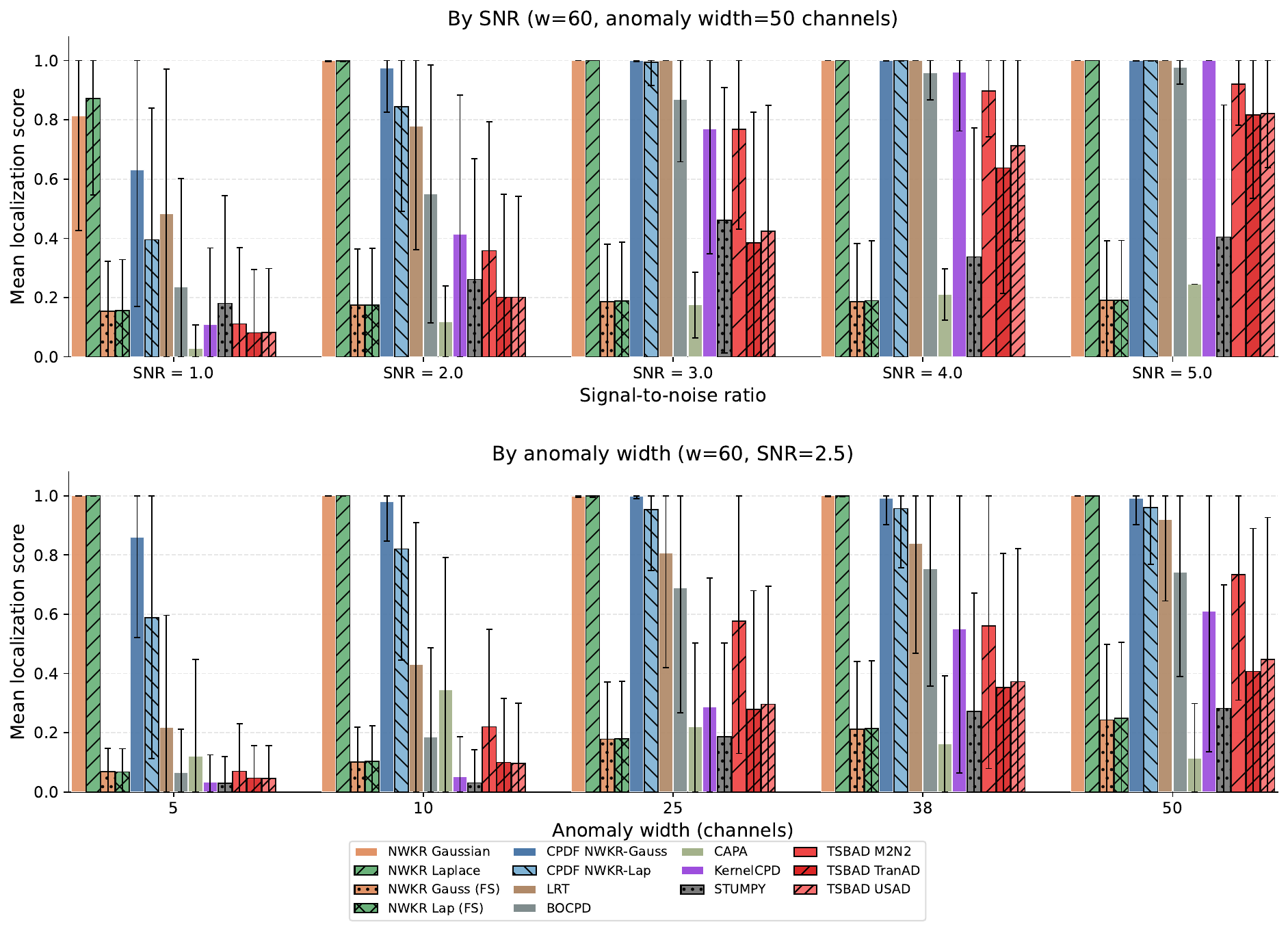}
\vspace{-2mm}
  \caption{Comparison of NWKR with other baselines on quadratic trend data in a varying SNR (left) and anomaly width (right).}
  \label{fig:methods-local}
\end{figure}

Comparing the performance of the different function families ($\F_0$, $\F_1$, $\F_2$, $\F_\KR$ with Gaussian and  Laplace), we see that with large SNR (at $5\times\!\sigma$) all methods have the maximum localization score of $1$.  But \textbf{as SNR decreases, the NWKR methods retains near-perfect localization}, while the mean ($\F_0$) or polynomial ($\F_1, \F_2$) lose the ability to localize the anomalies.  Similarly, the $\F_\KR$ models can localize perfectly at all widths, while parametric models (especially $\F_2$) do ok at large width of $50$, but do poorly at smaller widths.  This shows the $\F_\KR$ family is most robust at localizing anomalies.  

We also run the same experiments compared against the change point baselines and discord mining approaches in Figure \ref{fig:methods-local}.  Similarly we find that they perform poorly with small SNR and width, with much worse localization ability compared to our NWKR models.  However, as SNR and anomalous interval width increases they perform better, with the LRT method (and to less extent KernelCPD and BOCDP) approaching or matching performance of NWKR models at very high SNR or width. CPDF variants and TSBAD M2N2 also benefit from the easier regimes, but their gains are less consistent and they remain less competitive overall than the NWKR family.

\begin{figure}
  \includegraphics[width=0.33\linewidth]{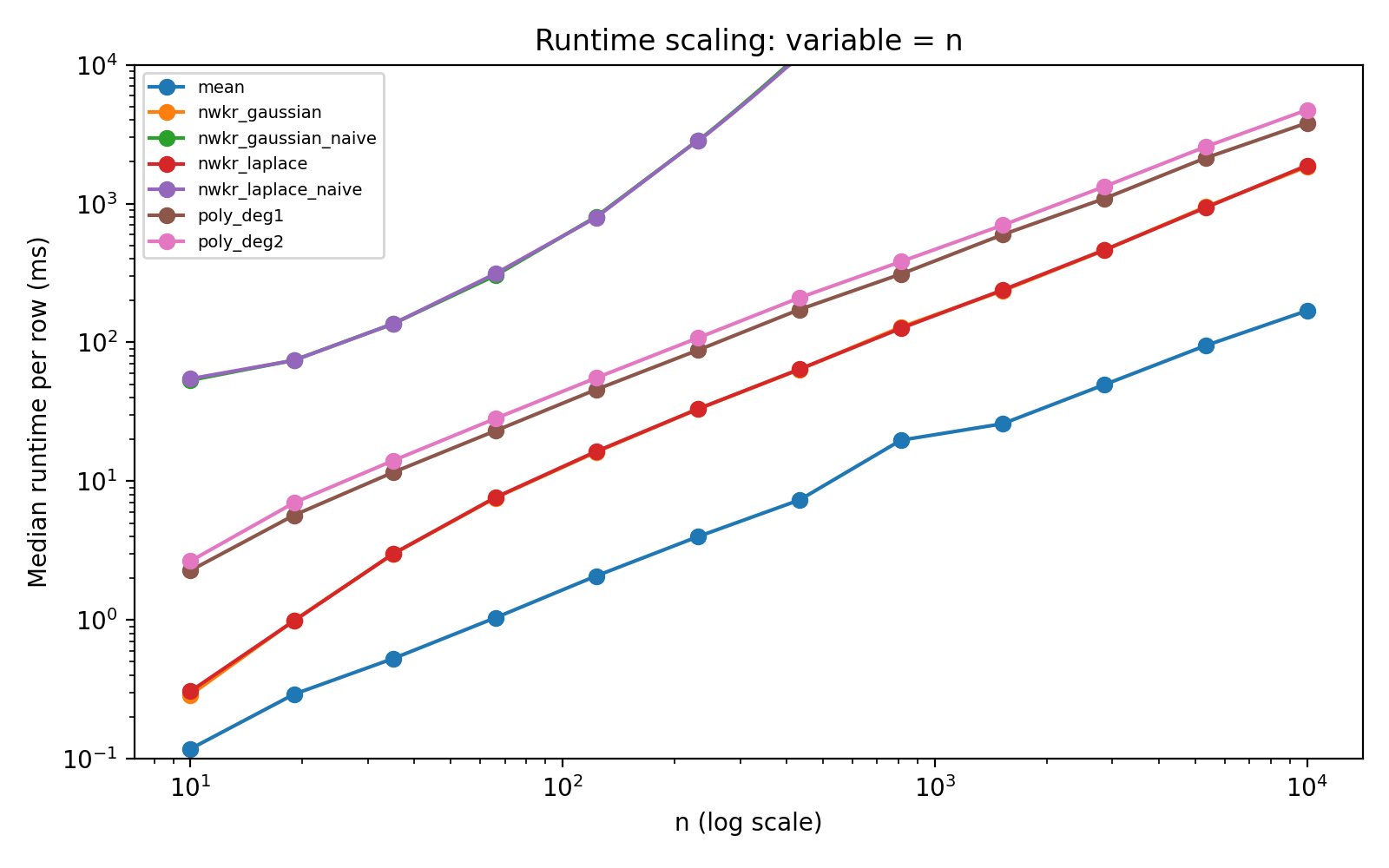}
  \includegraphics[width=0.33\linewidth]{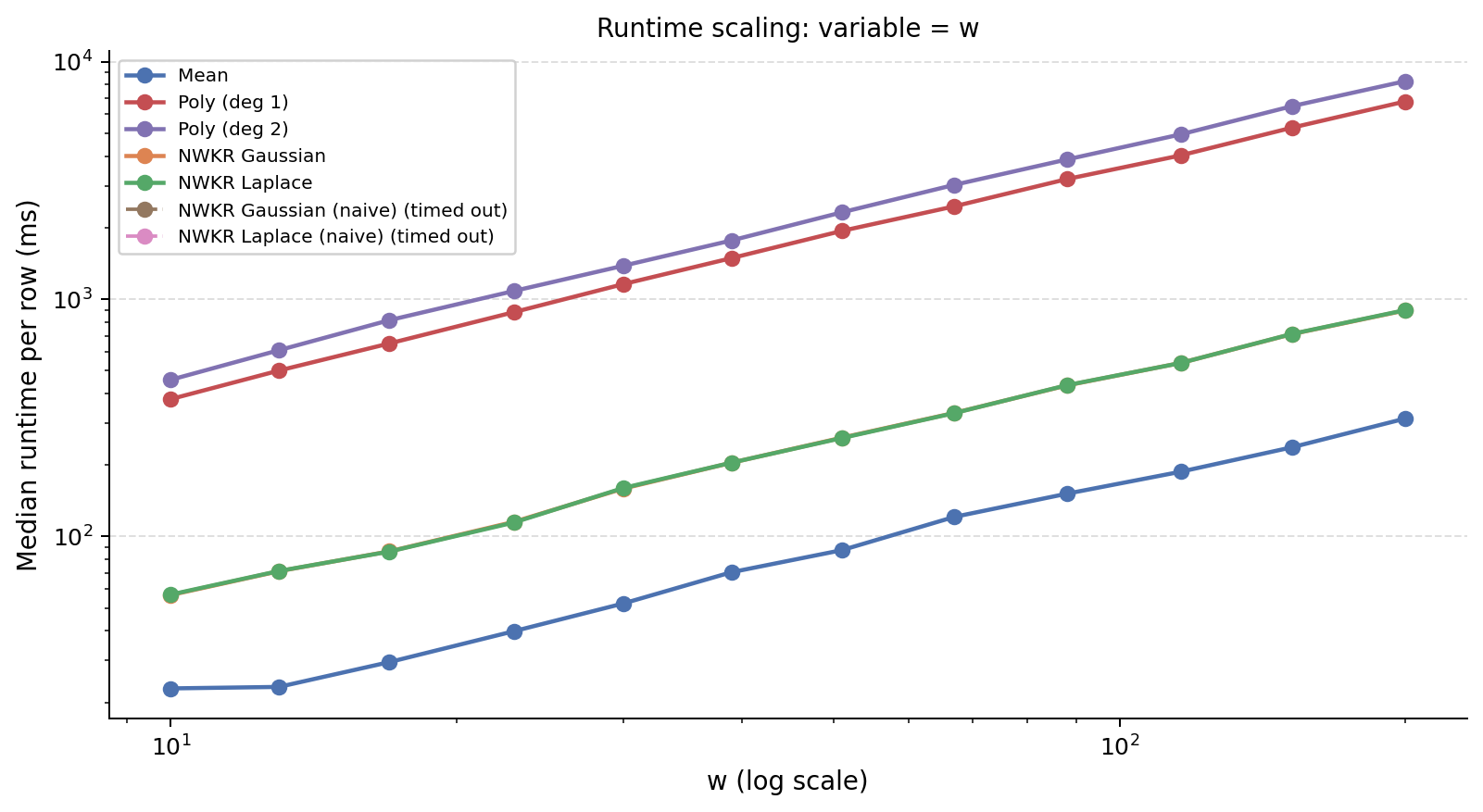}
  \includegraphics[width=0.33\linewidth]{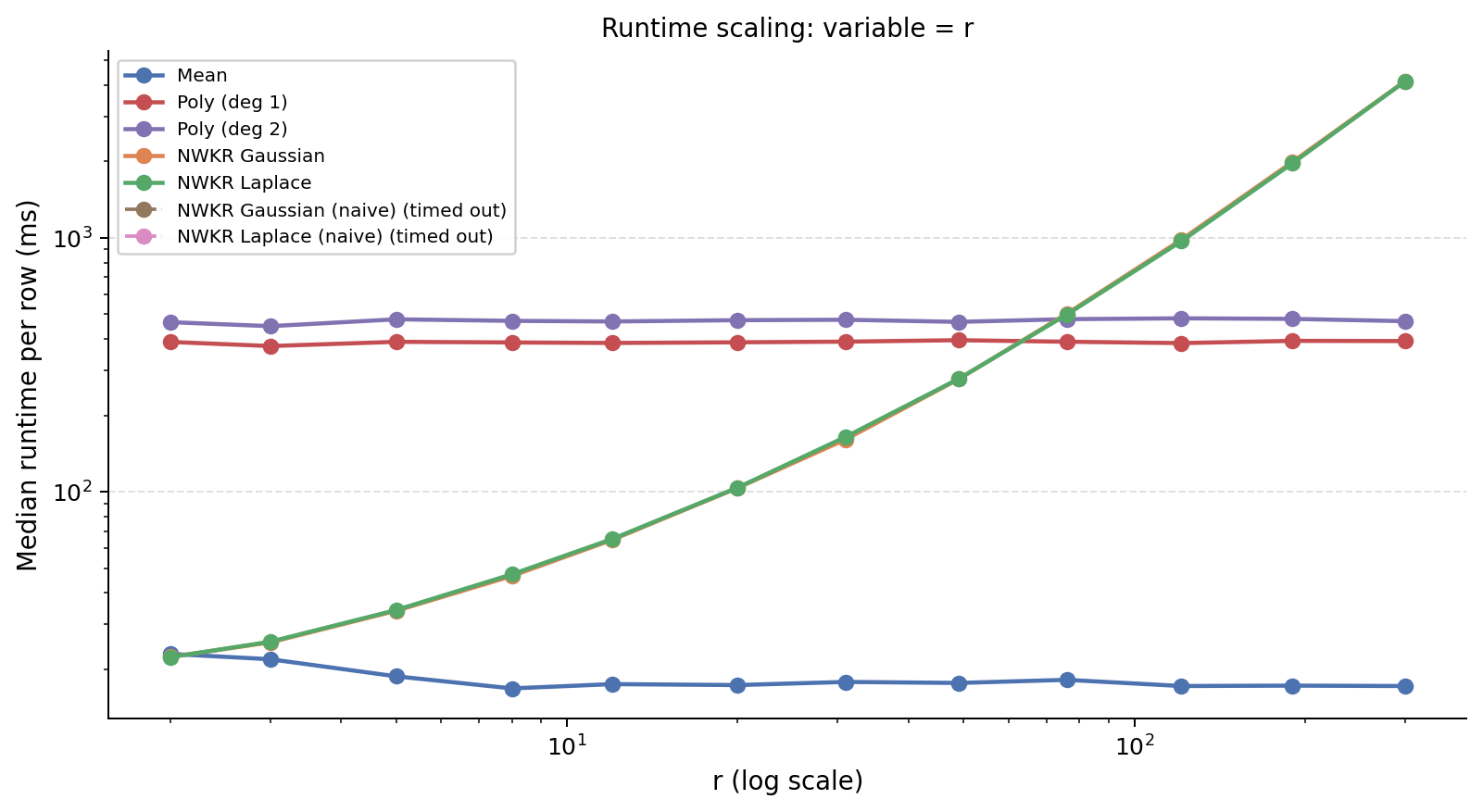}
  \caption{Log-log plots for runtime scaling on signal as a function of signal length $n$ (left),  $w$ (middle), and truncation range $r$ (right).  Otherwise parameters fixed ($n=10,000$, $w=200$, $r=300$)}
  \label{fig:runtime}
\end{figure}


\clearpage

\begin{wraptable}{r}{0.4\textwidth}
\vspace{-9mm}
\centering
\setlength{\tabcolsep}{4pt}
\renewcommand{\arraystretch}{0.92}
\caption{Median runtime at $n=1000$.}
\label{tab:runtimes}
\vspace{-2mm}
\begin{tabular}{@{}lr@{}}
\toprule
\textbf{Method} & \textbf{time (ms)} \\
\midrule
$\mathcal{F}_0$                  &      17 \\
\rowcolor{gray!15}
$\mathcal{F}_\text{NWKR}$ Laplace          &      69 \\
\rowcolor{gray!15}
$\mathcal{F}_\text{NWKR}$ Gaussian         &      70 \\
$\mathcal{F}_1$ Poly (deg 1)          &     392 \\
$\mathcal{F}_2$ Poly (deg 2)          &     472 \\
\midrule
$\mathcal{F}_\text{NWKR}$ Gaussian {\tiny (naive)} & 759{,}396 \\
$\mathcal{F}_\text{NWKR}$ Laplace {\tiny (naive)}  & 762{,}321 \\
$\mathcal{F}_\text{KRR}$ Gaussian & 708{,}696 \\
$\mathcal{F}_\text{KRR}$ Laplace & 709{,}193 \\
\midrule
CAPA                  &       3 \\
KernelCPD             &  5 \\
LRT                   &      37 \\
BOCPD                 &   2{,}785 \\
STUMPY                &   2{,}776 \\
\midrule
USAD                  &     173 \\
TranAD                &     210 \\
M2N2                  &     297 \\
\midrule
NWKR-FS Gaussian &       2 \\
NWKR-FS Laplace &       2 \\
NWKR-CPD Gaussian     &       5 \\
NWKR-CPD Laplace      &       5 \\
\bottomrule
\end{tabular}
\vspace{-0.8\baselineskip}
\end{wraptable}

\subsection{Runtime Scaling}
\label{sec:runtime-exp}

We confirm the efficiency of our regression scan algorithms in Table \ref{tab:runtimes}, and also plot runtime scaling plots Figure \ref{fig:runtime}.  We observe that for our (non-naive) methods for $\F_\KR, \F_d$ that there is linear scaling with $n$ and with $w$, and that $\F_\KR$ has linear scaling with $r$.  Moreover, while the simple constant rate mode $\F_0$ is the most efficient, our optimized $\F_\KR$ is also very efficient (about $70$ ms for $n=1000$), and almost an order of magnitude faster than the $\F_d$ models which need an expensive matrix inverse; even though the inverse is on a small matrix, the repeated call adds up.  Moreover, the naive implementations of $\F_\KR$ (before our optimizations in Section \ref{sec:faster}) become intractable for large values of $n$, taking almost $1000\times$ longer than ours on a signal of length $n=1000$.  
Moreover, some baseline methods (like BOCPD, STUMPY) are at least an order of magnitude slower than our approach.  While the simpler LRT, CAPA, and KernelCDP are faster than ours, as we see in Figure \ref{fig:methods-local}, they do not model the anomalies as well.  
We also run against some modern unsupervised learning methods USAD, TranAD, and M2N2 which are slower than ours.  
Finally, we consider the NWKR variants: the NWKR-FS (the fixed start variant) and NWKR-CPD (which subtracts the fixed NWKR model and then runes linear CPD); both are faster than our methods, but do not perform as well.

\section{Application: Radio Telescope Spectra Calibration Anomalies}
\label{sec:astro}





We have collected a dataset of $N = 38{,}881$ bandpass calibration solutions from the QA2 (quality assurance, phase 2) step that is part of the processing that converts raw interferometric radio telescope signals collected at ALMA~\citep{yus2020towards,Nakos20} into hyperspectral data cubes.  These data cubes are one of the primary objects of study in modern astronomy; they are an image with hundreds or thousands of frequency values per pixel.  Constructing these data cubes from the raw measurements is a complex process~\citep{alma_pipeline_heuristics_paper}.  One relatively common, instrumental problem is platforming~\citep{alma_pipeline_team_2025_users_guide} (see Figure \ref{fig:AR2-w5}) where a misalignment between parts of the measurement causes an interval drop in the raw frequency readings coming off of a pair of radio telescopes measured by the correlator \citep{escoffier2007alma}.  If this issue is not detected, the ultimately constructed data cubes can end up with strange artifacts that interfere with scientific inference.  
This platforming effect is most evident at QA2 during bandpass calibration review:  $n$-dimensional arrays of amplitude values at different frequencies, used to transform raw signals into a unbiased representation.  The expected input is not ``flat" as it varies smoothly with frequency in a way that corrects for telescope specific variation, for which it is used to correct.  Each value has noise, which can be modeled well as iid Normal, with variance depending on the measurement source and sensing conditions.  

\begin{figure}[ht]
    \includegraphics[width=0.95\textwidth]{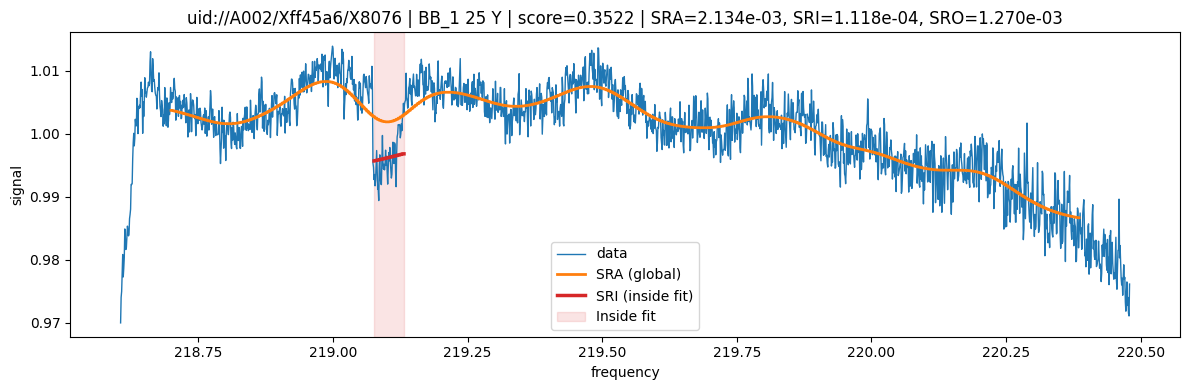}

    \caption{{ALMA calibration signal with platforming.}}
    \label{fig:AR2-w5}
    \vspace{-10pt}
\end{figure}

Until recently, platforming anomalies were identified and flagged by a person, known as a \emph{data reducer}, who looked at many (thousands of) frequency-amplitude plots try to spot pernicious issues that would be likely to cause reconstruction error.  While only a small fraction of flags were false-positives, the false-negative rate is quite high at about $50\%$.  This means about half of the platforming anomalies were missed, and can lead to corrupted data cubes.  
With new higher-throughput telescopes coming online in the new future, this failure rate is not acceptable, and moreover the process of having a human inspect each calibration is not scalable.

\paragraph{ALMA pipeline heuristic.}  
In October 2025 (Cycle 12), the ALMA pipeline installed a platforming detector~\citep{alma_pipeline_team_2025_users_guide}.  It is a hand-tuned set of threshold-based heuristics applied to bandpass calibration solutions, restricted to baseline-correlator FDM data.  It takes advantage that platforming anomalies are likely to occur at certainly locations, known as \emph{subbands}.  
Each spectral window is partitioned into effective $62.5\,\mathrm{MHz}$ subbands (implemented as $62.5\,\mathrm{MHz}\times 15/16$), flagged channels and known WVR local-oscillator leakage channels are masked, and subbands whose centers lie inside fitted atmospheric absorption features or where modeled transmission falls below $0.3$ are excluded using a Lorentzian fit to an atmospheric transmission profile. 
On the remaining data, the code defines a local noise scale as 
\[
\min(\mathrm{std}(x_{i+4}-x_i), \mathrm{std}(x_i)),
\]
computes per-subband phase/amplitude RMS and mean or median, and then applies five threshold tests: 
\begin{enumerate}
    \item anomalously high phase RMS:
    greater than $5\times$ the median subband phase RMS, excluding the largest
    subband, and greater than $10^\circ$, with a Sobel-gradient precheck;

    \item anomalously high amplitude RMS:
    greater than $5\times$ the corresponding median amplitude RMS, again with a
    Sobel precheck;

    \item anomalous phase offsets between adjacent subbands:
    jump greater than $3\sigma$ for interior subbands or greater than $6\sigma$
    at the band edges, plus boundary-step constraints and an absolute jump
    greater than $5^\circ$;

    \item anomalous amplitude offsets between adjacent subbands:
    analogous jump tests with $5\sigma$ boundary-step thresholds for interior
    subbands;

    \item amplitude spikes near subband boundaries:
    detected by comparing a target window of width $\approx 0.3$ subband to
    neighboring windows and requiring both greater than $6\sigma$ deviation and
    greater than $10\%$ relative excursion.
\end{enumerate}
A spectral window is marked affected if any antenna/polarization triggers one or more of these rules, with failures labeled as phase, amplitude, or both. While this combination of metrics is fully automated, and achieves slightly improved false positive rate, it still has a similar false positive rate as the data reducers.  This tradeoff is still undesirable for the detection of rare but consequential anomalies, and motivates the development of richer feature representations within a more principled detection framework.

\paragraph{Additional challenges.} 
Another challenge is that the raw signals can have other anomalous issues.  
The boundaries of the sensed frequency ranges sometimes include unreliable regions right near the edges of the signals; see left side in Figure \ref{fig:AR2-w5}.  So we pre-filter a small buffer of rows at the beginning and end of each signal to not be fooled by these measurement issues.  

Moreover, there can be other known (and independently modeled) challenges such as some sensing conditions have interference from the Earth's atmosphere which cause other absorption features in the amplitude values.  These are somewhat understood and can be inferred to some degree at the time of sensing.  While this atmospheric interference is an accepted and useful part of calibrating visibility signals ahead of the data cube reconstruction, it visually appears anomalous (somewhat similar to platforming, but more pointy), and can occur in conjunction with platforming issues.   In general, observational schedules for the telescopes try to avoid conditions that will cause this interference, but it is sometimes unavoidable.  To make sure we isolate the core issue of detecting the interval anomalies in platforming, and not getting into the nuance of how finely we model these rare atmospheric interference cases, in this work we omit any signal that has this sort of interference.  

Finally, anything used must be incorporated into a very high-throughput process~\citep{alma_pipeline_heuristics_paper} and cannot take time significantly longer than reading the data.

\paragraph{The Large Cleaned ALMA dataset.}  
Ultimately, our central labeled dataset was constructed from ALMA QA2 calibration tables. Each row $\bfx$ corresponds to a single polarization of a single spectral window observed by a single antenna. We then filter out the signals with atmospheric interference or fewer than $n=128$ channels, and trim a buffer of 5\% of the spectrum length from each end of each signal.
Ultimately, our domain experts identified a set of $231$ confirmed cases of platforming among all of this refined, large set of signals.  Rows were assigned a binary label: positive if there is an anomaly present, and negative otherwise. 
  The result is a set of $N = 38{,}881$ rows, of which $231$ are positive ($0.6\%$), and $38{,}650$ are negative.

We have released this:
\begin{itemize}
    \item data set: \url{https://github.com/BeardyMan37/RegressionScanStats/blob/main/dataset/labelled_dataset.parquet}
    \item code: \url{https://github.com/BeardyMan37/RegressionScanStats/}
    \item smaller sampled data set used in some experiments: \url{https://github.com/BeardyMan37/RegressionScanStats/blob/main/dataset/sampled_labelled_dataset.parquet}
\end{itemize}

\begin{figure}[b]
\vspace{-2mm}
\includegraphics[width=\linewidth]{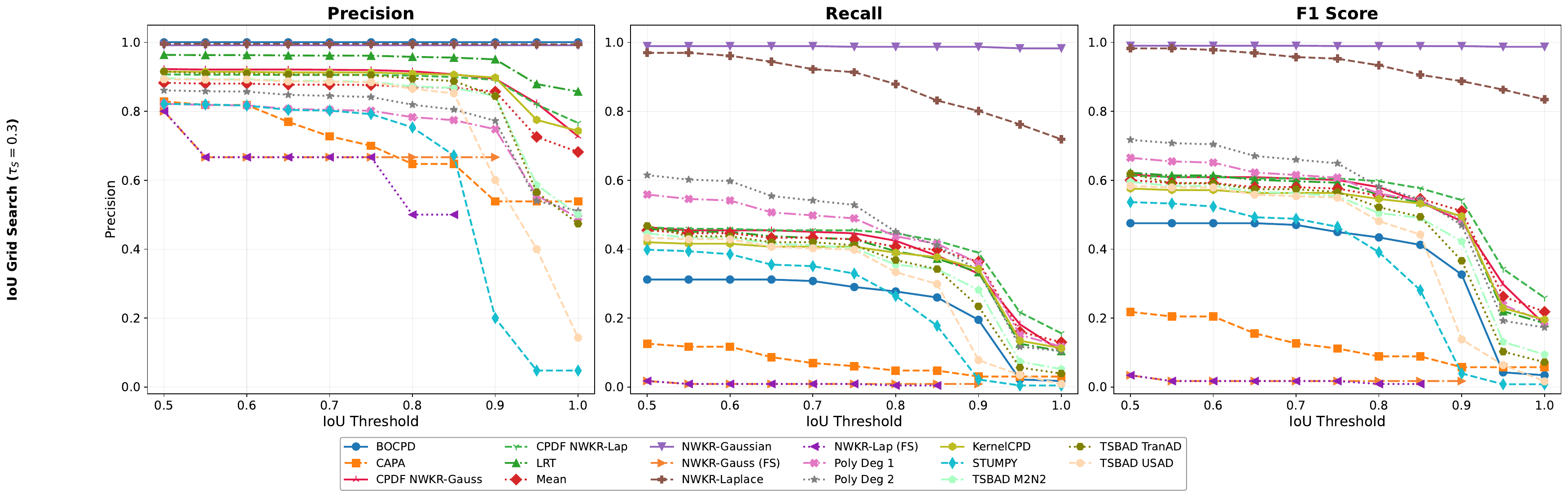}
  \vspace{-6mm}
  \caption{Performance analysis of different methods on variable IOU settings on ALMA dataset}
  \label{fig:qa2_IoU-performance-main}
\end{figure}

\subsection{Precision and Recall}

In Figure \ref{fig:qa2_IoU-performance-main} we show the performance in identifying the correct interval among the $231$ platforming examples collected.  We measure IoU (intersection-over-union, aka the Jaccard similarity), and show precision, recall, and F1 score as we vary the IoU threshold $\tau_I$.  We observe that among the $\F_d$ and $\F_\KR$ models, the NWKR methods significantly outperform the mean $\F_0$, the polynomial $\F_d$ methods, and all other baselines.  At an IoU threshold of $0.75$, then F1 score for $\F_\KR$ Gaussian is above $0.95$, whereas it is below $0.72$ for all other approaches, with some significantly worse. The advantage is consistent across all thresholds: as $\tau_I$ increases and the localization requirement becomes stricter, the NWKR methods degrade more gracefully than the baselines, reflecting their ability to precisely identify the interval boundaries rather than merely detecting the presence of an anomaly.


\begin{figure}[htbp]
\vspace{-2mm}
\includegraphics[width=\linewidth]{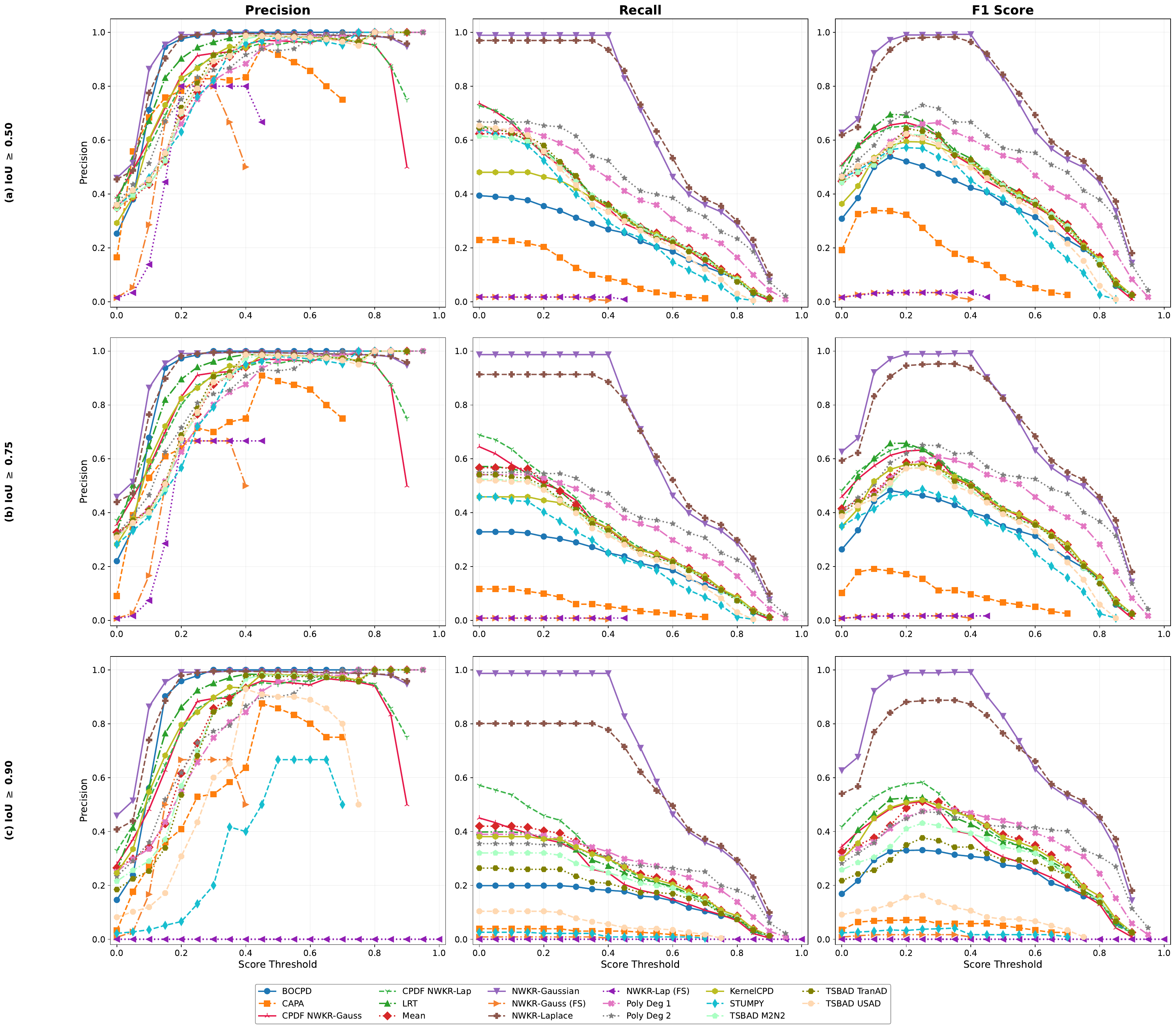}
  \vspace{-6mm}
  \caption{Performance analysis of different methods on variable score settings on ALMA dataset}
  \label{fig:qa2_score-performance-main}
\end{figure}

Then in Figures~\ref{fig:qa2_score-performance-main} (a), (b), and (c) we show the results of filtering on two criteria: an IoU threshold $\tau_I$ and also a score threshold.  In each figure we fix $\tau_I$ at $0.5$, $0.75$, and $0.9$ respectively and show precision, recall, and F1 as a function of the score threshold $\tau_S$.  To predict an anomaly correctly, a method needs to succeed in both the score $\geq \tau_S$ and the Iou $\geq \tau_I$.   Across all three settings, $\mathcal{F}_\text{KR}$ Gaussian achieves the highest peak F1 and maintains it over a wide range of $\tau_S$ values, demonstrating robustness to the choice of score threshold. At the strictest localisation requirement ($\tau_I = 0.9$, Figure~\ref{fig:qa2_score-performance-main} (c), most baselines collapse to near-zero F1 while the NWKR methods retain meaningful performance. 
We ablate the choice of $\tau_I$ further in Appendix~\ref{app:exp-astro}.

\subsection{Score-based Filtering}  
We also consider filtering among \emph{all} signals, based on a normalized likelihood scan statistic $\Phi(\bfx)$ (eq \ref{eq:scan-score}) computed for each signal $\bfx$.  
For non-$F_d$, $\F_\KR$ baselines which we configure to find an interval $\hat I$, we compute $S(\hat I)$ using the prior art $\F_0$ model.  We use a scan statistic threshold $\tau_S = 0.3$ (results are stable in $\tau_S \in [0.25, 0.4]$) to mark as anomalous, otherwise mark not anomalous.  

\begin{table}[htbp]
\centering
\caption{Performance metrics on the large ALMA dataset ($N = 38{,}881$) at $\tau_S = 0.3$}
\label{tab:full_alma_0.3}
\begin{tabular}{lrrrrrrrr}
\toprule
Method & TP & FP & TN & FN & Accuracy & Precision & Recall & F1 \\
\midrule
$\mathcal{F}_0$ Mean        & 118 & 2{,}561 & 36{,}089 & 113 & 0.9312 & 0.0440 & 0.5108 & 0.0811 \\
$\mathcal{F}_1$ Poly (deg 1) & 154 & 3{,}980 & 34{,}670 &  77 & 0.8957 & 0.0373 & 0.6667 & 0.0706 \\
$\mathcal{F}_2$ Poly (deg 2) & 174 & 2{,}932 & 35{,}718 &  57 & 0.9231 & 0.0560 & 0.7532 & 0.1043 \\
\rowcolor{gray!15}
$\mathcal{F}_\text{KR}$ Gaussian & 231 & 191 & 38{,}459 & 0 & 0.9951 & 0.5474 & 1.0000 & 0.7075 \\
\rowcolor{gray!15}
$\mathcal{F}_\text{KR}$ Laplace  & 231 & 277 & 38{,}373 & 0 & 0.9929 & 0.4547 & 1.0000 & 0.6252 \\
\midrule
CAPA      &  31 &    554 & 38{,}096 & 200 & 0.9806 & 0.0530 & 0.1342 & 0.0760 \\
LRT       & 111 &    718 & 37{,}932 & 120 & 0.9784 & 0.1339 & 0.4805 & 0.2094 \\
\midrule
FS NWKR Gaussian & 164 &  69 & 38{,}581 & 67 & 0.9965 & 0.7039 & 0.7100 & 0.7069 \\
FS NWKR Laplace & 162 & 112 & 38{,}538 & 69 & 0.9953 & 0.5912 & 0.7013 & 0.6416 \\
CPD NWKR Gaussian & 113 &    838 & 37{,}812 & 118 & 0.9754 & 0.1188 & 0.4892 & 0.1912 \\
CPD NWKR Laplace  & 113 & 1{,}264 & 37{,}386 & 118 & 0.9645 & 0.0821 & 0.4892 & 0.1405 \\
\midrule
USAD   & 104 & 2{,}318 & 36{,}332 & 127 & 0.9371 & 0.0429 & 0.4502 & 0.0784 \\
TranAD & 110 & 2{,}293 & 36{,}357 & 121 & 0.9379 & 0.0458 & 0.4762 & 0.0835 \\
M2N2   & 112 & 2{,}401 & 36{,}249 & 119 & 0.9352 & 0.0446 & 0.4848 & 0.0816 \\
\bottomrule
\end{tabular}
\end{table}

For our methods, and the efficient and best alternatives methods, we show results in Table \ref{tab:full_alma_0.3}, using $\tau_S = 0.3$.  
Both $\F_\KR$ models have recall of $1.0$ and precision above $0.45$ (Gaussian above $0.54$) with almost no other method hits $0.53$ recall and $0.13$ precision. The expectation is our fixed-side variants of NWKR Gaussian and Laplace that achieve higher precision ($0.70$ and $0.59$) but at the cost of a substantially elevated false negative rate ($0.29$ and $0.30$), missing roughly a third of anomalies whose onset does not align with the start of the search region.
The improved false negative rate of our proposed methods is extremely important, since each missed case of platforming is likely to cause corruption in the resulting hyperspectral data cube, leading to a possibly unusable or deceiving scientific product.  
These limited false positives are also not very devastating towards corruption of the data cube products, as they mean slightly fewer observations are averaged over in the reconstruction.

In Table \ref{tab:ALMA-compare} we report these as false-negative (FN) and false-positive (FP) \emph{rates}, and compare to the accuracy of the two in-production alternatives.  The first option is from data reducer (these are people) who manually review the calibrations, and the second is the recently introduced heuristic.  
The false positive rates of these methods ($0.0053$ and $0.0013$) are comparable to ours ($0.0049$), with the heuristic better.  However, those most meaningful and time saving aspect is the false negative rate, where our method obtains $0$, while these alternatives are much higher at $0.49$ and $0.58$; this means they fail to identify about half of the platforming anomalies.  
In other words, with existing methods, about half of the platforming anomalies slip through this flag, where they can cause artifacts in the science products, whereas our proposed approach virtually eliminates this concern.



\begin{table}[h]
    \centering
    \caption{Comparison with estimated existing ALMA False Positive and False Negative rates}
    \begin{tabular}{lccccc}
    \toprule
        Method & Precision & Recall & F1 & FN rate & FP rate \\
        \midrule
        Data reducer (a person) & 0.3665 & 0.5108 & 0.4268 & 0.4892 & 0.0053 \\
        Existing ALMA heuristic & 0.7018 & 0.4233 & 0.5281 & 0.5767 & 0.0013 \\
        \rowcolor{gray!15}
        Our $\F_\KR$ Gaussian & 0.5474 & 1.0000 & 0.7075 & 0 & 0.0049 \\
    \bottomrule
    \end{tabular}
    \label{tab:ALMA-compare}
\end{table}

In Appendix \ref{app:exp-astro} we show effects of jointly filtering over IoU, ablate parameter choices, and discuss both some caveats and significance of this improvement on the ALMA application. A richer study integrating these aspects, further calibrating scores, and looking towards integrating this new methodology in future iterations of the ALMA pipeline is important future work.
In summary, these $\F_\KR$ models are a perfect for the ALMA challenge, and a general efficient new method.

\section{Discussion}
\label{sec:discussion}

We present a new scan statistic model for identifying and scoring interval anomalies in smoothly varying 1-dimensional signals.  They extend prior work by allowing a more complex regression model to be fit on the background data.  We develop and implement very efficient algorithms for computing these statistics.  
The method using a Nadaraya-Watson Kernel Regression model is shown especially efficient and effective.  A deep scientific application in detecting ``platforming" effects in radio telescope quality control highlights the usefulness of this approach.  

Our runtime for the NWKR model is $O(nwr)$.  While the linear in $n$ is necessary by just reading the data, and truncating exponentially decaying kernels (at $C r$), or just using bounded support kernels is standard; the $w$ factor might be as large as $n$.  That is, if we do not have a bound on the maximum window length, or need to set $w$ at $20\%$ of the length (e.g., $w = 0.2\times n$); then, the runtime is actually quadratic in $n$ at $O(nwr) = O(n^2 r)$.  

The models we use assume a smoothly varying background signal, with either $\F_d$ or $\F_\KR$ families.  If the background has other sharp changes in, it may be easy to confuse this structure for anomalies.


\begin{ack}

This work was supported by the National Science Foundation under Cooperative Agreement 2421782 and the Simons Foundation award MPS-AI-00010515
(NSF-Simons AI Institute for Cosmic Origins -- CosmicAI, https://www.cosmicai.org/).
We also thank John Horel for encouraging us to explore the MesoWest data.


\end{ack}



\bibliographystyle{plainnat}
\bibliography{bibliography}

\clearpage
\appendix

\section{Background on Statistical Models and Anomalies}
\label{sec:related_work}

The scan statistics framework detects localized departures from an assumed background model by maximizing a windowed test statistic over a family of candidate intervals or regions~\cite{abolhassani2021up,glaz2024handbook,kulldorff1999spatial,costa2009applications,mcfowland2013fast}. 
%
Early work on scan statistics formalized the paradigm of moving a window across a 1-dimensional signal: fix a window length, slide it across the data, record the most extreme data distribution observed, and quantify against a null model~\cite{naus1965line,naus1965twodim,wallenstein1980test}. Much of this work focused on events not recorded at regular intervals or locations.  Typically, much denser regions corresponded with more anomalous events as compared against a null model of a uniform distribution.  
While two-dimensional extensions existed earlier~\cite{naus1965twodim}, Kulldorffs's work~\cite{kulldorff1997spatial} and widely deployed software SatScan~\cite{satscan_software} extended many of these models into the spatial domain; moreover, with Nagarwalla~\cite{kulldorff1995spatial}, they formulated the statistic as a likelihood ratio test.  This allowed for clean and rigorous generalization to various baseline data models.  The closest model to our work is by \cite{huang2007spatial,kulldorff2009normal} (see also similar derivation by \cite{agarwal2006spatial}); it allows each point to have a continuous measurement value which is assumed drawn independently from a normal distribution.  The null model considers this normal distribution as constant across data; the alternative hypothesis allows this to differ inside and outside of the identified scan window.  While these works were designed for readings on irregularly distributed data observations, it applies naturally to regularly observed values as well.   

Beyond scan statistics, change point detection is a classic anomaly detection approach for 1-d signals like time series.  Classically CUSUM~\cite{page1954continuous} works by incrementally maintaining average statistics for online detection of changes~\cite{lorden1971procedures,basseville1993detection}.  
Retrospective time-series methods often rely on dynamic programming to identify outliers and level shifts~\cite{fox1972outliers,bai1998estimating,killick2012pelt}, and then has polynomial runtime with the exponent dependent on the number of change points.  
This led to likelihood-ratio testing for unknown change points between two distinct mean-centered sequences~\cite{quandt1960tests,hinkley1970inference,siegmund1995using}; these model the data generating process analogous to the normal model in scan statistics discussed above~\cite{agarwal2006spatial,kulldorff2009normal}, but with different considerations for data splits. 
In the change point detection literature this was extended to replace the mean-centered models with polynomials~\cite{bai1998estimating} and kernel smoothing~\cite{loader1996changepoint}.  
Refined software exists for certain versions from this class of models such as \textsf{ruptures}~\cite{truong2020selective} or BOCPD~\cite{adams2007bayesian}.  





There is also broader work in time series anomaly detection.  One line is discord mining based on, for instance, Matrix profile methods~\cite{yeh2016matrixprofile,keogh2005hotsax}.  These methods build a database of overlapping fix-length segments of a time series and identify discords which are not similar to any other segment; the most common library is STUMPY~\cite{law2019stumpy}.  
Another recent variant is CAPA~\cite{fisch2022capa} which looks for intervals which deviate from an assumed mean-0, variance-1 baseline of the data.  

The scan statistics paradigm and also extends to other structured domains~\cite{chitra2021quantifying} like graphs~\cite{wang2008spatial,sharpnack2013near,sharpnackchangepoint13}.

\subsection{Benchmarking against Other Methods}
We compare our proposed NWKR-based approach against five baselines, each adapted to produce a scoreable interval using the same anomaly scoring formula of $1 - (\text{SSE}_\text{in} + \text{SSE}_\text{out}) / \text{SSE}_\text{all}$ with a constant mean baseline, so that scores are directly comparable across methods.  

\emph{Kernel-based change-point detection (Ruptures KernelCPD)}~\cite{truong2020selective} fits a piecewise-constant $\ell_2$ cost model using dynamic programming with exactly two breakpoints, partitioning the spectrum into three segments. Each candidate segment is then scored with the scoring formula and the highest-scoring segment is returned as the predicted interval. Segments wider than $w$ are trimmed to $w$ channels around their center. The $\ell_2$ cost directly minimizes within-segment variance, matching the mean-model scoring definition, and runs in $O(n^2)$ time via the \texttt{ruptures} library~\cite{truong2020selective}.

The ``kernel" in the library name refers to the its use in a kernel two-sample test~\cite{gretton2012kernel,celisse2018new,arlot2019kernel}, which treats the intervals between change points as distributions of values.  That is, the ordering of the values within that interval does not play a role; for given changes points, it treats values as unordered.  
We use the linear kernel, which maps to the $F_0$ model.  An RBF kernel would not map well into our setting.

\emph{Bayesian Online Changepoint Detection (BOCPD)}~\cite{adams2007bayesian} maintains a posterior distribution over the current run length (the number of steps since the last changepoint) updated recursively at each observation using a Normal-Gamma conjugate prior and a constant hazard function $H = 1/\lambda$. A sudden drop in the maximum a posteriori run length signals a changepoint. For our synthetic data, pairs of consecutive such drops define the onset and return boundaries of the anomalous interval. BOCPD is fully probabilistic and models both boundaries through the same posterior update. The expected changepoint interval $\lambda$ is set to $n/2$ to encode a prior of approximately two changepoints per spectrum.

\emph{Likelihood Ratio Test (LRT)}~\cite{siegmund1995using} directly maximises the Gaussian log-likelihood ratio between a two-segment anomalous model (background $\sim \mathcal{N}(\mu_\text{out}, \sigma_\text{out}^2)$, anomaly $\sim \mathcal{N}(\mu_\text{in}, \sigma_\text{in}^2)$) and a one-segment null ($\text{all} \sim \mathcal{N}(\mu, \sigma^2)$) over all candidate window positions and widths up to $w$. Under the Gaussian model the log-likelihood ratio reduces to a closed-form expression involving only segment means, variances, and lengths, computable in $O(nw)$ time via prefix sums. This makes LRT the theoretically optimal parametric baseline for rectangular step detection under Gaussian noise~\cite{enikeeva2019high}, and requires no hyperparameters beyond the window cap $w$.

\emph{Collective and Point Anomaly detection (CAPA)}~\cite{fisch2022capa} minimises a penalised Gaussian negative log-likelihood cost over all possible anomalous segment placements using a dynamic program with PELT-style pruning. Robust estimates of the background mean and variance are obtained once via the median and scaled MAD, after which the dynamic program selects the best collective anomaly interval subject to a log-penalty that controls the false positive rate. CAPA is theoretically grounded for the anomalous changepoint model and achieves near-linear runtime in practice.

\emph{STUMPY (FLOSS)}~\cite{law2019stumpy} uses the matrix profile~\cite{yeh2016matrixprofile} to compute a corrected arc curve (CAC) over the signal, whose minima identify regime-change boundaries corresponding to the onset and return of the anomalous segment. The anomalous interval is constructed from the two deepest CAC minima and scored with the same anomaly formula as all other methods.

\emph{NWKR-based change-point detection (CPD NWKR)} applies the same Nadaraya-Watson kernel regression estimator as our scan statistic as a detrending step, then detects change points in the residuals. The signal is first fit with a truncated NWKR smoother using bandwidth derived from the physical spectral scale (same as for our methods), and the residuals are passed to a two-breakpoint dynamic program via the  \texttt{Ruptures} \texttt{KernelCPD}~\cite{truong2020selective} code. The three resulting segments are scored with the standard anomaly formula and the highest-scoring segment is returned as the predicted interval. This construction roughly follows the nonparametric regression change-point framework of~\citet{loader1996changepoint} and uses the same background family as our scan, differing only in the search formulation.

\emph{Fixed-Search NWKR scan (FS NWKR)} is a restricted variant of our full scan in which the left boundary of the search window is fixed at the first valid channel and only the right boundary is optimized over the full signal length. All other components of the pipeline are unchanged. This variant corresponds to a one-dimensional search in contrast to the two-dimensional search of the full scan.

\emph{Deep unsupervised detectors (USAD, TranAD, M2N2)} represent the class of reconstruction- and forecasting-based neural anomaly detectors. USAD~\cite{audibert2020usad} trains a shared encoder with two decoders adversarially: one decoder reconstructs the input while the other is trained to detect reconstructions that deviate from the training distribution, amplifying anomaly scores at inference. TranAD~\cite{tuli2022tranad} uses a transformer encoder with a self-conditioning mechanism in which a first reconstruction pass produces a per-channel focus score that concentrates attention in a second pass onto deviating regions, with both passes contributing to the final anomaly score. M2N2~\cite{kim2024model} is a test-time adaptation method that detrends the signal via exponential moving average and updates its internal model online on test instances judged consistent with the training distribution, allowing it to track slow distributional shifts during inference. All three produce a per-timestamp anomaly score. Since our evaluation metric requires a single predicted interval, we convert the score vector to an interval by finding the contiguous window that maximises total excess score mass above the signal mean, subject to the same window constraint applied to all other baselines.

\section{Synthetic Dataset Construction}
\label{sec:synthetic-data}

We construct multiple synthetic datasets to evaluate the ability of methods to localize compact anomalous intervals in one-dimensional signals with smooth, heterogeneous backgrounds. Each signal is defined on a uniform grid and is composed of a slowly varying quadratic trend, augmented by low-frequency correlated structure or AR(2) process based variation to mimic realistic spectral behavior. The baseline is further perturbed by white noise. A fixed fraction of signals (typically 5\%) contains a single localized anomaly, implemented as a rectangular step of controllable width and strength, added on top of the existing background and noise. This design allows precise control over anomaly characteristics while preserving realistic background complexity, and provides ground-truth intervals for evaluating localization performance.

\begin{figure}[htbp]
  \centering
  \begin{subfigure}[b]{0.48\linewidth}
    \centering
    \vfill
    \includegraphics[width=\linewidth]{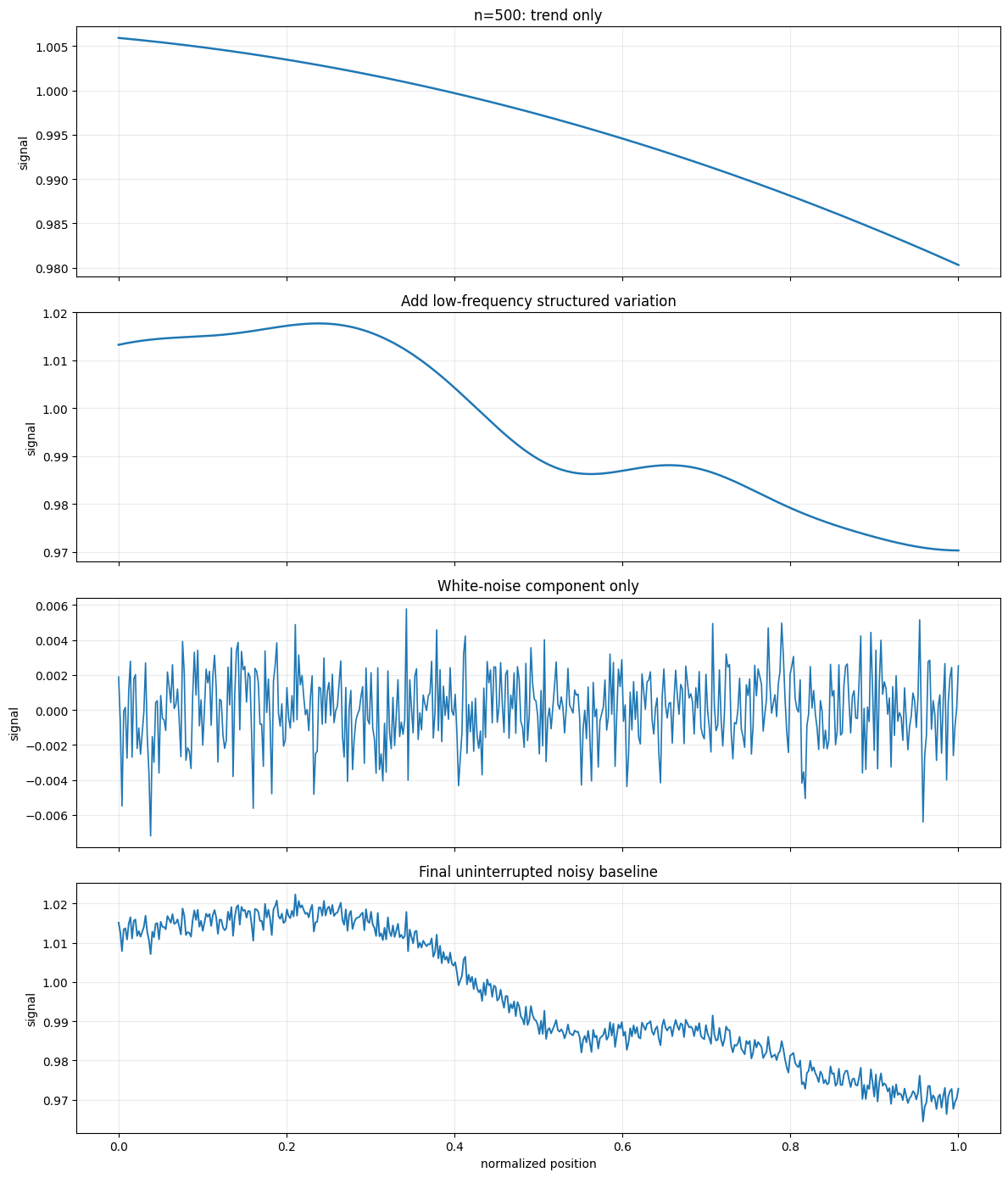}
    \caption{Synthetic components}
    \label{fig:signal_1}
    \vfill
  \end{subfigure}
  \hfill
  \begin{minipage}[b]{0.48\linewidth}
    \centering
    \begin{subfigure}[b]{\linewidth}
      \centering
      \includegraphics[width=\linewidth]{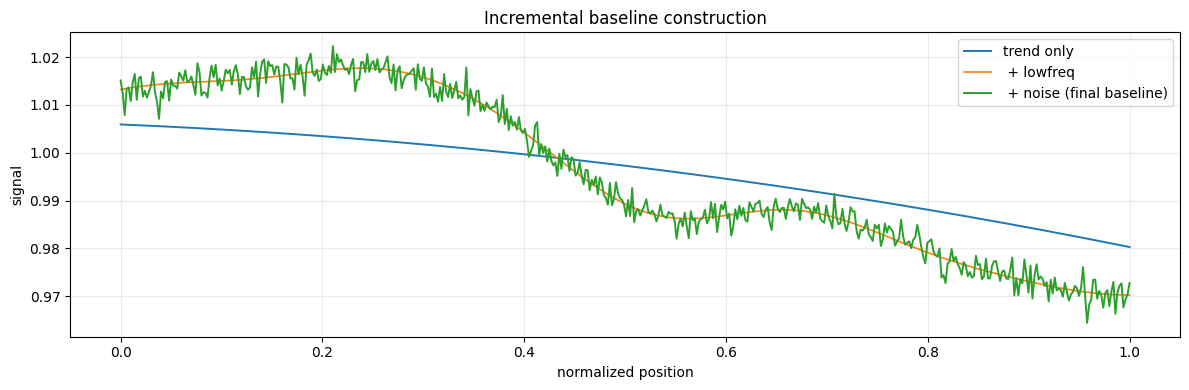}
      \caption{Construction step by step}
      \label{fig:signal_2}
    \end{subfigure}

    \vspace{0.3em}

    \begin{subfigure}[b]{\linewidth}
      \centering
      \includegraphics[width=\linewidth]{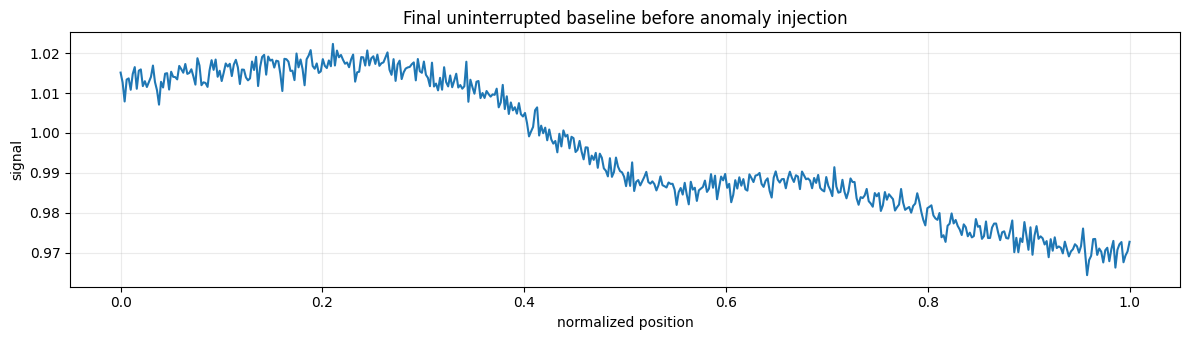}
      \caption{Baseline}
      \label{fig:signal_3}
    \end{subfigure}

    \vspace{0.3em}

    \begin{subfigure}[b]{\linewidth}
      \centering
      \includegraphics[width=\linewidth]{figures/signal_4.png}
      \caption{Anomaly injected into baseline}
      \label{fig:signal_4}
    \end{subfigure}
  \end{minipage}
  \caption{Synthetic dataset construction and anomaly injection}
  \label{fig:signals}
\end{figure}

\begin{figure}
  \centering
  \includegraphics[width=0.9\linewidth]{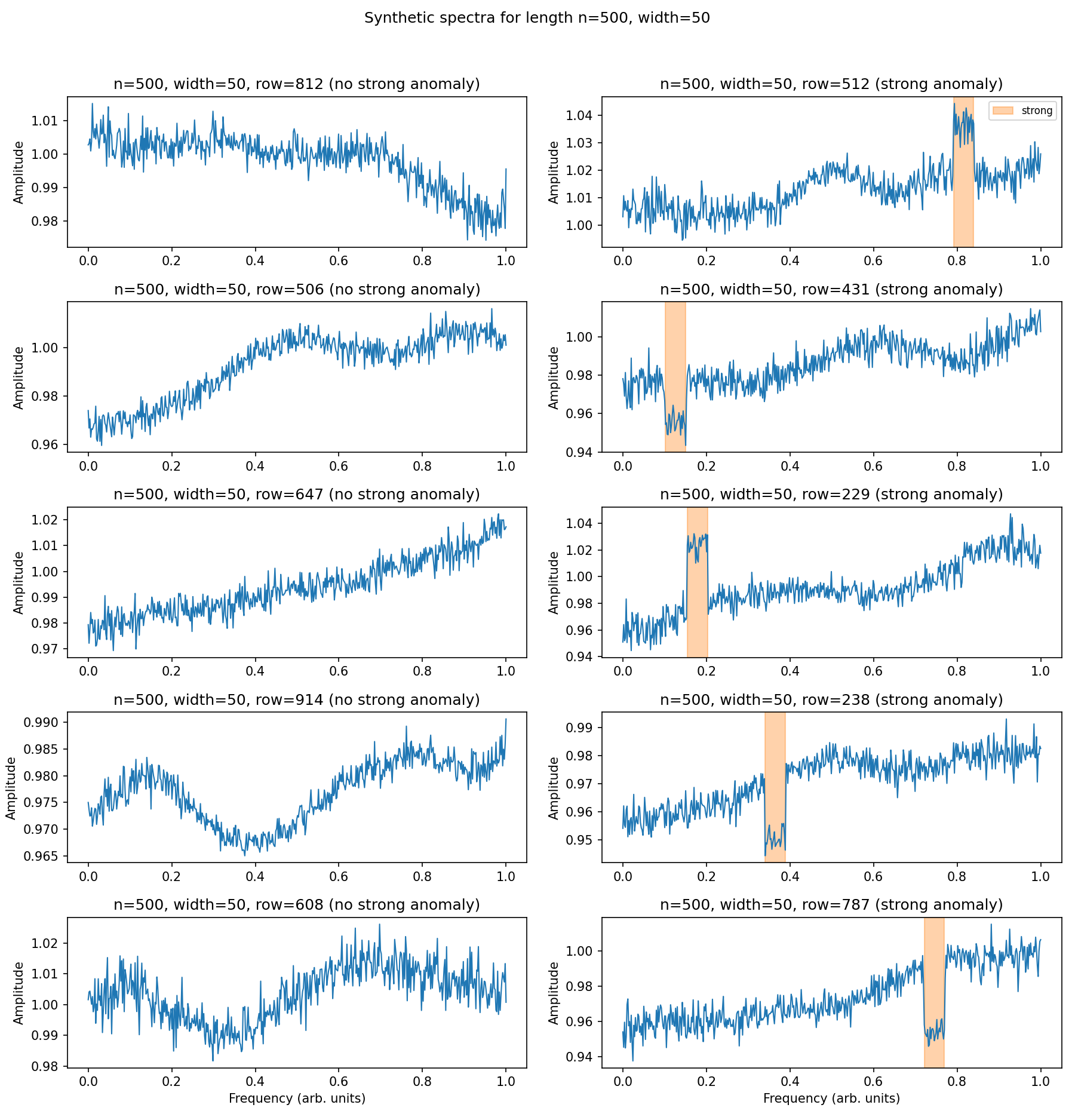}
  \caption{Different generated signals}
  \label{fig:signal_5}
\end{figure}

\begin{figure}
  \centering
  \includegraphics[width=0.9\linewidth]{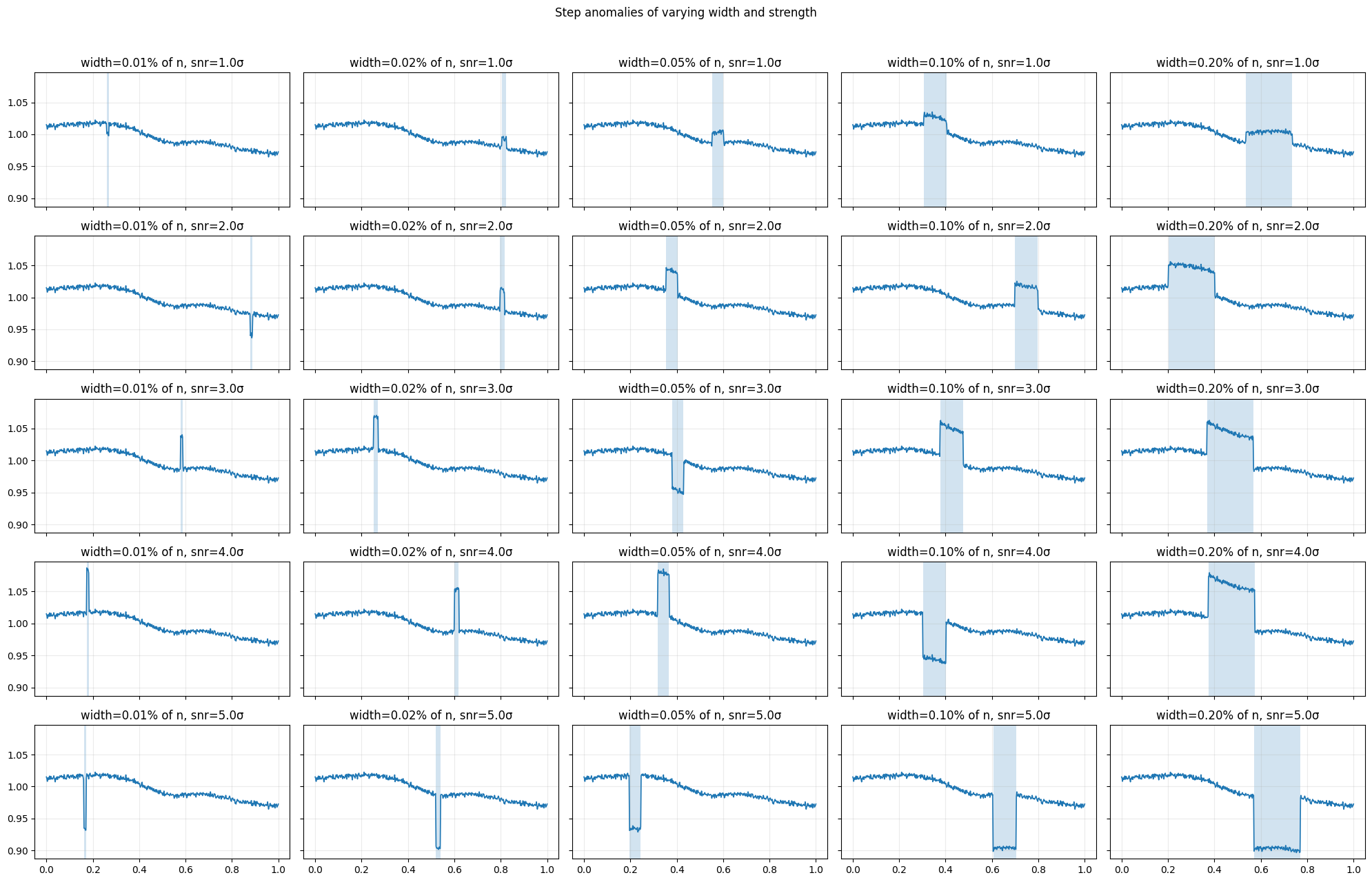}
  \caption{Different anomalies visualized in signal}
  \label{fig:signal_6}
\end{figure}

\begin{figure}[htbp]
  \centering
  \begin{subfigure}[b]{0.48\linewidth}
    \centering
    \includegraphics[width=\linewidth]{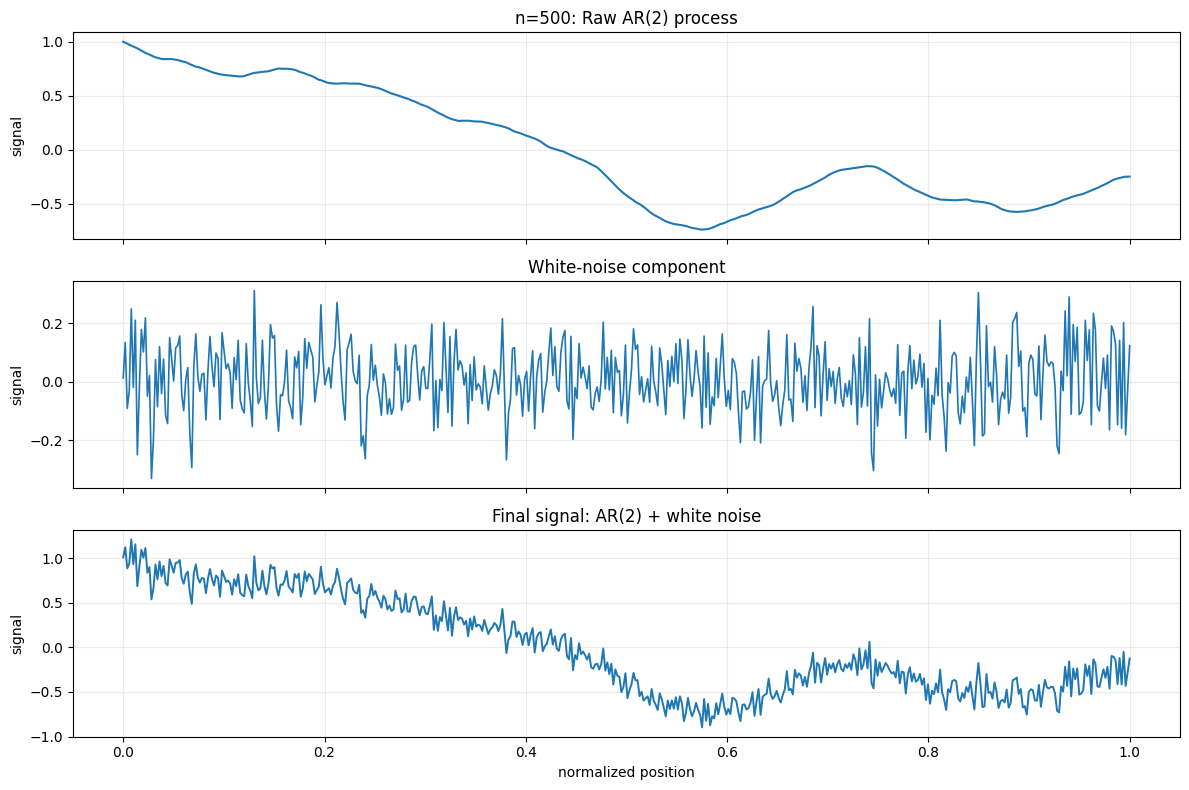}
    \caption{AR(2) components}
    \label{fig:signal_7}
    \vspace{4em}
  \end{subfigure}
  \hfill
  \begin{minipage}[b]{0.48\linewidth}
    \centering
    \begin{subfigure}[b]{\linewidth}
      \centering
      \includegraphics[width=\linewidth]{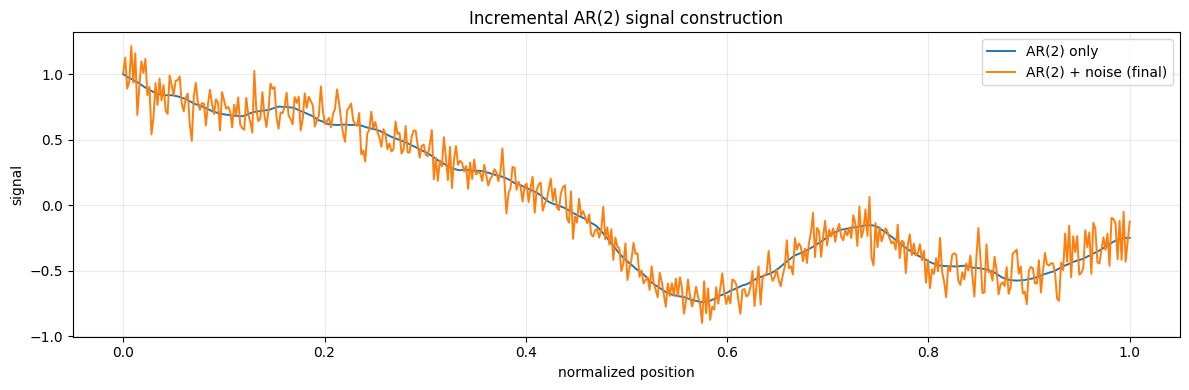}
      \caption{Construction step by step}
      \label{fig:signal_8}
    \end{subfigure}
    \vspace{0.3em}
    \begin{subfigure}[b]{\linewidth}
      \centering
      \includegraphics[width=\linewidth]{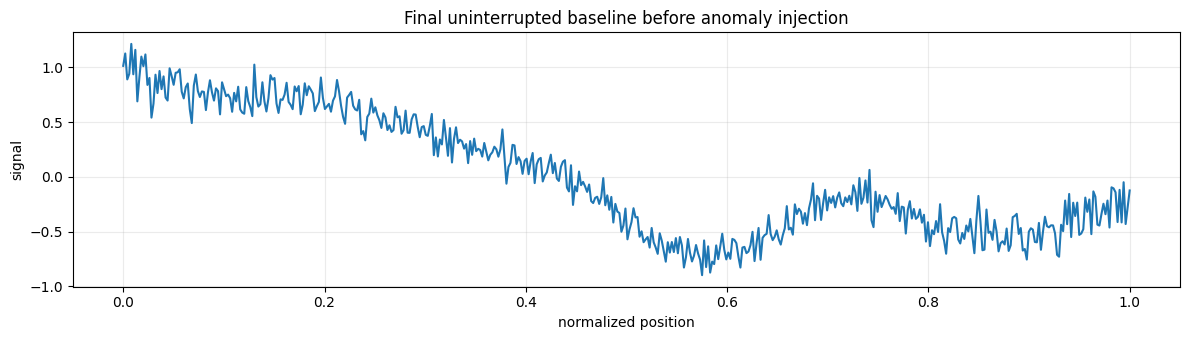}
      \caption{AR(2) baseline}
      \label{fig:signal_9}
    \end{subfigure}
    \vspace{0.3em}
    \begin{subfigure}[b]{\linewidth}
      \centering
      \includegraphics[width=\linewidth]{figures/signal_10.png}
      \caption{Anomaly injected into baseline}
      \label{fig:signal_10}
    \end{subfigure}
  \end{minipage}
  \caption{AR(2) dataset construction and anomaly injection}
  \label{fig:signals_ar2}
\end{figure}

\begin{figure}
  \centering
  \includegraphics[width=0.9\linewidth]{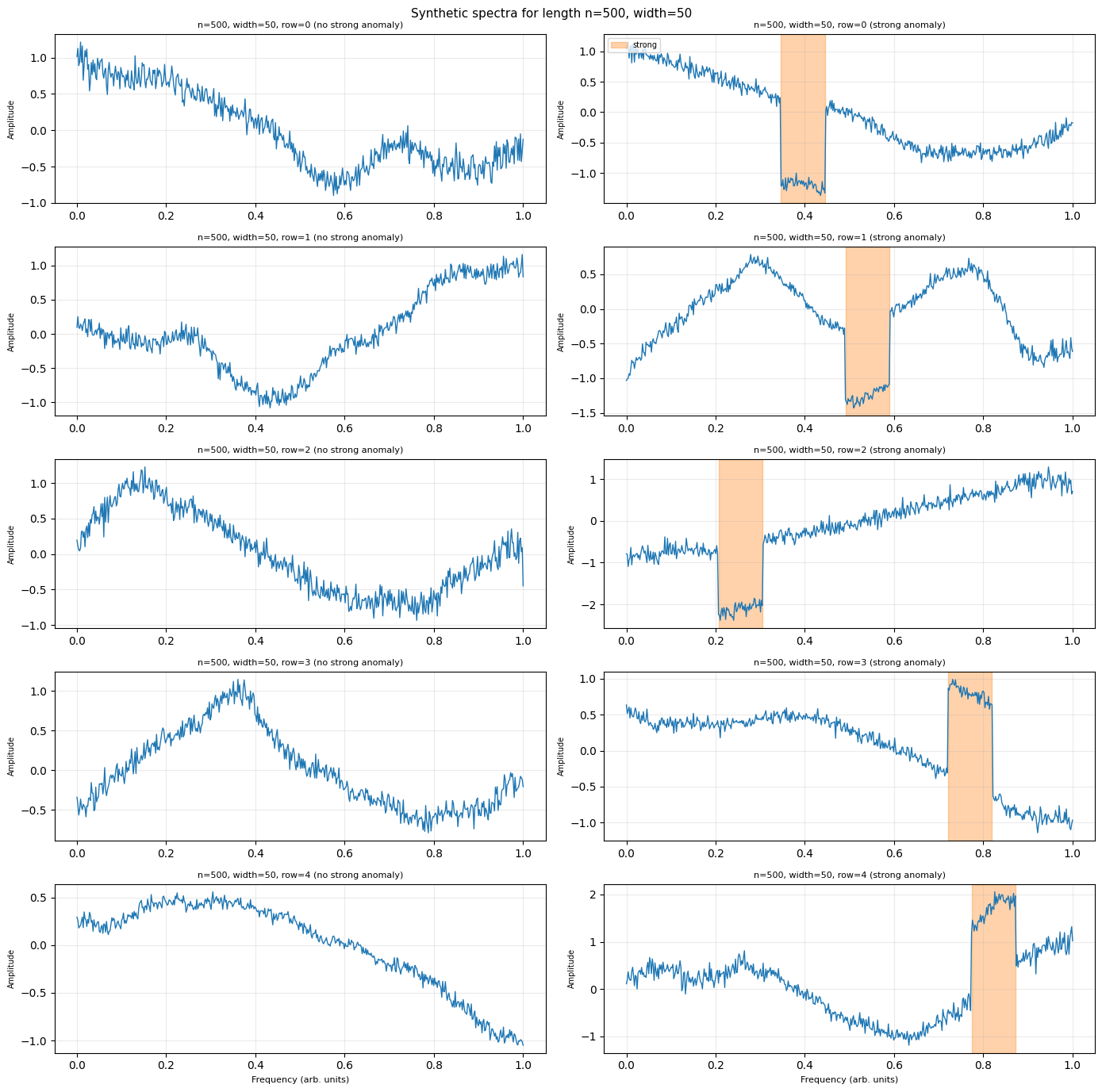}
  \caption{Different generated AR(2) signals}
  \label{fig:signal_11}
\end{figure}

\begin{figure}
  \centering
  \includegraphics[width=0.9\linewidth]{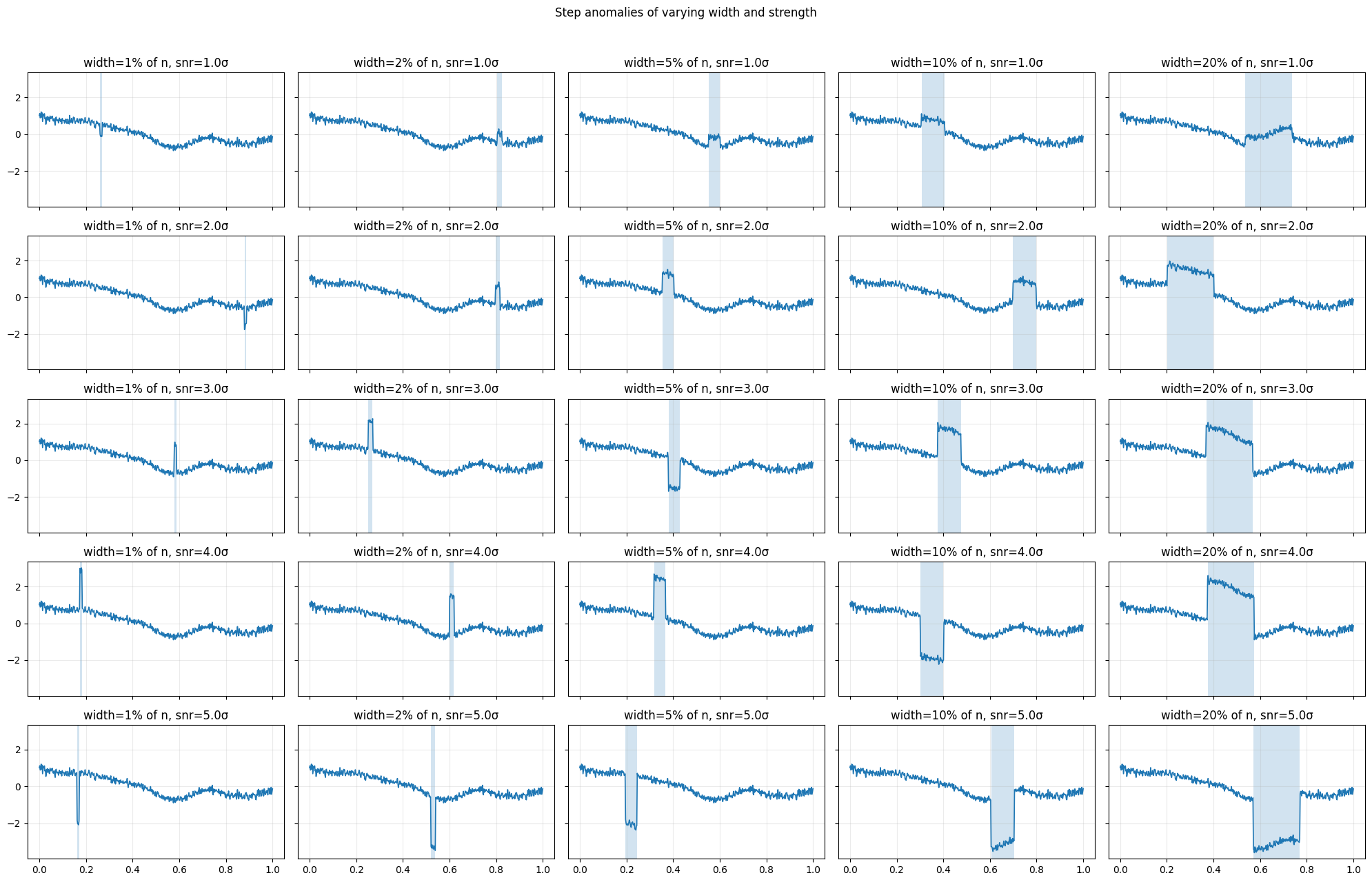}
  \caption{Different anomalies visualized in AR(2) signal}
  \label{fig:signal_12}
\end{figure}

\paragraph{Grouped dataset specification.}
Signals are generated in groups defined by a triple $(n, w, N)$, where $n$ is the signal length, $w$ is the scan-window cap used during evaluation, and $N$ is the number of signals per group. In the window-sensitivity experiments, we vary $w$ systematically across a length of $500$ while keeping $N=1000$. This produces regimes in which the scan window is substantially narrower than the anomaly width, approximately matched to it, or substantially broader. By explicitly varying $w$ across scales, we isolate the interaction between anomaly geometry and search breadth.

In separate experiments, anomaly strength is parameterized relative to the empirical standard deviation of the signal. This allows direct control of signal-to-noise ratio and clarifies detection behavior under varying contrast conditions.

Overall, the dataset construction enables controlled analysis of localization behavior as a function of (i) anomaly width, (ii) anomaly amplitude relative to noise, and (iii) scan-window configuration.

\subsection{Localization Results on Synthetic Data}
\label{app:localize}

\paragraph{Evaluation metrics.}
Localization performance is evaluated only on rows containing a ground-truth anomaly. For a predicted interval $\widehat{I}$ and ground-truth interval $I^\star$, we compute

\[
\text{recall} = \frac{| \widehat{I} \cap I^\star |}{| I^\star |}, 
\qquad
\text{precision} = \frac{| \widehat{I} \cap I^\star |}{| \widehat{I} |},
\]

and define the \emph{localization score} as their geometric mean,
\[
\text{Loc} = \sqrt{\text{recall} \times \text{precision}}.
\]

We report mean and median localization score together with mean recall, mean precision, and mean runtime per row.

\paragraph{Effect of scan-window cap $w$.}
The window-sensitivity experiments show that performance depends critically on the alignment between the anomaly width and the search breadth.

For small window caps (e.g., $w\leq 3$ at $n=100$), localization scores remain modest across all methods. Polynomial baselines often outperform NWKR in this regime because the scan window is too narrow to capture the full anomaly support, leading kernel smoothers to under-cover the anomalous region.

As $w$ increases to moderate values (e.g., $w\in\{5,6\}$ for $n=100$ and $w\in\{10,12\}$ for $n=200$), localization performance improves substantially. In this regime, Gaussian NWKR achieves perfect or near-perfect localization (mean Loc $\approx 1.0$), outperforming polynomial baselines. The Laplace NWKR variant typically matches Gaussian accuracy but at higher computational cost.

For excessively large window caps (e.g., $w\geq 20$ for $n=100$), the performance of simple baselines deteriorates due to over-extended detections, while the Gaussian and Laplace NWKR remain stable and achieve perfect localization. This demonstrates that kernel-based scanning is more robust to over-large search breadth than polynomial or mean models.

\paragraph{Effect of anomaly amplitude relative to noise.}
When anomaly amplitude is expressed as a multiple of the signal standard deviation, detection behavior becomes strongly SNR-dependent.

For weak anomalies (e.g., $1\sigma$), localization scores are low across all methods, with median scores often zero. This reflects intrinsic detectability limits rather than model deficiencies.

At moderate amplitudes (e.g., $2\sigma$), Gaussian NWKR exhibits sharp performance transitions, achieving near-perfect localization even for narrow anomalies, while polynomial methods improve more gradually.

For high amplitudes ($\geq 4\sigma$), all methods approach perfect localization, though NWKR typically achieves this regime at lower amplitude thresholds. Laplace NWKR occasionally matches Gaussian performance but does not consistently exceed it.

\section{ALMA Calibration Anomaly Evaluation with IoU Match}
\label{app:ALMA-IoU-exp}

\paragraph{IoU Evaluation protocol.}
In this section, we further require an bandpass calibration anomaly detection to be counted as found if the predicted interval achieved an intersection-over-union (IoU; aka Jaccard Similarity) of at least $\tau_I$ (e.g. $\tau_I = 0.75$) with the ground truth interval, where
\[
    \text{IoU}([a_1, b_1],\, [a_2, b_2])
    = \frac{\max(0,\, \min(b_1, b_2) - \max(a_1, a_2) + 1)}
           {\max(b_1, b_2) - \min(a_1, a_2) + 1}.
\]

But not every signal has an anomaly, so we also considered detecting among a mix of anomalous and non-anomalous signals.  For this setting, to predict something is anomalous, we need a score threshold $\tau_S$  (e.g., $\tau_S = 0.3$) which is a minimal value of the normalized scan statistic $\Phi(\bfx)$ to predict as non-trivial.  
In this setting, for a marked anomalous interval to be predicted correctly (a True-Positive), the signal must have $\Phi(\bfx) > \tau_S$ and then the identified interval $\hat I$ must have IoU at least $\tau_I$.  
For a Negative prediction it only needs score below $\tau_S$ (for a True-Negative), since there is no ground-truth interval to consider.  But we can still have a False-Negative for a anomalous interval if its score is above $\tau_S$, but its IoU with the ground-truth interval is below $\tau_I$.  
For methods which do not compute a (normalized) generalized log-likelihood ratio (as we do for $\F_d$ and $\F_\KR$ in Section \ref{sec:scan}), we can still use their predicted $\hat I$ for IoU, and then use the baseline $\F_0$ model to produce $S(\hat I)$ as a score for comparison.

\subsection{Experimental Results}
\label{app:exp-astro}

\paragraph{Ablation of Score Threshold.}
We first next revisiting fixing the IoU threshold $\tau_I$  and showing the ability to filter based on score at that threshold.  
Table \ref{tab:f1_iou} (a) fixes $\tau_S = 0.3$, and shows all increments of $\tau_I$ from $0.5$ to $1$ in increments of $0.05$.  

As a result it is again clear that our $\F_\KR$ is the best choice for this data, and that there is a wide range of score thresholds $\tau_S$ where it is effective.  Since the method is robust to this choice, we do not perform a detailed test/train split evaluation.  It will depend more on a user preference for favoring precision over recall in this sense.




\begin{table}[htbp]
\centering
\caption{F1 Score across IoU Thresholds $\tau_I$ for each method, with $\tau_S = 0.3$}
\label{tab:f1_iou}
\resizebox{\textwidth}{!}{%
\begin{tabular}{lccccccccccc}
\toprule
Method & 0.50 & 0.55 & 0.60 & 0.65 & 0.70 & 0.75 & 0.80 & 0.85 & 0.90 & 0.95 & 1.00 \\
\midrule
$\F_0$ Mean & 0.6000 & 0.5920 & 0.5920 & 0.5797 & 0.5797 & 0.5756 & 0.5546 & 0.5460 & 0.5106 & 0.2624 & 0.2182 \\
$\F_1$ Poly (deg 1) & 0.6649 & 0.6545 & 0.6510 & 0.6223 & 0.6150 & 0.6075 & 0.5611 & 0.5408 & 0.4854 & 0.2381 & 0.1888 \\
$\F_2$ Poly (deg 2) & 0.7172 & 0.7074 & 0.7041 & 0.6702 & 0.6596 & 0.6489 & 0.5810 & 0.5444 & 0.4699 & 0.1922 & 0.1727 \\
\rowcolor{gray!15}
$\F_\KR$ Gaussian & 0.9902 & 0.9902 & 0.9902 & 0.9902 & 0.9902 & 0.9892 & 0.9892 & 0.9892 & 0.9892 & 0.9870 & 0.9870 \\
\rowcolor{gray!15}
$\F_\KR$ Laplace & 0.9825 & 0.9825 & 0.9780 & 0.9689 & 0.9573 & 0.9526 & 0.9333 & 0.9057 & 0.8873 & 0.8627 & 0.8342 \\
FS NWKR Gaussian & 0.0339 & 0.0171 & 0.0171 & 0.0171 & 0.0171 & 0.0171 & 0.0171 & 0.0171 & 0.0171 & 0.0000 & 0.0000 \\
FS NWKR Laplace & 0.0339 & 0.0171 & 0.0171 & 0.0171 & 0.0171 & 0.0171 & 0.0086 & 0.0086 & 0.0000 & 0.0000 & 0.0000 \\
CPD NWKR Gaussian & 0.6167 & 0.6087 & 0.6087 & 0.6087 & 0.6047 & 0.6006 & 0.5799 & 0.5366 & 0.4810 & 0.2979 & 0.1818 \\
CPD NWKR Laplace & 0.6132 & 0.6092 & 0.6092 & 0.6052 & 0.6052 & 0.6052 & 0.5971 & 0.5765 & 0.5422 & 0.3425 & 0.2590 \\
\midrule
CAPA & 0.2180 & 0.2045 & 0.2045 & 0.1556 & 0.1265 & 0.1116 & 0.0887 & 0.0887 & 0.0574 & 0.0574 & 0.0574 \\
BOCPD & 0.4752 & 0.4752 & 0.4752 & 0.4752 & 0.4702 & 0.4497 & 0.4339 & 0.4124 & 0.3261 & 0.0424 & 0.0340 \\
KernelCPD & 0.5757 & 0.5714 & 0.5714 & 0.5629 & 0.5629 & 0.5629 & 0.5455 & 0.5321 & 0.4953 & 0.2288 & 0.1955 \\
LRT & 0.6217 & 0.6136 & 0.6136 & 0.6012 & 0.5970 & 0.5928 & 0.5583 & 0.5358 & 0.4936 & 0.2197 & 0.1853 \\
STUMPY & 0.5364 & 0.5322 & 0.5235 & 0.4925 & 0.4880 & 0.4648 & 0.3910 & 0.2808 & 0.0391 & 0.0079 & 0.0079 \\
\midrule
TSBAD M2N2 & 0.5954 & 0.5831 & 0.5831 & 0.5664 & 0.5621 & 0.5536 & 0.5046 & 0.4907 & 0.4221 & 0.1308 & 0.0941 \\
TSBAD TranAD & 0.6189 & 0.5906 & 0.5906 & 0.5740 & 0.5740 & 0.5655 & 0.5215 & 0.4937 & 0.3661 & 0.1024 & 0.0720 \\
TSBAD USAD & 0.5831 & 0.5789 & 0.5789 & 0.5579 & 0.5536 & 0.5493 & 0.4812 & 0.4423 & 0.1379 & 0.0637 & 0.0163 \\
\bottomrule
\end{tabular}}
\end{table}


\subsection{Detection with Score and Interval Overlap on Balanced Subset}
\label{app:IoU-ALMA-balanced}
We next show in Table \ref{tab:metrics_075} the results on $N'= 500$ signals, with $269$ random non-anomalous signals added to the $231$ marked anomalies, and $\tau_I = 0.75$ and $\tau_S = 0.3$. Here we count TP as any anomaly boundary having IoU greater than $\tau_I$ and score being greater than $\tau_S$. And any positive label breaking these conditions are considered as FP. Also, we count TN as any predicted score for negative labels less than $\tau_S$, and any predicted score breaking this condition for negative labels are considered as FN. From here on, we calculate accuracy as $\frac{TP + TN}{500}$, precision as $\frac{TP}{TP + FP}$, and recall as $\frac{TP}{TP + FN}$.
For this setting both $\F_\KR$ models have precision above $0.99$ and both accuracy and F1 score above $0.95$ (Gaussian above $0.98$).  While other methods can have high precision, no other approach has F1 above $0.72$ or accuracy above $0.75$.  

A few methods are faster than ours, as shown on average time (in ms) per signal.  CAPA is much faster, but with F1 around $0.1$.  Also LRT (which has an $\F_0$ like model for change points) is about a factor 5 faster than $\F_\KR$, but like our $\F_0$ (which is factor 10 faster) has reasonable, but worse performance; it fits the background as constant, which is not a good fit for this data.  All other methods are slower than our $\F_\KR$ models.  

\begin{table}[htbp]
\centering
\caption{Performance metrics on sampled balanced dataset ($N = 500$) at $\tau_I = 0.75$, $\tau_S = 0.3$}
\label{tab:metrics_075}
\resizebox{\textwidth}{!}{%
\begin{tabular}{lcccccccc}
\toprule
Method & TP & FP & FN & Accuracy & Precision & Recall & F1 & Runtime (ms) \\
\midrule
$\F_0$ Mean & 99 & 14 & 132 & 0.708 & 0.8761 & 0.4286 & 0.5756 & 680.4 \\
$\F_1$ Poly (deg 1) & 113 & 28 & 118 & 0.708 & 0.8014 & 0.4892 & 0.6075 & 15,392.7 \\
$\F_2$ Poly (deg 2) & 122 & 23 & 109 & 0.736 & 0.8414 & 0.5281 & 0.6489 & 18,561.5 \\
\rowcolor{gray!15}
$\F_\KR$ Gaussian & 228 & 2 & 3 & 0.990 & 0.9913 & 0.9870 & 0.9892 & 223.0 \\
\rowcolor{gray!15}
$\F_\KR$ Laplace & 211 & 1 & 20 & 0.958 & 0.9953 & 0.9134 & 0.9526 & 214.1 \\
FS NWKR Gaussian & 2 & 1 & 229 & 0.540 & 0.6667 & 0.0087 & 0.0171 & 143.1 \\
FS NWKR Laplace & 2 & 1 & 229 & 0.540 & 0.6667 & 0.0087 & 0.0171 & 139.4 \\
CPD NWKR Gaussian & 103 & 9 & 128 & 0.726 & 0.9196 & 0.4459 & 0.6006 & 24.0 \\
CPD NWKR Laplace & 105 & 11 & 126 & 0.726 & 0.9052 & 0.4545 & 0.6052 & 24.9 \\
\midrule
CAPA & 14 & 6 & 217 & 0.554 & 0.7000 & 0.0606 & 0.1116 & 15.4 \\
BOCPD & 67 & 0 & 164 & 0.672 & 1.0000 & 0.2900 & 0.4497 & 25,869.1 \\
KernelCPD & 94 & 9 & 137 & 0.708 & 0.9126 & 0.4069 & 0.5629 & 24.9 \\
LRT & 99 & 4 & 132 & 0.728 & 0.9612 & 0.4286 & 0.5928 & 1,219.4 \\
STUMPY & 76 & 20 & 155 & 0.650 & 0.7917 & 0.3290 & 0.4648 & 47.0 \\
\midrule
TSBAD M2N2 & 93 & 12 & 138 & 0.700 & 0.8857 & 0.4026 & 0.5536 & 260.9 \\
TSBAD TranAD & 95 & 10 & 136 & 0.708 & 0.9048 & 0.4113 & 0.5655 & 281.0 \\
TSBAD USAD & 92 & 12 & 139 & 0.698 & 0.8846 & 0.3983 & 0.5493 & 231.8 \\
\bottomrule
\end{tabular}}
\end{table}


Then Table \ref{tab:summary} shows for three different $\tau_I$ thresholds ($0.5$, $0.75$, and $0.9$) the maximum operating point with respect to the score threshold $\tau_S$, denoted $\tau_S^*$ (checked in $0.05$ increments).  Notably, this is $\tau_S^* = 0.35$ for our $\F_\KR$ Gaussian, and this is the maximum Precision, Recall, and F1 for each $\tau_I$ among all options we compare to.  The F1 score is always at least $0.99$.  The next best is $\F_\KR$ Laplace, and then $\F_2$ before the first technique not developed in this paper LRT.  However, its F1 score is consistently at least $0.3$ below that for our $\F_\KR$ Gaussian; for the high IoU threshold of $0.9$, its F1 score is close to $0.5$.

\begin{table}[htbp]
\centering
\caption{Summary of best F1 operating points per method across IoU thresholds. For each method, the optimal score threshold ($\tau_S^*$) is selected to maximize F1. Precision, Recall, and F1 are reported at that threshold.}
\label{tab:summary}
\resizebox{\textwidth}{!}{%
\begin{tabular}{l c c c c c c c c c c c c}
\toprule
 & \multicolumn{4}{c}{IoU $\geq$ 0.50} & \multicolumn{4}{c}{IoU $\geq$ 0.75} & \multicolumn{4}{c}{IoU $\geq$ 0.90} \\
\cmidrule(lr){2-5} \cmidrule(lr){6-9} \cmidrule(lr){10-13}
Method & $\tau_S^*$ & Prec. & Rec. & F1 & $\tau_S^*$ & Prec. & Rec. & F1 & $\tau_S^*$ & Prec. & Rec. & F1 \\
\midrule
$\mathcal{F}_0$ Mean        & 0.20 & 0.690 & 0.558 & 0.617 & 0.25 & 0.766 & 0.480 & 0.590 & 0.25 & 0.728 & 0.394 & 0.511 \\
$\mathcal{F}_1$ Poly (deg 1) & 0.30 & 0.822 & 0.558 & 0.665 & 0.30 & 0.801 & 0.489 & 0.608 & 0.30 & 0.748 & 0.359 & 0.485 \\
$\mathcal{F}_2$ Poly (deg 2) & 0.25 & 0.833 & 0.649 & 0.730 & 0.25 & 0.808 & 0.546 & 0.651 & 0.25 & 0.730 & 0.351 & 0.474 \\
\rowcolor{gray!15}
$\mathcal{F}_\text{KR}$ Gaussian & 0.35 & \textbf{0.996} & \textbf{0.987} & \textbf{0.991} & 0.35 & \textbf{0.996} & \textbf{0.987} & \textbf{0.991} & 0.35 & \textbf{0.996} & \textbf{0.987} & \textbf{0.991} \\
\rowcolor{gray!15}
$\mathcal{F}_\text{KR}$ Laplace & 0.30 & 0.996 & 0.957 & 0.976 & 0.30 & 0.995 & 0.913 & 0.953 & 0.30 & 0.995 & 0.801 & 0.887 \\
\midrule
CAPA      & 0.10 & 0.684 & 0.225 & 0.339 & 0.10 & 0.529 & 0.117 & 0.192 & 0.25 & 0.529 & 0.039 & 0.073 \\
LRT       & 0.15 & 0.831 & 0.597 & 0.695 & 0.15 & 0.819 & 0.550 & 0.658 & 0.25 & 0.924 & 0.368 & 0.526 \\
BOCPD     & 0.15 & 0.946 & 0.377 & 0.539 & 0.15 & 0.938 & 0.325 & 0.482 & 0.25 & 0.979 & 0.199 & 0.331 \\
KernelCPD & 0.20 & 0.830 & 0.463 & 0.594 & 0.25 & 0.863 & 0.437 & 0.580 & 0.25 & 0.843 & 0.372 & 0.516 \\
STUMPY    & 0.20 & 0.630 & 0.524 & 0.572 & 0.25 & 0.720 & 0.368 & 0.487 & 0.35 & 0.417 & 0.022 & 0.041 \\
\midrule
FS NWKR Gaussian & 0.20 & 0.800 & 0.017 & 0.034 & 0.20 & 0.667 & 0.009 & 0.017 & 0.20 & 0.667 & 0.009 & 0.017 \\
FS NWKR Laplace  & 0.20 & 0.800 & 0.017 & 0.034 & 0.20 & 0.667 & 0.009 & 0.017 & 0.00 & 0.000 & 0.000 & 0.000 \\
CPD NWKR Gaussian & 0.20 & 0.841 & 0.550 & 0.665 & 0.25 & 0.911 & 0.485 & 0.633 & 0.25 & 0.883 & 0.359 & 0.511 \\
CPD NWKR Laplace  & 0.20 & 0.804 & 0.550 & 0.653 & 0.20 & 0.801 & 0.541 & 0.646 & 0.25 & 0.857 & 0.442 & 0.583 \\
\midrule
USAD   & 0.20 & 0.705 & 0.558 & 0.623 & 0.25 & 0.774 & 0.446 & 0.566 & 0.25 & 0.434 & 0.100 & 0.162 \\
TranAD & 0.20 & 0.720 & 0.580 & 0.643 & 0.20 & 0.689 & 0.498 & 0.578 & 0.25 & 0.682 & 0.260 & 0.376 \\
M2N2   & 0.20 & 0.699 & 0.563 & 0.624 & 0.25 & 0.770 & 0.450 & 0.568 & 0.25 & 0.699 & 0.312 & 0.431 \\
\bottomrule
\end{tabular}}
\end{table}

\clearpage
\subsection{Detection with Score and Interval Overlap on Full ALMA}
\label{app:ALMA-IoU-full}

Table~\ref{tab:full_alma} reports the detection performance on the large ALMA dataset ($N = 38{,}881$ rows, of which $231$ are positive) at and $\tau_S = 0.3$. Due to the extreme class imbalance ($0.6\%$ positive rate), as expected, precision drops substantially for all methods compared to the balanced setting in Table~\ref{tab:metrics_075}. Nevertheless, $\F_\KR$ Gaussian achieves the highest precision ($0.544$) and recall ($0.995$) among all methods, with only $191$ false positives out of $38{,}650$ negative signals and missing just $3$ of the $231$ true anomalies. $\F_\KR$ Laplace follows with comparable recall ($0.992$) but lower precision ($0.432$) due to additional false positives. The parametric models $\F_1$ and $\F_2$ achieve moderate recall ($0.489$ and $0.528$, respectively) but yield thousands of false positives, driving their precision below $0.04$. Among the efficient baselines, LRT achieves the best balance with precision $0.121$ and recall $0.429$, while CAPA produces very few false positives ($554$) but detects only $14$ of the $231$ anomalies.

The NWKR-CPD baselines achieve recall of approximately $0.45$ for both kernel families, roughly matching LRT, while the three deep detectors (USAD, TranAD, and M2N2) reach recall between $0.40$ and $0.41$ at precision below $0.04$, comparable to the parametric baselines despite substantially higher false positive counts.    

The fixed-left variants tell a particularly informative story at the $\tau_I = 0.75$ threshold. Whereas at $\tau_I = 0$ (Table~\ref{tab:full_alma_0.3}) these methods achieved recall of $0.710$ and $0.701$, imposing the IoU constraint collapses recall to just $0.009$ for both variants. This confirms that while the fixed-left scan correctly identifies many anomalous signals, it almost never recovers the correct interval: the left boundary, anchored at the start of the search region, rarely coincides with the true anomaly onset. This isolates the contribution of the two-dimensional interval search in the full scan, which is responsible for the precise localisation underlying the high IoU performance.

Notably, $\F_\KR$ Gaussian also offers competitive runtime, with a median of $45$\,ms per signal, faster than $\F_0$ Mean ($95$\,ms), and orders of magnitude faster than the polynomial models ($\F_1$ at $2{,}214$\,ms, $\F_2$ at $2{,}697$\,ms).  Note that due to the scale, we used parallelization for this experiment, and the runtime measurements may be noisy and not reflect the precise expected runtime values under ideal conditions.


\begin{table}[h]
\centering
\caption{Performance metrics on the full ALMA dataset ($N = 38{,}881$) at $\tau_I = 0.75$, $\tau_S = 0.3$}
\label{tab:full_alma}
\resizebox{\textwidth}{!}{%
\begin{tabular}{lrrrrrrrrrr}
\toprule
Method & TP & FP & TN & FN & Accuracy & Precision & Recall & F1 & Mean (ms) & Median (ms) \\
\midrule
$\mathcal{F}_0$ Mean        &  99 & 2{,}561 & 36{,}089 & 132 & 0.9307 & 0.0372 & 0.4286 & 0.0685 &   394.3 &    94.9 \\
$\mathcal{F}_1$ Poly (deg 1) & 113 & 3{,}980 & 34{,}670 & 118 & 0.8946 & 0.0276 & 0.4892 & 0.0523 & 9{,}200.6 & 2{,}213.6 \\
$\mathcal{F}_2$ Poly (deg 2) & 122 & 2{,}932 & 35{,}718 & 109 & 0.9218 & 0.0399 & 0.5281 & 0.0743 & 11{,}181.2 & 2{,}697.3 \\
\rowcolor{gray!15}
$\mathcal{F}_\text{KR}$ Gaussian & 228 & 191 & 38{,}459 &  3 & 0.9950 & 0.5442 & 0.9870 & 0.7015 & 190.3 & 45.2 \\
\rowcolor{gray!15}
$\mathcal{F}_\text{KR}$ Laplace  & 211 & 277 & 38{,}373 & 20 & 0.9924 & 0.4324 & 0.9134 & 0.5869 & 185.5 & 45.1 \\
\midrule
CAPA      &  14 &    554 & 38{,}096 & 217 & 0.9802 & 0.0246 & 0.0606 & 0.0350 &   8.6 &   5.4 \\
LRT       &  99 &    718 & 37{,}932 & 132 & 0.9781 & 0.1212 & 0.4286 & 0.1889 & 756.3 & 182.6 \\
\midrule
FS NWKR Gaussian &   2 &  69 & 38{,}581 & 229 & 0.9923 & 0.0282 & 0.0087 & 0.0132 & 162.9 &  6.6 \\
FS Laplace  &   2 & 112 & 38{,}538 & 229 & 0.9912 & 0.0175 & 0.0087 & 0.0116 & 161.1 &  6.3 \\
CPD NWKR Gaussian & 103 &    838 & 37{,}812 & 128 & 0.9752 & 0.1095 & 0.4459 & 0.1758 & 26.6 & 7.4 \\
CPD NWKR Laplace  & 105 & 1{,}264 & 37{,}386 & 126 & 0.9642 & 0.0767 & 0.4545 & 0.1313 & 26.7 & 7.4 \\
\midrule
USAD   &  92 & 2{,}318 & 36{,}332 & 139 & 0.9368 & 0.0382 & 0.3983 & 0.0697 & 237.8 & 175.7 \\
TranAD &  95 & 2{,}293 & 36{,}357 & 136 & 0.9375 & 0.0398 & 0.4113 & 0.0725 & 281.5 & 210.8 \\
M2N2   &  93 & 2{,}401 & 36{,}249 & 138 & 0.9347 & 0.0373 & 0.4026 & 0.0683 & 279.8 & 309.8 \\
\bottomrule
\end{tabular}}
\end{table}

Furthermore, we searched over the score threshold $\tau_S$ for $\F_\KR$ Gaussian in Table \ref{tab:fine_score_grid}, and even at $\tau_S = 0.1$ it does not identify all anomalies.  The issue with the remaining $3$ is not the score, but an interval mismatch.  One has IoU at $0.7$, and the other two have multiple instrumental issues, and our method identified ones different from those marked by the experts.  
For this reason, in the main paper Table \ref{tab:full_alma_0.3} we show results which only filter by score threshold at $\tau_S = 0.3$.  Moreover, if the signal is marked as having platforming, the standard procedure is to discard it entirely, and pinpointing the interval is not issue of the central importance.

\begin{table}[ht]
\centering
\caption{Fine-grained score threshold grid search at $\tau_I = 0.75$ for $\F_\KR$ Gaussian on the full ALMA dataset ($N = 38{,}881$).}
\label{tab:fine_score_grid}
\begin{tabular}{c rrrrrrrrr}
\toprule
$\tau_S$ & TP & FP & FN & Accuracy & Precision & Recall & F1 & FPR & FNR \\
\midrule
0.10 & 228 & 6,166 & 3 & 0.8413 & 0.0357 & 0.9870 & 0.0688 & 0.1595 & 0.0130 \\
0.11 & 228 & 4,544 & 3 & 0.8831 & 0.0478 & 0.9870 & 0.0911 & 0.1176 & 0.0130 \\
0.12 & 228 & 3,365 & 3 & 0.9134 & 0.0635 & 0.9870 & 0.1192 & 0.0871 & 0.0130 \\
0.13 & 228 & 2,542 & 3 & 0.9345 & 0.0823 & 0.9870 & 0.1519 & 0.0658 & 0.0130 \\
0.14 & 228 & 1,957 & 3 & 0.9496 & 0.1043 & 0.9870 & 0.1887 & 0.0506 & 0.0130 \\
0.15 & 228 & 1,507 & 3 & 0.9612 & 0.1314 & 0.9870 & 0.2319 & 0.0390 & 0.0130 \\
0.16 & 228 & 1,212 & 3 & 0.9688 & 0.1583 & 0.9870 & 0.2729 & 0.0314 & 0.0130 \\
0.17 & 228 & 995 & 3 & 0.9743 & 0.1864 & 0.9870 & 0.3136 & 0.0257 & 0.0130 \\
0.18 & 228 & 809 & 3 & 0.9791 & 0.2199 & 0.9870 & 0.3596 & 0.0209 & 0.0130 \\
0.19 & 228 & 670 & 3 & 0.9827 & 0.2539 & 0.9870 & 0.4039 & 0.0173 & 0.0130 \\
0.20 & 228 & 571 & 3 & 0.9852 & 0.2854 & 0.9870 & 0.4427 & 0.0148 & 0.0130 \\
0.21 & 228 & 500 & 3 & 0.9871 & 0.3132 & 0.9870 & 0.4755 & 0.0129 & 0.0130 \\
0.22 & 228 & 437 & 3 & 0.9887 & 0.3429 & 0.9870 & 0.5089 & 0.0113 & 0.0130 \\
0.23 & 228 & 387 & 3 & 0.9900 & 0.3707 & 0.9870 & 0.5390 & 0.0100 & 0.0130 \\
0.24 & 228 & 348 & 3 & 0.9910 & 0.3958 & 0.9870 & 0.5651 & 0.0090 & 0.0130 \\
0.25 & 228 & 310 & 3 & 0.9919 & 0.4238 & 0.9870 & 0.5930 & 0.0080 & 0.0130 \\
0.26 & 228 & 282 & 3 & 0.9927 & 0.4471 & 0.9870 & 0.6154 & 0.0073 & 0.0130 \\
0.27 & 228 & 254 & 3 & 0.9934 & 0.4730 & 0.9870 & 0.6396 & 0.0066 & 0.0130 \\
0.28 & 228 & 233 & 3 & 0.9939 & 0.4946 & 0.9870 & 0.6590 & 0.0060 & 0.0130 \\
0.29 & 228 & 213 & 3 & 0.9944 & 0.5170 & 0.9870 & 0.6786 & 0.0055 & 0.0130 \\
0.30 & 228 & 191 & 3 & 0.9950 & 0.5442 & 0.9870 & 0.7015 & 0.0049 & 0.0130 \\
0.31 & 228 & 168 & 3 & 0.9956 & 0.5758 & 0.9870 & 0.7273 & 0.0043 & 0.0130 \\
0.32 & 228 & 154 & 3 & 0.9960 & 0.5969 & 0.9870 & 0.7439 & 0.0040 & 0.0130 \\
0.33 & 228 & 134 & 3 & 0.9965 & 0.6298 & 0.9870 & 0.7690 & 0.0035 & 0.0130 \\
0.34 & 228 & 117 & 3 & 0.9969 & 0.6609 & 0.9870 & 0.7917 & 0.0030 & 0.0130 \\
0.35 & 228 & 106 & 3 & 0.9972 & 0.6826 & 0.9870 & 0.8071 & 0.0027 & 0.0130 \\
0.36 & 228 & 97 & 3 & 0.9974 & 0.7015 & 0.9870 & 0.8201 & 0.0025 & 0.0130 \\
0.37 & 228 & 91 & 3 & 0.9976 & 0.7147 & 0.9870 & 0.8291 & 0.0024 & 0.0130 \\
0.38 & 228 & 83 & 3 & 0.9978 & 0.7331 & 0.9870 & 0.8413 & 0.0021 & 0.0130 \\
0.39 & 228 & 73 & 3 & 0.9980 & 0.7575 & 0.9870 & 0.8571 & 0.0019 & 0.0130 \\
0.40 & 228 & 63 & 3 & 0.9983 & 0.7835 & 0.9870 & 0.8736 & 0.0016 & 0.0130 \\
\bottomrule
\end{tabular}
\end{table}

\clearpage
\section{ALMA Examples with Differences in Methods}
\label{app:example-signals}

\paragraph{Example with NWKR-SS Advantage.}
The examples shown in Figure \ref{fig:qa2_high_rows} that our method consistently assigns higher scores to the true positive rows than the competing methods, indicating stronger agreement with the underlying target structure. In particular, the NWKR-based approach preserves both localization and scoring, yielding high scores on the rows of interest, whereas the alternative methods more frequently produce lower scores, or misplaced windows. This suggests that our method is more robust to noise and local spectral variation, and therefore more effective at identifying the relevant absorption features.

\begin{figure}[htbp]
  \centering
  \includegraphics[width=\linewidth]{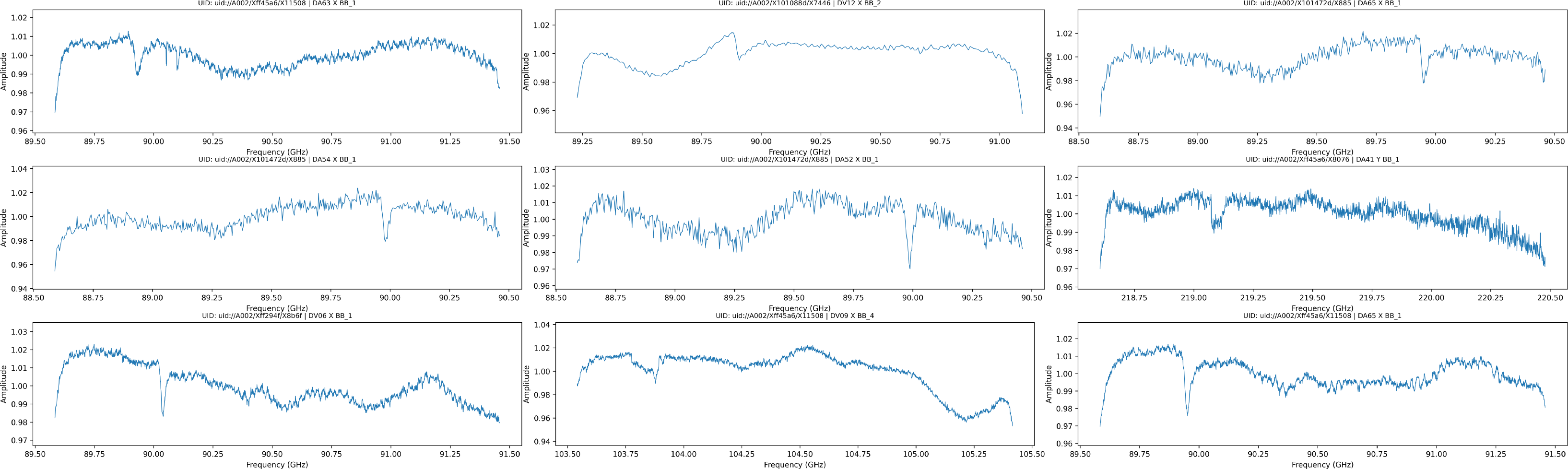}
  \caption{Example spectra where $\F_\KR$ methods outperform other baselines}
  \label{fig:qa2_high_rows}
\end{figure}

\paragraph{Example large and low scoring instances.}
Figures \ref{fig:qa2_high} and \ref{fig:qa2_low} show the contrast between the scores reported. These figures illustrate how NWKR Scan Statistics capture anomalies in real-life datasets.
Note that the low-scoring ones are not platforming anomalies.  Due to random variation, our scanning algorithms still find some most anomalous region, but since it is not that different from the global fit, the returned score is low.  

\begin{figure}[htbp]
  \begin{minipage}[b]{0.48\linewidth}
    \centering
    \begin{subfigure}[b]{\linewidth}
      \centering
      \includegraphics[width=\linewidth]{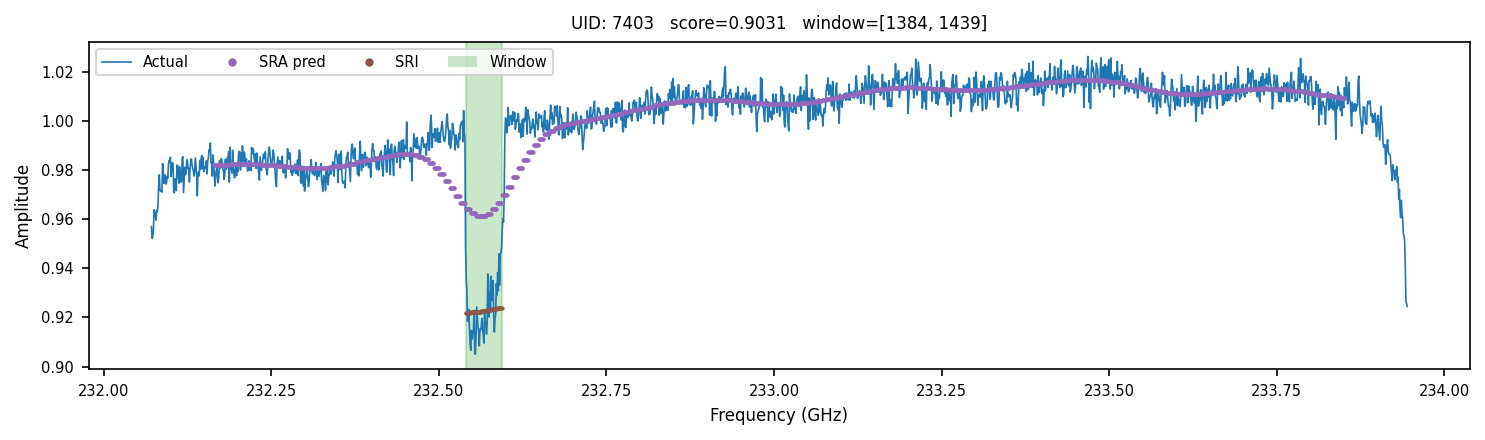}
      \label{fig:high_1}
    \end{subfigure}

    \vspace{0.3em}

    \begin{subfigure}[b]{\linewidth}
      \centering
      \includegraphics[width=\linewidth]{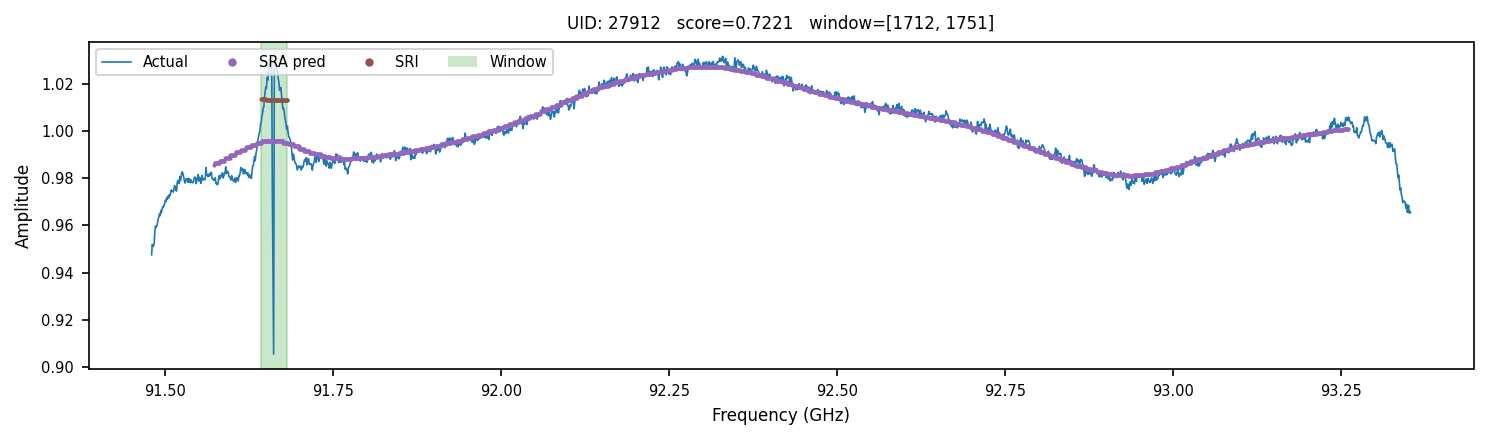}
      \label{fig:high_2}
    \end{subfigure}

  \end{minipage}
  \hfill
  \begin{minipage}[b]{0.48\linewidth}
    \centering
    \begin{subfigure}[b]{\linewidth}
      \centering
      \includegraphics[width=\linewidth]{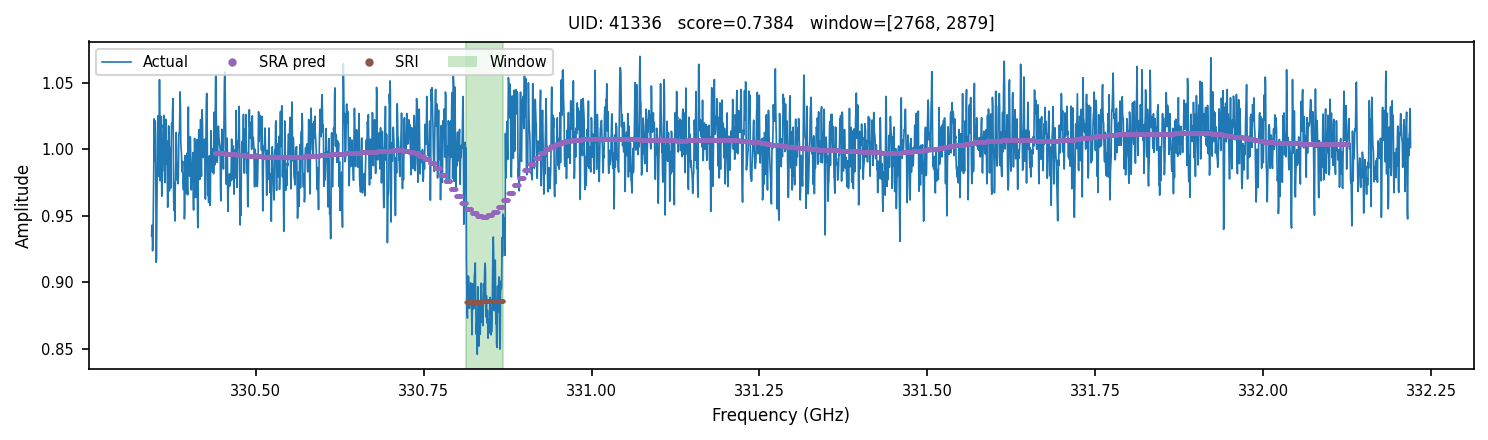}
      \label{fig:high_3}
    \end{subfigure}

    \vspace{0.3em}

    \begin{subfigure}[b]{\linewidth}
      \centering
      \includegraphics[width=\linewidth]{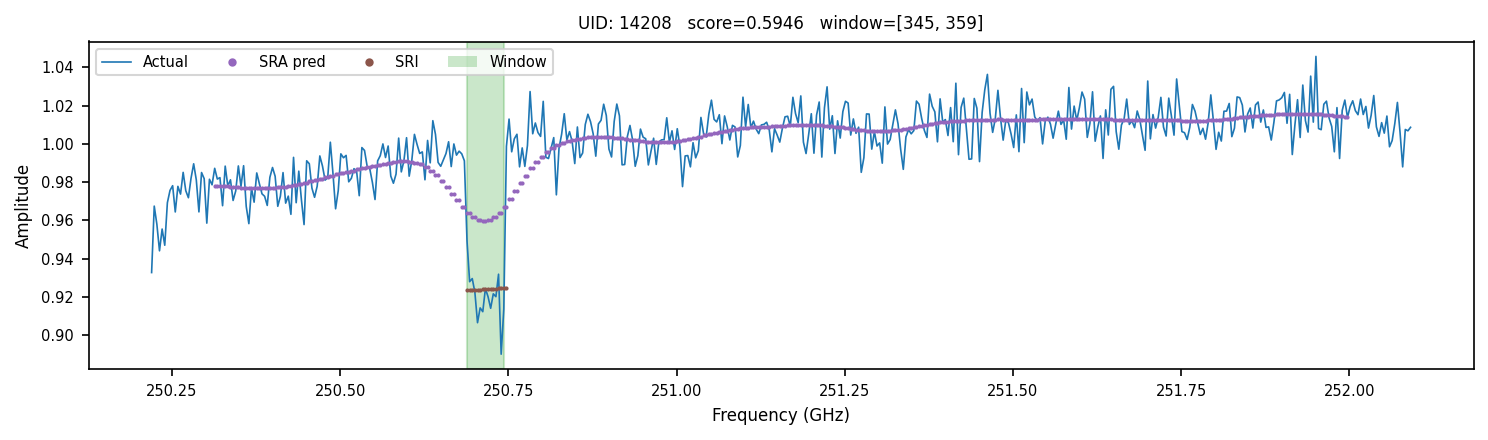}
      \label{fig:high_4}
    \end{subfigure}

  \end{minipage}
  \caption{High-scoring rows produced from ALMA dataset}
  \label{fig:qa2_high}
\end{figure}

\begin{figure}[htbp]
  \begin{minipage}[b]{0.48\linewidth}
    \centering
    \begin{subfigure}[b]{\linewidth}
      \centering
      \includegraphics[width=\linewidth]{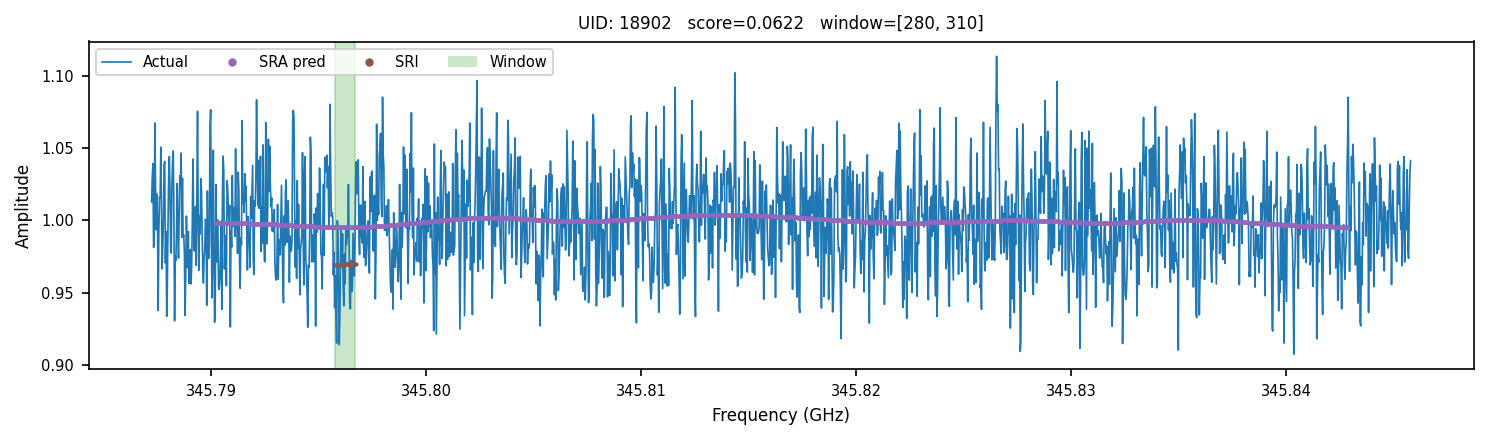}
      \label{fig:low_1}
    \end{subfigure}

    \vspace{0.3em}

    \begin{subfigure}[b]{\linewidth}
      \centering
      \includegraphics[width=\linewidth]{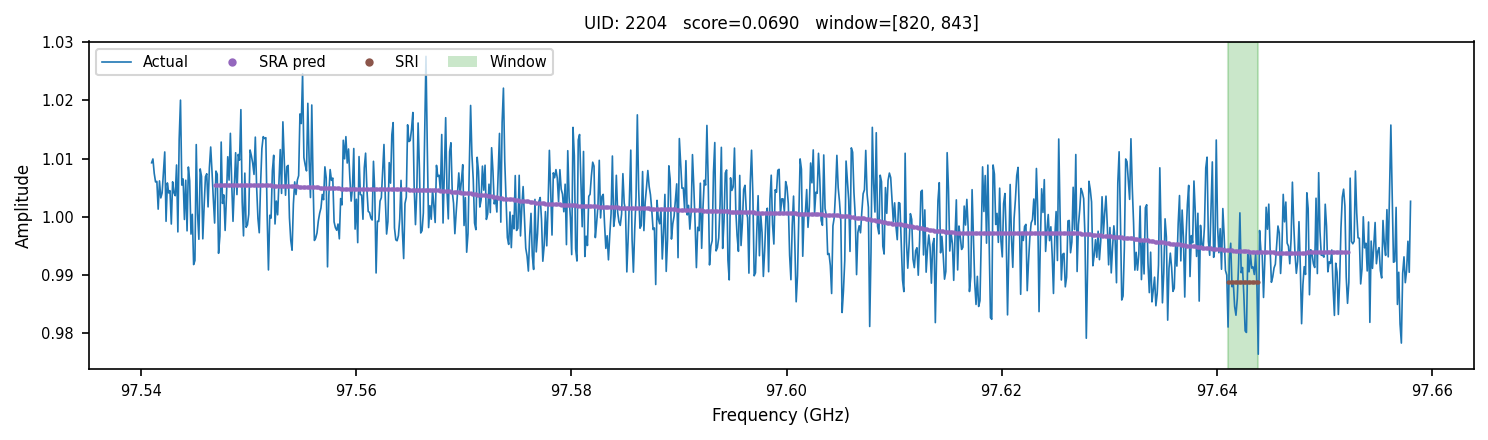}
      \label{fig:low_2}
    \end{subfigure}

  \end{minipage}
  \hfill
  \begin{minipage}[b]{0.48\linewidth}
    \centering
    \begin{subfigure}[b]{\linewidth}
      \centering
      \includegraphics[width=\linewidth]{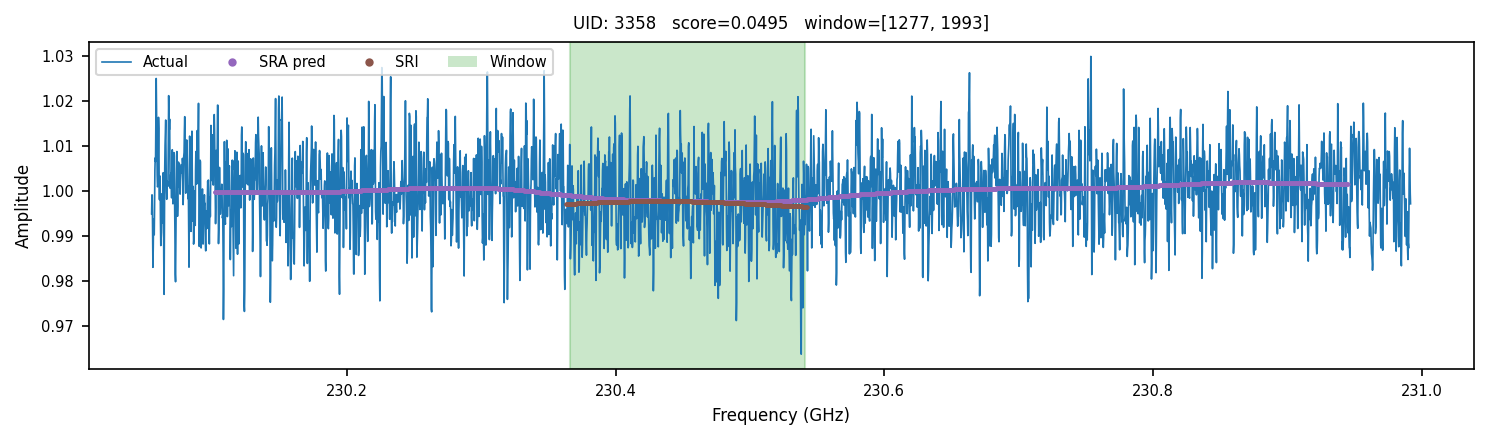}
      \label{fig:low_3}
    \end{subfigure}

    \vspace{0.3em}

    \begin{subfigure}[b]{\linewidth}
      \centering
      \includegraphics[width=\linewidth]{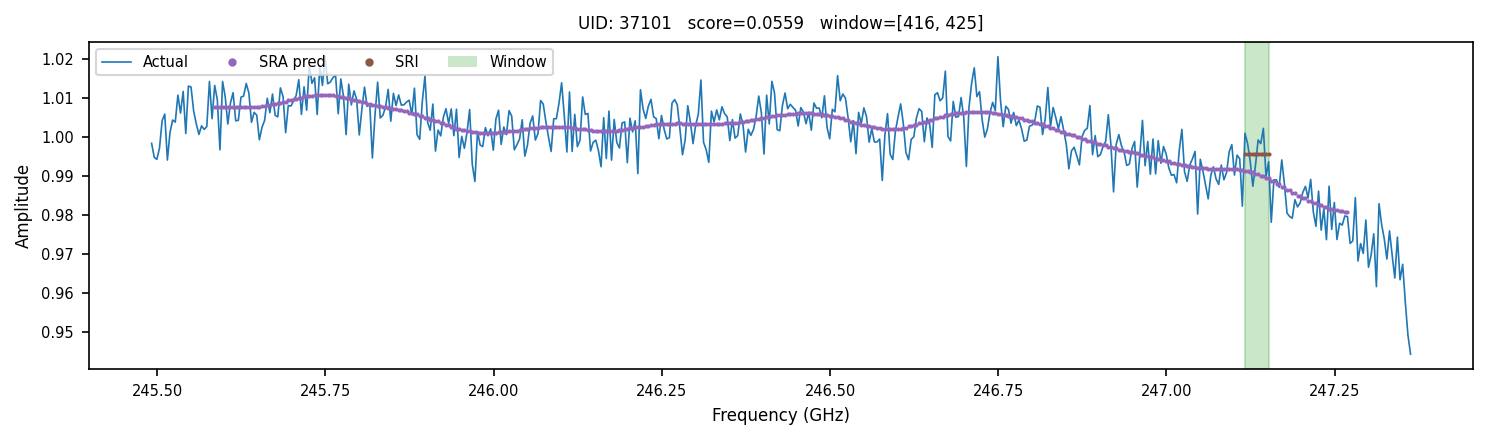}
      \label{fig:low_4}
    \end{subfigure}

  \end{minipage}
  \caption{Low-scoring rows produced from ALMA dataset}
  \label{fig:qa2_low}
\end{figure}

\clearpage
\section{Application:  Solar Radiation}
\label{app:solar}

To further demonstrate the generality of our approach, we apply our regression-based scan statistics framework to surface weather observations. We use \href{https://horel.chpc.utah.edu/uunet_time.html}{solar radiation data} from the default wbb station (William Browning Building), collected by MesoWest~\citep{horel2002mesowest}. Each day of solar radiation data produces a smooth curve as the sun rises, peaks, and sets. Occasionally, shadows interrupt this curve, caused by passing clouds, sensor obstruction, or instrumental glitches, creating interval anomalies.  Prominent and regular ones are likely caused by shadows from polls that depend on the sun-angle and do not reflect true solar radiation effects.  
We extracted individual daytime segments and ran all methods. 

Figures ~\ref{fig:wbb1} to \ref{fig:wbb5} show five representative examples. In all the figures, $\F_\KR$ manages to localize the anomaly almost perfectly, outperforming all the other methods. These results confirm the advantages of NWKR-based scan statistics in robust localization of interval anomalies in smoothly varying non-stationary signals.

\begin{figure}[htbp]
\centering
\includegraphics[width=\columnwidth]{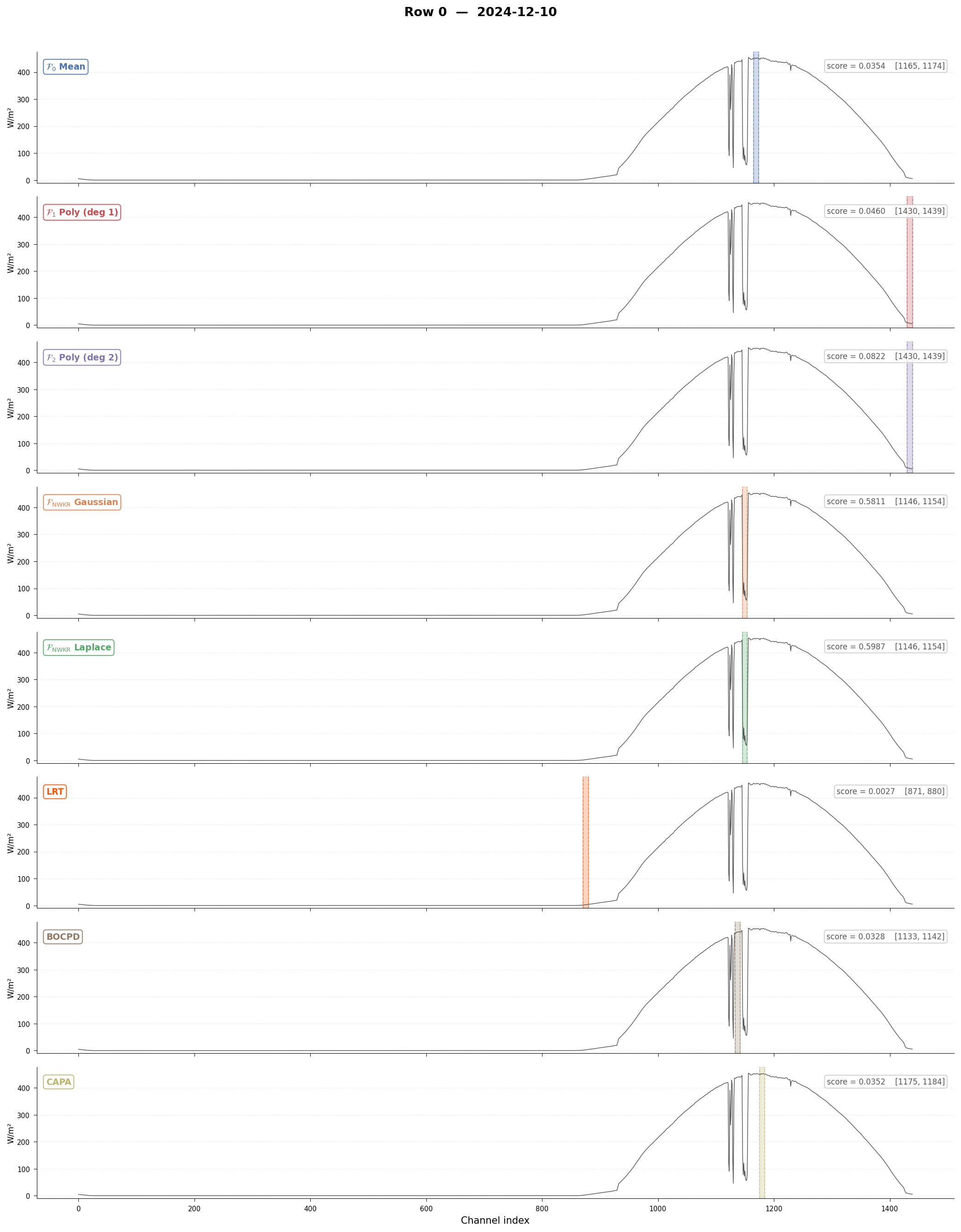}
\caption{Scan statistic results on WBB solar radiation signal from 2024-12-10.}
\label{fig:wbb1}
\end{figure}

\begin{figure}[htbp]
\centering
\includegraphics[width=\columnwidth]{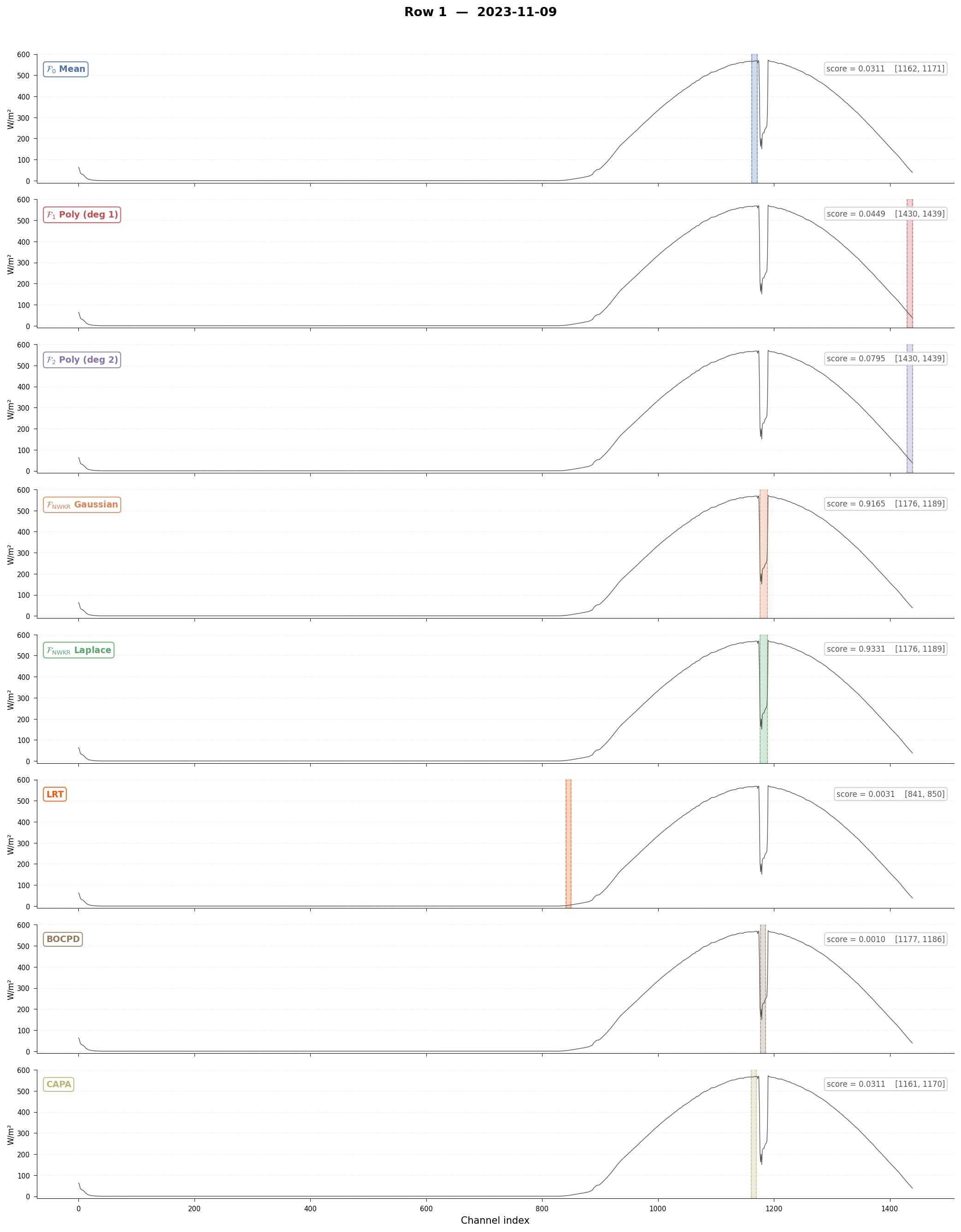}
\caption{WBB signal from 2023-11-09.}
\label{fig:wbb2}
\end{figure}

\begin{figure}[htbp]
\centering
\includegraphics[width=\columnwidth]{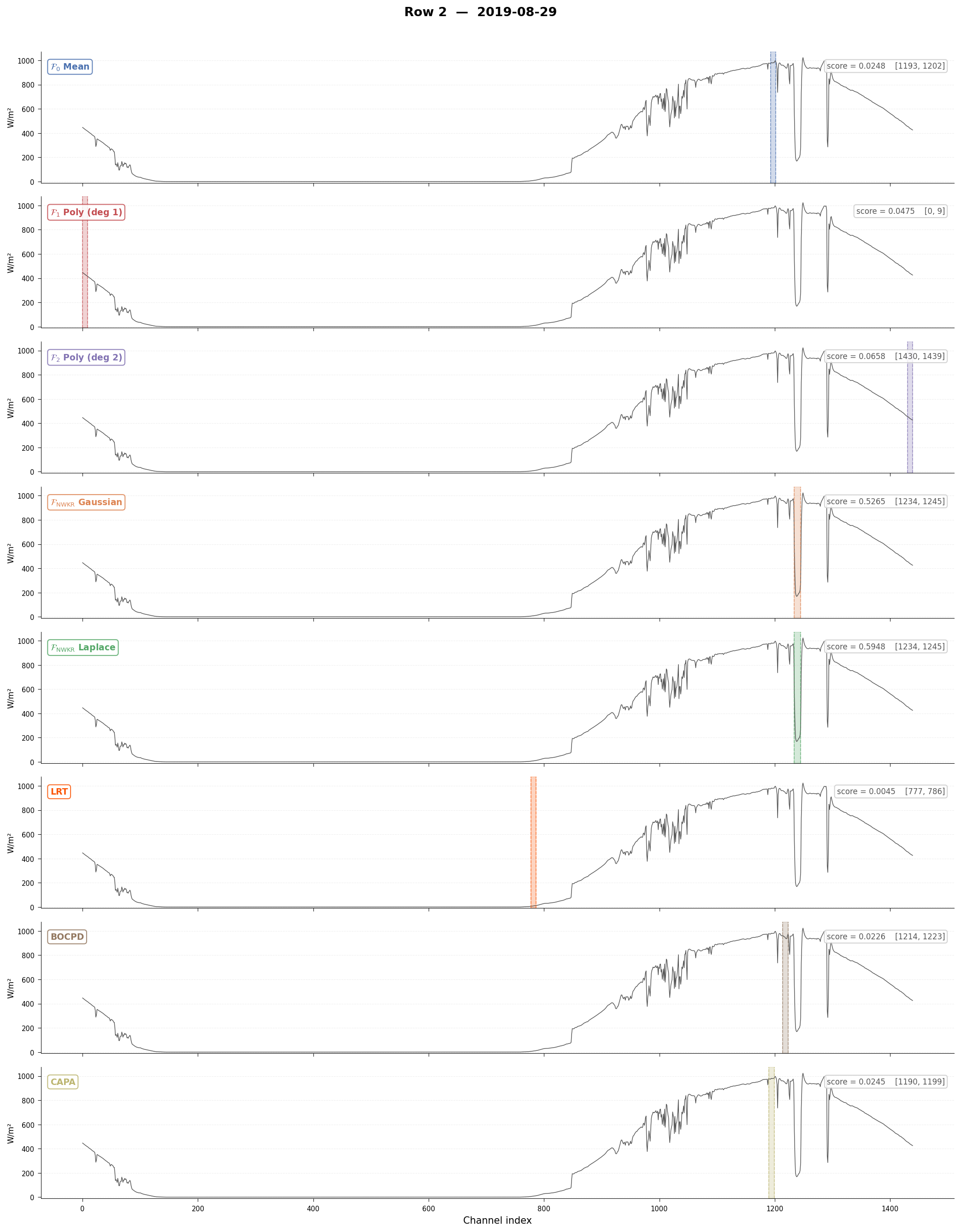}
\caption{WBB signal from 2019-08-29.}
\label{fig:wbb3}
\end{figure}

\begin{figure}[htbp]
\centering
\includegraphics[width=\columnwidth]{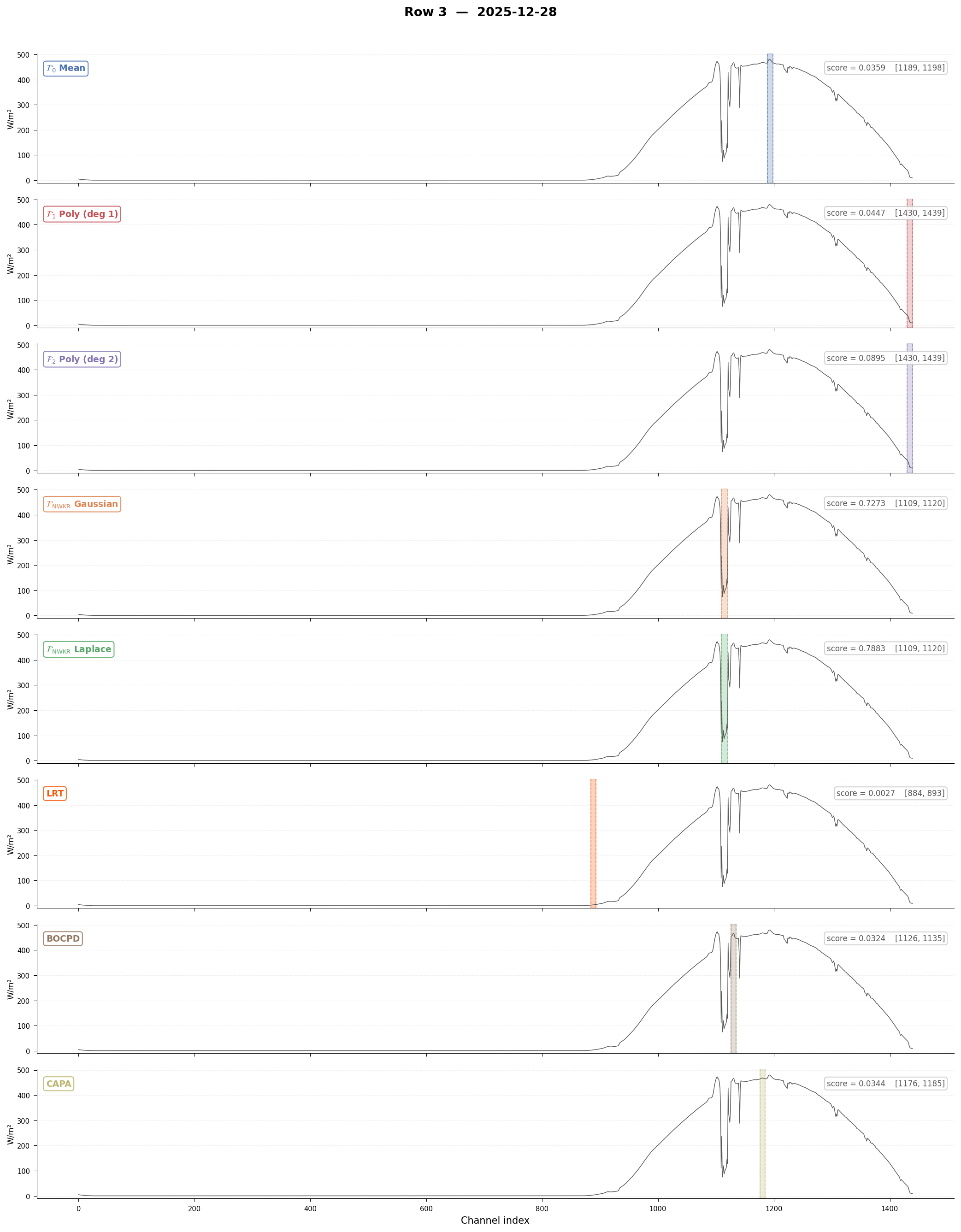}
\caption{WBB signal from 2025-12-28.}
\label{fig:wbb4}
\end{figure}

\begin{figure}[htbp]
\centering
\includegraphics[width=\columnwidth]{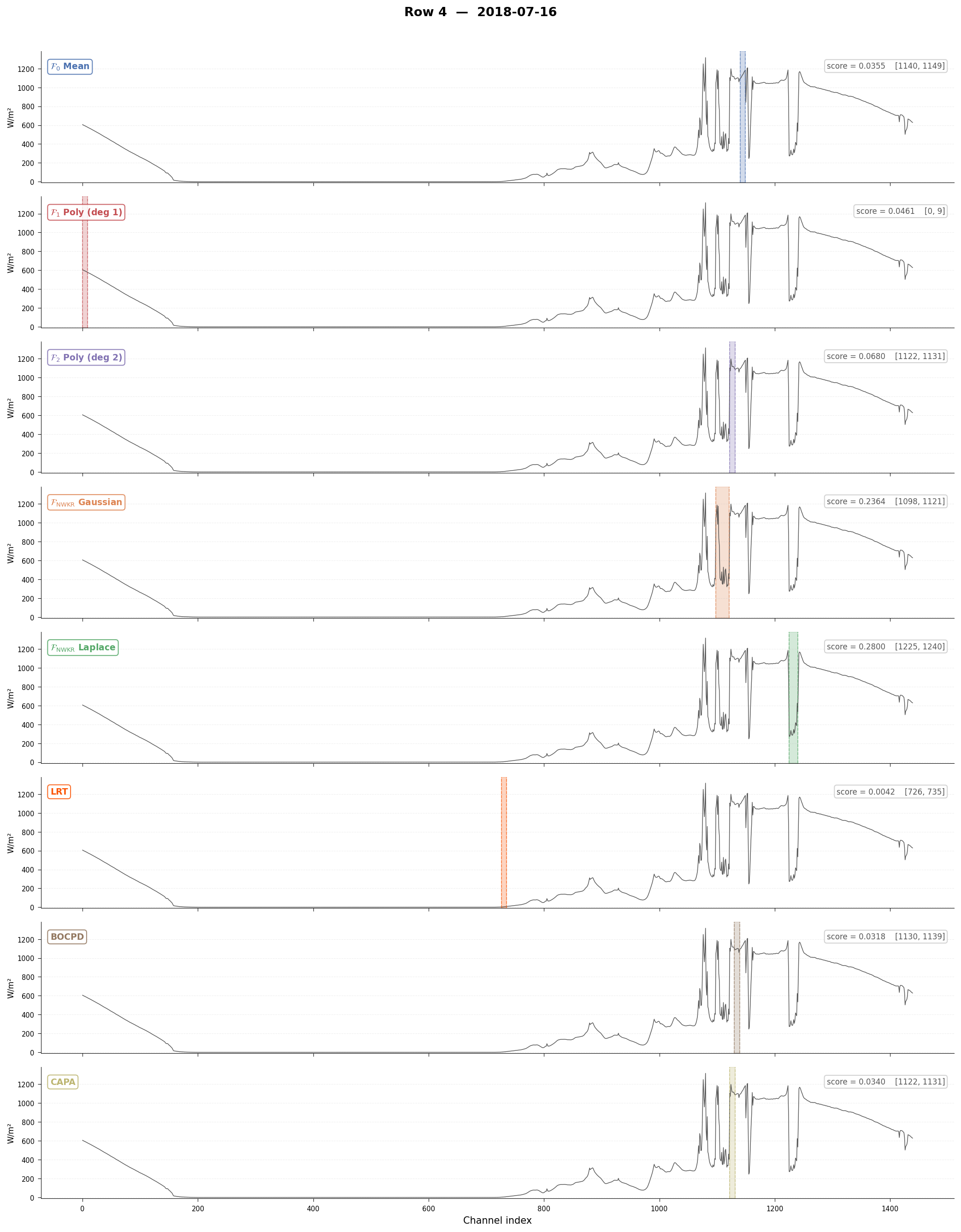}
\caption{WBB signal from 2018-07-16.}
\label{fig:wbb5}
\end{figure}


\clearpage

\end{document}